\documentclass[3p,times]{elsarticle}

\usepackage{ecrc}
\usepackage{amsthm}
\usepackage{amsmath}
\usepackage{amsfonts}
\usepackage{amssymb}

\volume{00}

\firstpage{1}

\runauth{}

\CopyrightLine{2011}{Published by Elsevier Ltd.}

\usepackage{amssymb}

\usepackage[figuresright]{rotating}
\newtheorem{thm}{Theorem}[section]
\newtheorem{co}[thm]{Corollary}
\newtheorem{lem}[thm]{Lemma}
\newdefinition{rmk}{Remark}[section]
\newdefinition{defi}{Definition}[section]
\newproof{pf}{Proof}
\newproof{pot}{Proof of Theorem \ref{thm2}}

\begin{document}

\begin{frontmatter}



\dochead{}

\title{Hamiltonian Paths in $C-$shaped Grid Graphs}


\author{Fatemeh Keshavarz-Kohjerdi and Alireza Bagheri
 }

\address{Department of Computer Engineering \& IT, Amirkabir University of Technology, Tehran, Iran}
\address
{fatemeh.keshavarz@aut.ac.ir
}
\address
{Corresponding author: ar\_bagheri@aut.ac.ir}

\begin{abstract}
We study the Hamiltonian path problem in $C-$shaped grid graphs, and
present the necessary and sufficient conditions for the existence of
a Hamiltonian path between two given vertices in these graphs. We
also give a linear-time algorithm for finding a Hamiltonian path
between two given vertices of a $C-$shaped grid graph, if it exists.
\end{abstract}


\end{frontmatter}


\section{Introduction}\label{IntroSect}
One of the well-known NP-complete problems in graph theory is the Hamiltonian path problem; i.e., finding a simple path in the graph such that every vertex visits exactly once \cite{GJ:CAI}.
The two-dimensional integer grid $G^\infty$ is an infinite
undirected graph in which vertices are all points of the plane with
integer coordinates and two vertices are connected by an edge if and
only if the Euclidean distance between them is equal to $1$. A grid
graph $G_{g}$ is a finite vertex-induced subgraph of the
two-dimensional integer grid $G^\infty$. A solid grid graph is a
grid graph without holes. A rectangular grid graph $R(m,n)$ is the
subgraph of $G^\infty$ (the infinite grid graph) induced by $V(R)=
\{v \ |\ 1 \leq v_{x}\leq m, \ 1\leq v_{y}\leq n\}$, where $v_{x}$
and $v_{y}$ are $x$ and $y$ coordinates of $v$, respectively. A
$C-$shaped grid graph $C(m,n,k,l)$ is a rectangular grid graph
$R(m,n)$ such that a rectangular subgraph $R(k,l)$ is removed from
it while $R(m,n)$ and $R(k,l)$ have exactly one border side in
common, where $k,l\geq 1$ and $m,n>1$ (see Fig. \ref{fig:non0}(c)).
In this paper, we only focus on the results on grid graphs. There
are some results on Hamiltonian path for other classes of graphs
which we do not mention here, see \cite{8, 1} for more details. \par In
\cite{IPS:HPIGG}, Itai \textit{et al.} proved that the Hamiltonian
path problem for general grid graphs, with or without specified
endpoints, is NP-complete. They showed that the problem for
rectangular grid graphs can be solved in linear time. Chen
\textit{et al.} \cite{CST:AFAFCHPIM} gave a parallel algorithm for
the problem in mesh architecture. Lenhart and Umans \cite{LU:HCISGG}
gave a polynomial-time algorithm for finding Hamiltonian cycles in
solid grid graphs. Their algorithm runs in $O(n^{4})$ time. Also,
Salman \cite{s1} introduced a family of grid graphs, that is,
alphabet grid graphs, and determined classes of alphabet grid graphs
that contain Hamiltonian cycles. In \cite{kb:hpiscogg}, the authors
proposed a linear-time algorithm for the Hamiltonian path problem
for some small classes of grid graphs, namely $L-$alphabet,
$C-$alphabet, $E-$alphabet, and $F-$alphabet grid graphs. In
\cite{991}, necessary and sufficient conditions for the existence of
a Hamiltonian path in $L-$shaped grid graphs have been studied.
$L-$alphabet and $C-$alphabet grid graphs considered in
\cite{kb:hpiscogg} are special cases of $L-$shaped and $C-$shaped
grid graphs, respectively. Some other
results about grid graphs are investigated in \cite{3w, vyf:hpotgg, Imnrx:hcihgg, 98,99,E,CT:HPOGG,wqz}.\par In this paper, we obtain necessary
and sufficient conditions for the existence of a Hamiltonian path
between two given vertices in $C-$shaped grid graphs, which are a
special type of solid grid graphs. Also, we show that a Hamiltonian
path in this graph can be found in linear time. Since the
Hamiltonian path problem for solid grid graphs is open, thus solving
the problem for special cases can be considered as the first
attempts to solve the problem in solid grid graphs.  Moreover, this
problem has many applications such as 
\begin{enumerate}
\item In the problem of embedding a graph in a given
grid \cite{SGS:SDORIG}, the first step is to recognize if there are
enough rooms in the host grid for the guest graph. If the guest
graph is a path, then the problem makes relation to the well-known
longest path and Hamiltonian path problems. If we would like to see
if a given solid grid graph has a Hamiltonian path we may reach to
the problem of finding a Hamiltonian path between two given
vertices.
\item In the offline exploration problem \cite{IKKL}, a
mobile robot with limited sensor should visit every cell in a
known cellular room without obstacles in order to explore it and
return to start point such that the number of multiple cell visits is
small. In this problem, let the vertices correspond to the center of
each cell and edges connect adjacent cells, then we have a grid
graph with a given start and end points. Finding a Hamiltonian cycle
in the grid graph corresponds to visiting each cell exactly once
(i.e., a cycle containing all the vertices of the grid graph).
 \item In the
picturesque maze generation problem \cite{kh}, we are given a
rectangular black-and-white raster image and want to randomly
generate a maze in which the solution path fills up the black
pixels. The solution path is a Hamiltonian path of a subgraph
induced by the vertices that correspond to the black cells.
\end{enumerate}
\par The rest of the paper is organized as follows.
 Section 2
gives the preliminaries. Necessary conditions for the existence of a
Hamiltonian path in $C-$shaped grid graphs are given in Section 3.
In Section 4, we show how to obtain a Hamiltonian path for
$C-$shaped grid graphs (sufficient conditions). The conclusion is
given in Section 5.
\section{Preliminaries}
In this section, we quote some definitions and results which we need in the following sections. Some of the definitions are given here are previously defined in \cite{CST:AFAFCHPIM, IPS:HPIGG, 98,99,991}.
\par The \textit{two-dimensional integer grid} is an undirected graph in which vertices are all points of the plane with integer coordinates and two vertices are connected by an edge if and only if the Euclidean distance between them is equal to $1$. For a vertex
$v$ of this graph, let $v_{x}$ and $v_{y}$ denote $x$ and $y$
coordinates of its corresponding point, respectively (sometimes we use $(v_x,v_y)$
instead of $v$). We color the vertices of the two-dimensional
integer grid as black and white. A vertex $\upsilon$ is colored
\textit{white} if $v_{x}+v_{y}$ is even, otherwise it is
colored \textit{black}. \par A \textit{grid graph} $G_{g}$ is a
finite vertex-induced subgraph of the two-dimensional integer grid $G^\infty$.
In a grid graph $G_{g}$, each vertex has degree at most four.
Clearly, there is no edge between any two vertices of the same
color. Therefore, $G_{g}$ is a bipartite graph. Note that any cycle
or path in a bipartite graph alternates between black and white
vertices. Assume $G=(V(G),E(G))$ is a graph with vertex set $V(G)$ and edge set $E(G)$. Assume $v\in V(G)$. The number of edges incident at $v$ in $G$ is called degree of the vertex $v$ in $G$ and is denoted by $degree(v)$.
\par A \textit{rectangular grid graph}, denoted by $R(m,n)$ (or $R$ for
short), is a grid graph whose vertex set is $V(R)= \{v \ |\ 1
\leq v_{x}\leq m, \ 1\leq v_{y}\leq n\}$. The graph $R(9,5)$ is illustrated in Fig.
\ref{fig:non0}(a). 
The size of $R(m,n)$ is defined to be $m\times n$.
$R(m,n)$ is called \textit{odd-sized} if $m\times n$ is odd, otherwise  it is
called \textit{even-sized}. $R(m,n)$ is called an $k-$rectangle if $k=m$ or $n$.
\par A \textit{$L-$shaped grid graph} (resp. \textit{$C-$shaped grid graph}),
denoted by $L(m,n,k,l)$ (resp. $C(m,n,k,l)$) (or $L$ (resp. $C$) for
short), is a rectangular grid graph $R(m,n)$ such that a rectangular
subgraph $R(k,l)$ is removed from it while $R(m,n)$ and $R(k,l)$
have exactly two (resp. one) border side in common, where $k,l\geq
1$ and $m,n>1$. Fig. \ref{fig:non0}(b) and \ref{fig:non0}(c) show
a $L-$shaped grid graph with $m=9$, $n=7$, $k=5$, and $l=2$, and a
$C-$shaped grid graph, with $m=11$, $n=8$, $k=3$, and $l=4$,
respectively. In this paper, we consider $C-$shaped grid graph
$C(m,n,k,l)$ shown in Fig. \ref{fig:non0}(c) with any values of $d$,
$c$, $k$, $l$, $m$, and $n$. Let $G(m,n,k,l)$ be a $L-$shaped or
$C-$shaped grid graph. The size of $G(m,n,k,l)$ is $m\times n-
k\times l$. $G(m,n,k,l)$ is called \textit{even-sized} if $m\times
n- k\times l$ is even, otherwise it is called \textit{odd-sized}.
\par We will refer to a grid graph $G(m,n)$ with two specified
distinct vertices $s$ and $t$ as $(G(m,n,s,t)$. We say that
$G(m,n,s,t)$ is Hamiltonian if there is a Hamiltonian path between
$s$ and $t$ in $G$.  In the following by Hamiltonian $(s,t)-$path we
mean a Hamiltonian path between $s$ and $t$. Throughout this paper
in the figures, $(1,1)$ is the coordinates of the vertex in the
upper left corner, except we explicitly change this assumption.
Without loss of generality, we assume that $s_{x} \leq t_{x}$.
\begin{defi}\label{Definition:d1} Suppose that $G (V_1\cup V_2, E)$ is
a bipartite graph such that $|V_1|\geq |V_2|$ and the vertices of
$G$ colored by two colors, black and white. All the vertices of
$V_1$ will be colored by one color, the majority color, and the
vertices of $V_2$ by the minority color. The Hamilton path problem
$(G, s, t)$ is \textit{color-compatible} if
\begin{enumerate}
\item $s$ and $t$ have different colors and $G$ is even-sized ($|V_1|=|V_2|$), or
\item $s$ and $t$ have the majority color and $G$ is odd-sized ($|V_1|=|V_2|+1$)
\end{enumerate}
\end{defi}
\begin{figure}[tb]
  \centering
  \includegraphics[scale=1]{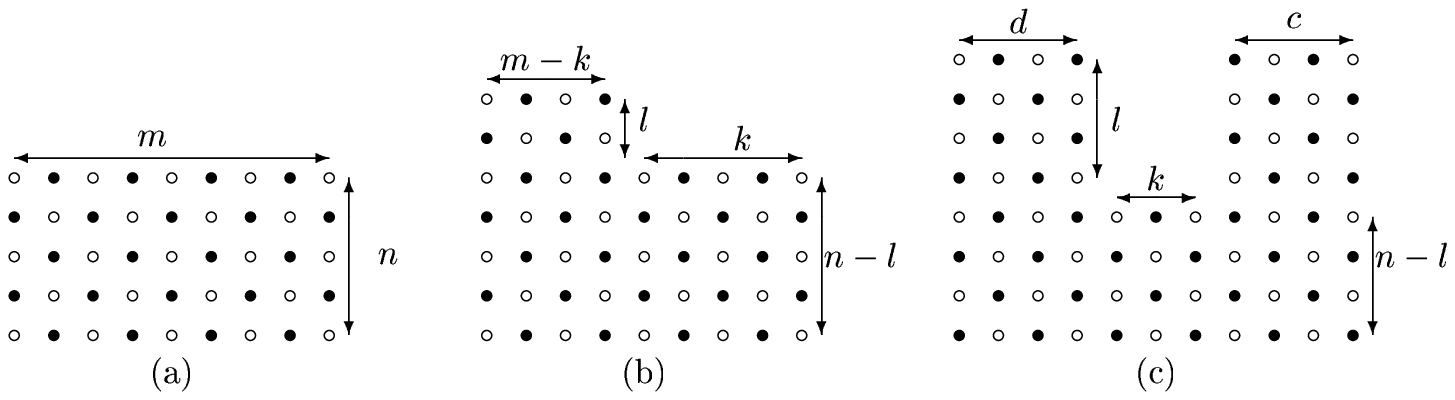}
  \caption[]%
  {\small (a) $R(9,5)$, (b) $L(9,7,5,2)$, and (c) $C(11,8,3,4)$}.
\label{fig:non0}
\end{figure}
\begin{defi}
Let $G$ be a connected graph and $V_1$ be a subset of the vertex set
$V(G)$. $V_1$ is a \textit{vertex cut} of $G$ if $G-V_1$ is
disconnected. A vertex $v$ of $G$ is a \textit{cut vertex} of $G$ if
$\{v\}$ is a vertex cut of $G$. For an example, in Fig.
\ref{RecFig}(a) $t$ is a cut vertex and in Fig. \ref{RecFig}(b)
$\{s,t\}$ is a vertex cut.
 \end{defi}
\par In an odd-sized grid graph
the number of vertices with the minority color is one less than the number of
vertices with the majority color. Therefore, the two end-vertices of any Hamiltonian path in
such a graph must have the majority color. Similarly, in an even-sized grid graph
the number of black vertices is equal to the number of white vertices. Thus, the two
end-vertices of any Hamiltonian path in the graph must have
different colors. Hence, we conclude that the color-compatibility of $s$
and $t$ is a necessary condition for a grid graph to be
Hamiltonian. \\
 Additionally, Itai \textit{et al.} \cite{IPS:HPIGG}
showed that if one of the following conditions holds, then
$(R(m,n),s,t)$ is not Hamiltonian:
\begin{itemize}
\item [(F1)] $s$ or $t$ is a cut vertex or $\{s,t\}$ is a vertex cut
(Fig. \ref{RecFig}(a) and \ref{RecFig}(b)). Notice that, here,
$s$ or $t$ is a cut vertex if $R(m,n)$ is a 1-rectangle and
either $s$ or $t$ is not a corner vertex, and $\{s,t\}$ is a vertex
cut if $R(m,n)$ is a 2-rectangle and $[(2\leq s_x=t_x\leq m-1$ and
$n=2)$ or $(2\leq s_y=t_y\leq n-1$ and $m=2)]$.
\item [(F2)] All the cases that are 
isomorphic to the following cases:
\begin{enumerate}
\item $m$ is even, $n=3$, 
\item $s$ is black, $t$ is white,
\item  $s_{y}=2$ and $s_{x}<t_{x}$ (Fig. \ref{RecFig}(c)) or
$s_{y}\neq 2$ and $s_{x}< t_{x}-1$ (Fig. \ref{RecFig}(d)).
\end{enumerate}
\end{itemize}
\begin{figure}[tb]
  \centering
  \includegraphics[scale=1]{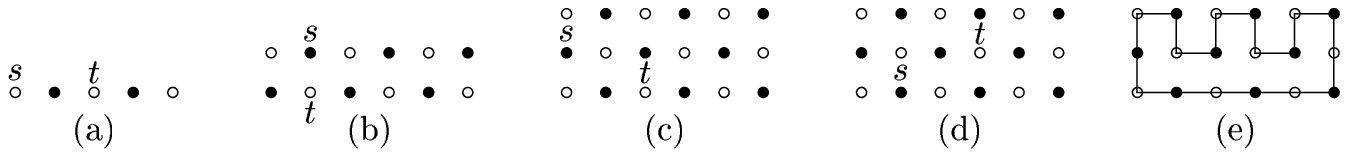}
  \caption[]%
  {\small The rectangular grid graphs in which there is no Hamiltonian $(s,t)-$path, and a Hamiltonian cycle in $R(6, 3)$.}
  \label{RecFig}
\end{figure}
\begin{defi}
 \cite{IPS:HPIGG} A rectangular Hamiltonian path problem $(R(m,n), s, t)$ is \textit{acceptable} if it is color-compatible and $(R(m,n),s,t)$
does not satisfy any of conditions (F1) and (F2).
\end{defi}
\begin{thm} \label{Theorem:1a}\cite{IPS:HPIGG}
There exists a Hamiltonian path between $s$ and $t$ in $R(m,n)$ if and only if $(R(m,n), s, t)$ is acceptable.
\end{thm}
\begin{lem}
\label{Lemma:1m} \cite{CST:AFAFCHPIM} $R(m,n)$ has a Hamiltonian
cycle if and only if it is even-sized and $m,n>1$.
\end{lem}
\par Fig. \ref{RecFig}(e) shows a Hamiltonian cycle for an even-sized rectangular grid
graph, according to Lemma \ref{Lemma:1m}. Every Hamiltonian cycle according to this pattern
contains all the boundary edges on the three sides of the rectangular grid
graph. This means that for an even-sized rectangular grid graph $R$, we can
always find a Hamiltonian cycle, such that it contains all the boundary edges,
except of exactly one side of $R$ which contains an even number of vertices. We need this result in the following.
\begin{defi}\label{Definition:d2}\cite{991}
 A \textit{separation} of a $L-$shaped grid graph $L(m,n,k,l)$ is a partition of $L(m,n,k,l)$
into two disjoint grid subgraphs $G_1$ and $G_2$, i.e., $V
(L(m,n,k,l)) = V (G_1) \cup V (G_2)$, and $V (G_1) \cap V (G_2)=
\emptyset$. $G_1$ and $G_2$ may be rectangular or $L-$shaped grid
graphs. Three types of separations, vertical, horizontal and
$L-$shaped separations are shown in Fig. \ref{fig:non2}(c)-(f).
\end{defi}
\par
In \cite{991}, we show that in addition to condition (F1) (as shown
in Fig. \ref{fig:non1}(a) and \ref{fig:non1}(b)) whenever one of the
following conditions is satisfied then $(L(m,n,k,l),s,t)$ has no
Hamiltonian $(s,t)-$path.
\begin{figure}[tb]
  \centering
  \includegraphics[scale=.95]{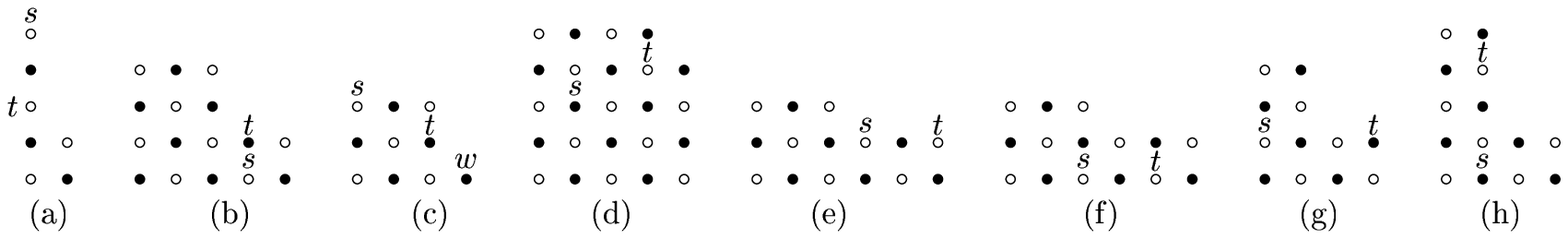}
  \caption[]%
  {\small Some $L-$shaped grid graphs in which there is no Hamiltonian $(s,t)-$path.}
\label{fig:non1}
\end{figure}
\begin{itemize}
\item [(F3)] $w\in V(L(m,n,k,l))$, $degree(w)=1$, $t\neq w$, and $s\neq w$ (Fig. \ref{fig:non1}(c)).
\item [(F4)] $L(m,n,1,1)$ is even-sized, $m-1=even> 2$, $n-1=even>2$, $s=(m-1,2)$, and $t\neq (m-1,1)$
or $t\neq (m,2)$ $($here the role of $s$ and $t$ can be swapped;
i.e., $t=(m-1,2)$ and $s\neq (m-1,1))$ (Fig. \ref{fig:non1}(d)).
\item [(F5)] $L(m,n,k,l)$ is odd-sized, $n-l=2$, $m-k=odd\geq 3$, and
 \\ \indent \indent \indent (i) $s_x,t_x>m-k$ (Fig. \ref{fig:non1}(e)); or
 \\ \indent \indent \indent (ii) $s=(m-k,n)$ and $t_x>m-k$ (Fig. \ref{fig:non1}(f)).
 \item [(F6)] $L(m,n,k,l)$ is even-sized, $n-l=2$, $m-k=2$, and
 \par (i) $s=(1,n-l)$ and $t_x>2$ (Fig. \ref{fig:non1}(g)); or
\par (ii) $s=(2,n)$ and $t_y<l$ $($here the role of $s$ and $t$ can be swapped; i.e., $t=(2,n)$ and $s_y\leq l)$ (Fig. \ref{fig:non1}(h)).
\item [(F7)] $L(m,n,k,l)$ is even-sized and
\par (i)  $n=3$, $l=1$, $m-k=even>2$, $s=(m-k-1,1)$, and $t=(m-k,3)$ (Fig. \ref{fig:non2}(a)); or
\par (ii) $m=3$, $k=1$, and $n-l=even>2$, $s=(1,l+1)$, and $t=(m,l+2)$ (Fig. \ref{fig:non2}(b)).
 \item [(F8)] $L(m,n,k,l)$ is even-sized and $[(m-k=2$ and $n-l> 2)$ or $(n-l=2$ and $m-k> 2)]$.
 Let $\{G_1, G_2\}$ be a vertical $($or horizontal$)$ separation of $L(m,n,k,l)$ such that  $G_1$
 is a 3-rectangle grid graph, $G_2$ is a 2-rectangle grid graph, and exactly two vertices $u$ and
$v$ are in $G_1$ that are connected to $G_2$. Let $s^{'}=s$ and
$t^{'}=t$, if $s^{'}$ (or $t^{'})\notin G_1$ then $s^{'}=u$ (or
$t^{'}=u)$. And $(G_1,s^{'},t^{'})$ satisfies condition (F2) (Fig.
\ref{fig:non2}(c) and \ref{fig:non2}(d)).
\item [(F9)] $L(m,n,k,l)$ is even-sized and $[(m-k=3$ and $n-l\geq 3)$ or $(m-k>3$ and $n-l=3)]$.
Let $\{G_1,G_2\}$ be a vertical $($ horizontal or $L-$shaped$)$
separation of $L(m,n,k,l)$ such that $G_1$ and $G_2$ are even-sized,
$G_1$ is a $3-$rectangle grid graph, and $G_2$ is
\par (1) a rectangular grid graph $($see Fig. \ref{fig:non2}(e), or
\par (2) a $L-$shaped grid graph, where $m\times n=$even$\times$odd, $k\times l=$odd$\times$even,
$n-l=3$, and $m-k\geq 5$. Here, $V(G_1)=\{ m-k\leq x\leq m \ and\  l+1\leq y\leq n  \}$ and
$G_2= L(m,n,k,l)\backslash G_1$ $($see Fig. \ref{fig:non2}(f)).\\
Let exactly three vertices $v$, $w$ and $u$ be in $G_1$ that are connected
to $G_2$. Let $s^{'}=s$ and $t^{'}=t$, if $s^{'}$ (or $t^{'})\notin
G_1$ then $s^{'}=w$ (or $t^{'}=w)$. And $(G_1,s^{'},t^{'})$
satisfies condition (F2).
\end{itemize}
\begin{figure}[tb]
  \centering
  \includegraphics[scale=1]{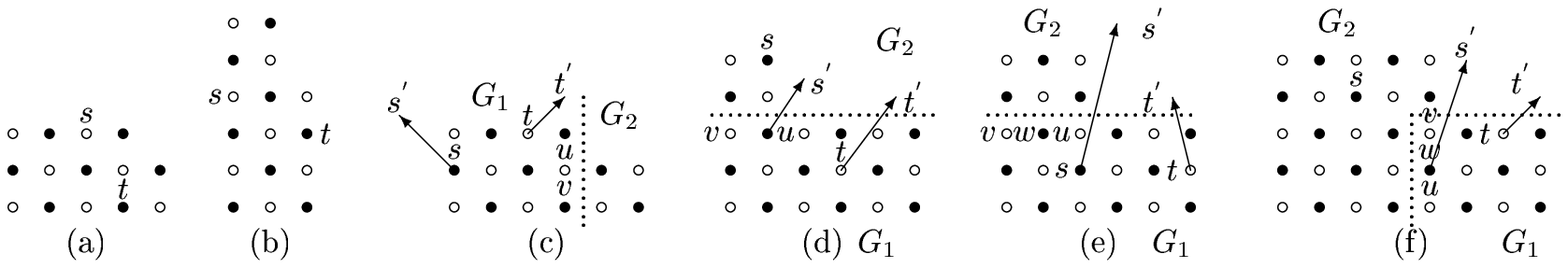}
  \caption[]%
 {\small Some $L-$shaped grid graphs in which there is no Hamiltonian $(s,t)-$path, where dotted lines indicate the separations.}
\label{fig:non2}
\end{figure}
\begin{defi}
A $L-$shaped Hamiltonian path problem $(L(m,n,k,l), s, t)$ is \textit{acceptable} if it is color compatible and $(L(m,n,k,l),s,t)$
does not satisfy any of conditions (F1) and (F3)-(F9).
\end{defi}
\begin{thm}\label{Theorem:6t}\cite{991}
$L(m,n,k,l)$ has a Hamiltonian $(s,t)-$path if and only if $(L(m,n,k,l),s,t)$ is acceptable.
\end{thm}
\begin{thm}\label{Theorem:6t2}\cite{991}
In an acceptable $P(L(m,n,k,l),s,t)$, a Hamiltonian $(s,t)-$path can be found in linear time.
\end{thm}
\begin{lem}\label{Lemma:6t1}\cite{991}
$L(m,n,k,l)$ has a Hamiltonian cycle if and only if it is even-sized, $m-k>1$, and $n-l>1$.
\end{lem}
\section{Necessary conditions}
In this section, we are going to obtain necessary conditions for the existence of a Hamiltonian
$(s,t)-$path in $C-$shaped grid graph $C(m,n,k,l)$.
\begin{defi}\label{Definition:d3}
 A \textit{separation} of a $C-$shaped grid graph $C(m,n,k,l)$ is a partition of $C(m,n,k,l)$
into at most five disjoint grid subgraphs $G_1$, $G_2$, $G_3$,
$G_4$, and $G_5$ that is, $V (C(m,n,k,l)) = V(G_1) \cup V(G_2)\cup
V(G_3)\cup V(G_4)\cup V(G_5)$, and $V(G_1) \cap V(G_2)\cap
V(G_3)\cap V(G_4)\cap V(G_5)= \emptyset$. $G_1$, $G_2$, $G_3$,
$G_4$, and $G_5$ may be rectangular, $L-$shaped, or $C-$shaped grid graph. We consider the four types of separation,
vertical, horizontal, $L-$shaped, and $C-$shaped separations are
shown in Fig. \ref{sep1} and \ref{sep2}.
\end{defi}
\begin{figure}[tb]
  \centering
  \includegraphics[scale=1]{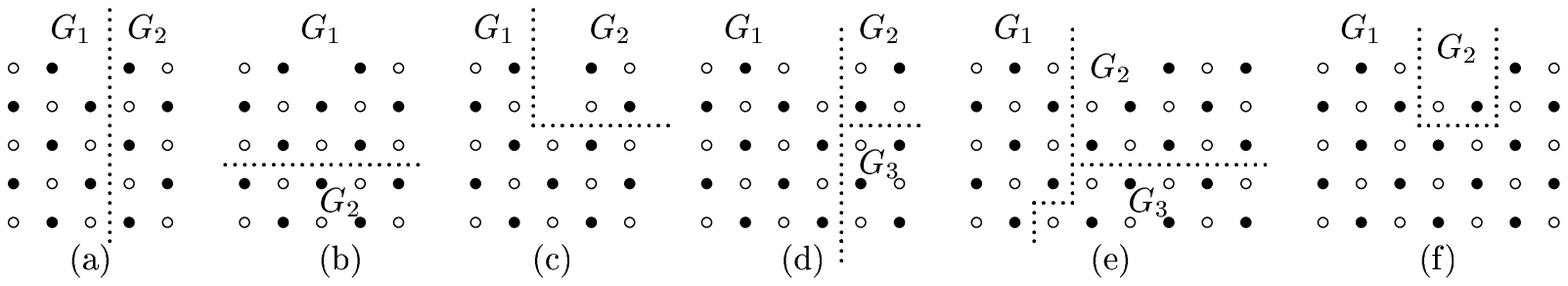}
  \caption[]%
  {\small The four types of separations; (a) a vertical separation (b) a horizontal separation, (c)
  a $L-$shaped separation type I, (d) a $L-$shaped separation type II, (e) a $L-$shaped separation type III, and (f) a $C-$shaped separation type I, where dotted lines indicate the separations.}
  \label{sep1}
\end{figure}
\begin{figure}[tb]
  \centering
  \includegraphics[scale=.94]{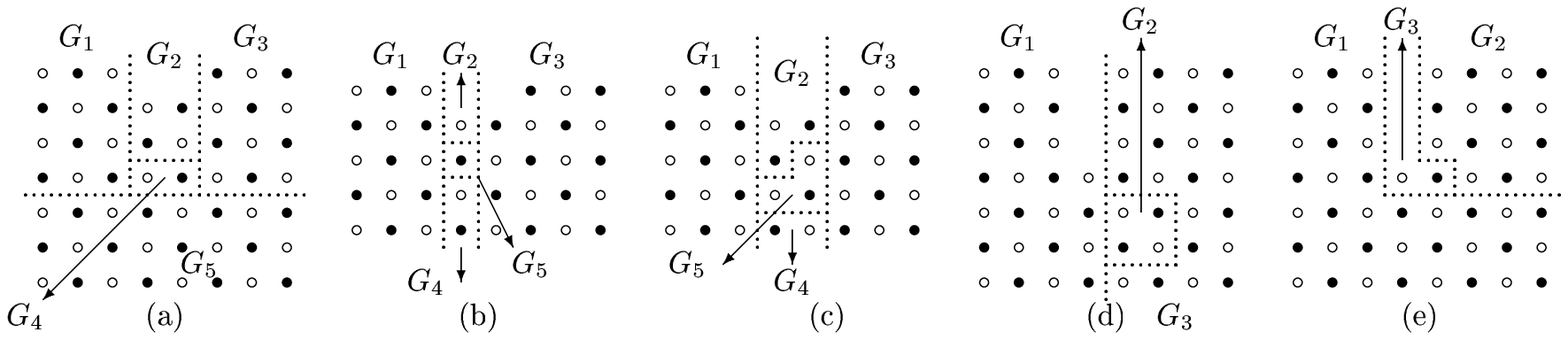}
  \caption[]%
  {\small The four types of separations;
  (a) a $C-$shaped separation type II, (b) and (c) a $C-$shaped separation type III,
  (d) a $C-$shaped separation type IV, and (e) a $C-$shaped separation type V, where dotted lines indicate the separations.}
  \label{sep2}
\end{figure}
\begin{lem}\label{Lemma:c0}\cite{991}
Let $G$ be any grid graph. Let $s$ and $t$ be two given vertices of
$G$ such that $(G,s,t)$ is color-compatible. If we can partition
$(G,s,t)$ into $n$ subgraphs $G_1,G_2,\ldots, G_{n-1}, G_n$ such
that $s,t\in G_n$ and in $V(G_1\cup G_2\cup \ldots\cup G_{n-1})$ the
number of white and black vertices are equal, then $(G_n,s,t)$ is
color-compatible.
\end{lem}
Because $C(m,n,k,l)$ is bipartite, colors of vertices of any
path must alternate between black and white. Hence, the
color-compatibility of $s$ and $t$ in $C(m,n,k,l)$ is a necessary
condition for $(C(m,n,k,l),s,t)$ to be Hamiltonian. Besides, in
addition to conditions (F1) and (F3) (as shown in Fig.
\ref{fig:s1}(a)-(d)) whenever one of the following conditions holds
then $(C(m,n,k,l),s,t)$ has no Hamiltonian $(s,t)-$path.
\begin{figure}[tb]
  \centering
  \includegraphics[scale=1]{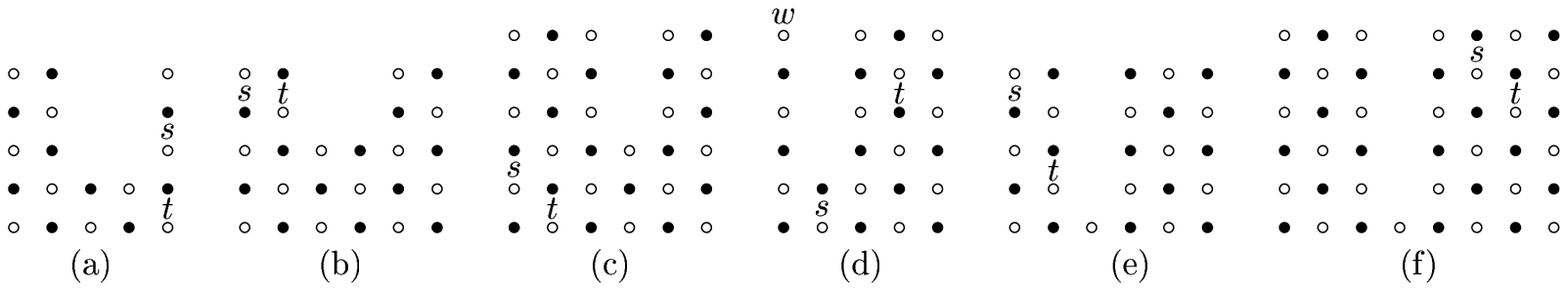}
  \caption[]%
  {\small The $C-$shaped grid graphs in which there is no Hamiltonian $(s,t)-$path.}
\label{fig:s1}
\end{figure}
\begin{figure}[tb]
  \centering
  \includegraphics[scale=.96]{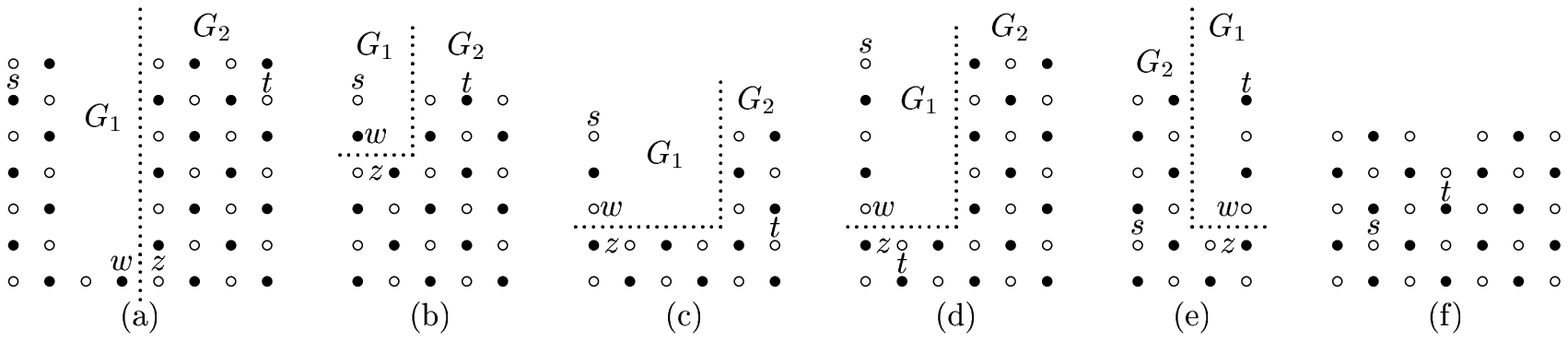}
  \caption[]%
  {\small The $C-$shaped grid graphs in which there is no Hamiltonian $(s,t)-$path.}
\label{fig:s1c}
\end{figure}
\begin{itemize}
\item [(F10)] $n-l=1$, $c,d>1$, and
\begin{itemize}
\item [(i)] $s_x,t_x\leq d$ or $s_x,t_x>d+k$ (Fig. \ref{fig:s1}(e) and \ref{fig:s1}(f)); or
\item [(ii)] Let $C(m,n,k,l)$ be even-sized, and let $\{G_1,G_2\}$ be a vertical separation of
$C(m,n,k,l)$ such that $G_1=L(m^{'},n,k,l)$, where $m^{'}=d+k$,
$G_2=R(m-m^{'},n)$ (as shown Fig. \ref{fig:s1c}(a)), and exactly a
vertex $w$ is in $G_1$ that is connected to $G_2$. Let $z\in G_2$ such
that $w$ and $z$ are adjacent.  And $s\in G_1$, $t\in G_2$, and
$(G_1,s,w)$ or $(G_2,z,t)$ is not acceptable (Fig.
\ref{fig:s1c}(a)).
\end{itemize}
\item [(F11)] $n-l>1$ and $[(d=1$, $c>1$, and $s=(1,1))$ or $(d>1$, $c=1$, and $t=(m,1))]$.
 Let $\{G_1,G_2\}$ be $L-$shaped separation $($ type I$)$ of $C(m,n,k,l)$ such that
 $G_1=R(m^{'},l)$, $G_2=L(m,n,k^{'},l)$, $m^{'}=d$ if $d=1$; otherwise $m^{'}=c$, and
 $k^{'}=k+m^{'}$. Let exactly a vertex $w$ be in $G_1$ that is connected to $G_2$ and let $z\in G_2$ such that $w$ and $z$ are adjacent.
 And one of the following cases occurs:
\begin{itemize}
\item [(i)] $d=1$, $t\in G_2$, and $(G_2,z,t)$ is not acceptable (Fig. \ref{fig:s1c}(b)-(d)); or
\item [(ii)] $c=1$, $s\in G_2$, and $(G_2,s,z)$ is not acceptable(Fig. \ref{fig:s1c}(e)).
\end{itemize}
\item [(F12)] $R(m,n)$ is odd$\times$odd with white majority color and $R(k,l)$ is odd$\times$odd with black majority color (Fig. \ref{fig:s1c}(f)).
\item [(F13)] $n=odd$, $n-l=2$, and $d,c>1$. Let $\{G_1,G_2\}$
be a vertical separation of $C(m,n,k,l)$ such that
$G_1=L(m^{'},n,k,l)$, $G_2=R(m-m^{'},n)$, and $m^{'}=d+k$ (or $G_1=L(m-m^{'},n,k,l)$, $G_2=R(m^{'},n)$, and $m^{'}=d)$. Let exactly two vertices
$u$ and $v$ be in $G_1$ that are connected to $G_2$. And one of the
following cases occurs
\begin{itemize}
\item [(a)] $C(m,n,k,l)$ is odd-sized and
\begin{itemize}
\item [(a$_1$)] $G_1$ is odd-sized, $G_2$ is even-sized, and
\begin{itemize}
\item [(a$_{11}$)] $s,t\in G_2$ (see Fig. \ref{fig:s2}(a)); or
\item [(a$_{12}$)] $s\in G_1$, $t\in G_2$, $s^{'}=s$, $t^{'}=u$
(or $t\in G_1$, $s\in G_2, t^{'}=t$, $s^{'}=u)$, and
$(G_1,s^{'},t^{'})$ satisfies condition (F1) ( i.e.,
$\{s^{'},t^{'}\}$ is a vertex cut$)$ (see Fig. \ref{fig:s2}(b)).
\end{itemize}
 \begin{figure}[tb]
  \centering
  \includegraphics[scale=1]{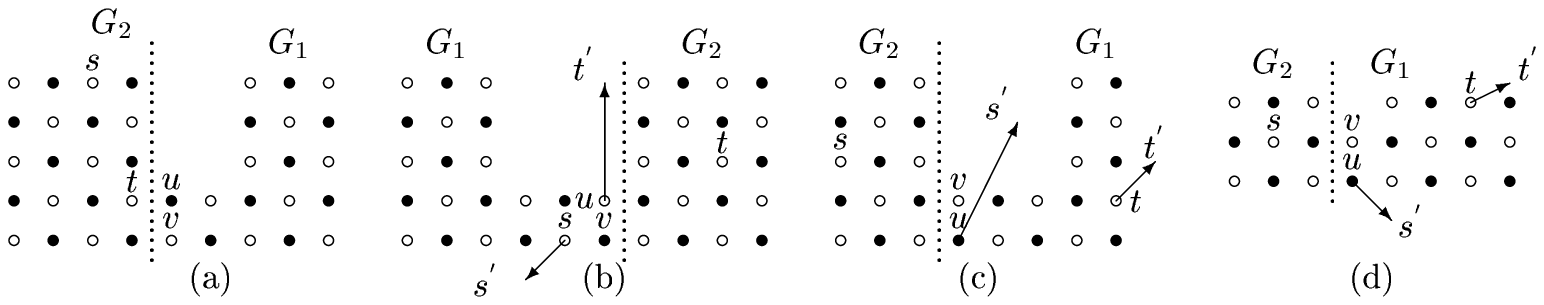}
  \caption[]%
  {\small The $C-$shaped grid graphs in which there is no Hamiltonian $(s,t)-$path.}
\label{fig:s2}
\end{figure}
\item [(a$_2$)] $m$ is even, $G_1$ is even-sized, $G_2$ is odd-sized,
$s\in G_1$, $t\in G_2$, $s^{'}=s$, $t^{'}=u$ $($or $t\in G_1$, $s\in
G_2, t^{'}=t$, $s^{'}=u)$, and $(G_1,s^{'},t^{'})$ satisfies
condition (F6) or (F8) (see Fig. \ref{fig:s2}(c) and
\ref{fig:s2}(d)).
\end{itemize}
\begin{figure}[tb]
  \centering
  \includegraphics[scale=1]{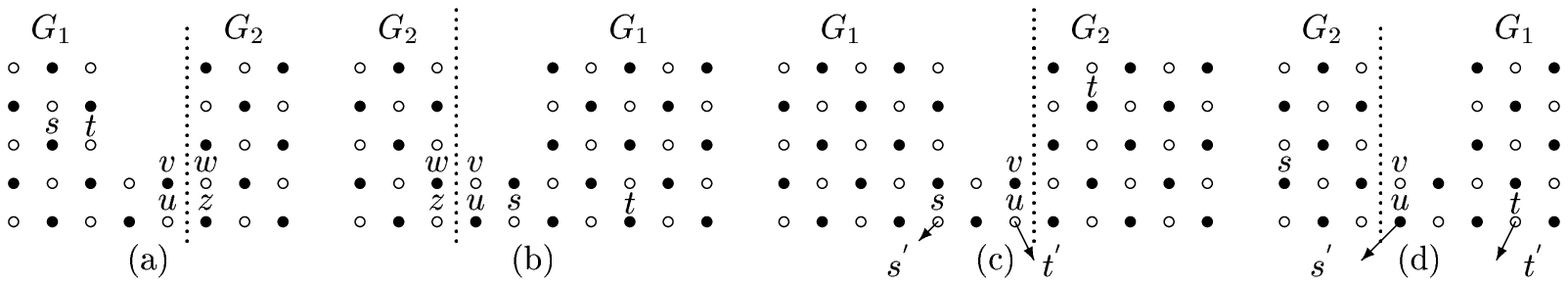}
  \caption[]%
  {\small The $C-$shaped grid graphs in which there is no Hamiltonian $(s,t)-$path.}
\label{fig:s6a}
\end{figure}
\item [(b)] $C(m,n,k,l)$ is even-sized, $d=odd$, $c=odd$, and
\begin{itemize}
\item [(b$_1$)] $s,t\in G_1$ (see Fig. \ref{fig:s6a}(a) and \ref{fig:s6a}(b)); or
\item [(b$_2$)] $s_x\leq d$, $t_x>d+k$, $s^{'}=s$, $t^{'}=u$ (or $s^{'}=u$, $t^{'}=t)$, and $(G_1,s^{'},t^{'})$ is not acceptable (see Fig. \ref{fig:s6a}(c) and \ref{fig:s6a}(d)).
\end{itemize}
\end{itemize}
\begin{figure}[tb]
  \centering
  \includegraphics[scale=1]{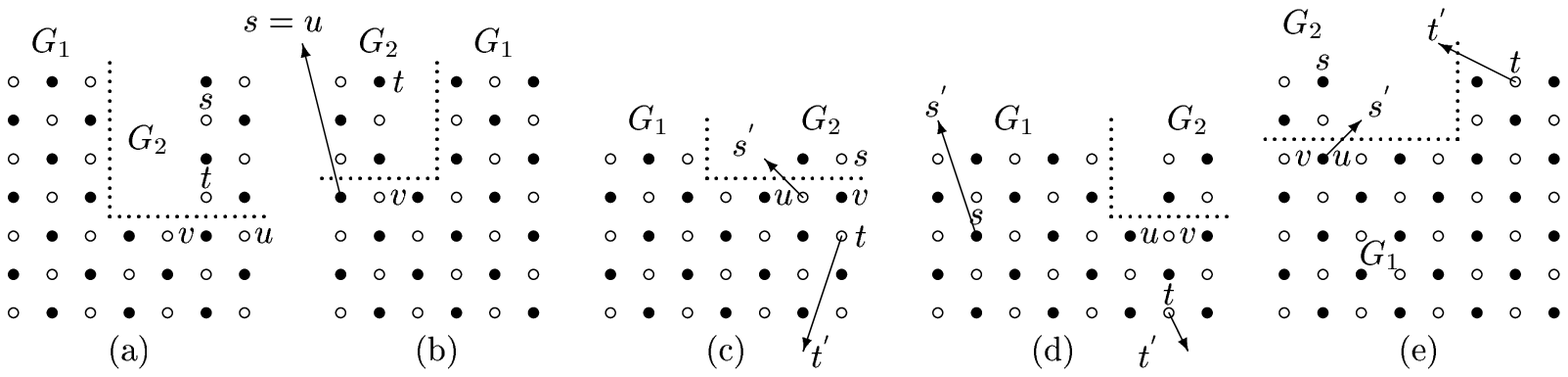}
  \caption[]%
  {\small The $C-$shaped grid graphs in which there is no Hamiltonian $(s,t)-$path.}
\label{fig:s4}
\end{figure}
\item [(F14)] $n=odd$, $n-l>2$, $[(d=odd>1$ and $c=2)$ or $(d=2$ and $c=odd>1)]$,
and $[(C(m,n,k,l)$ is odd-sized$)$ or $(m=even$ and $k\times
l=$odd$\times$even$)]$. Let $\{G_1,G_2\}$ be a $L-$shaped separation
$($type I$)$ of $C(m,n,k,l)$ such that $G_1=L(m,n,k^{'},l)$, where
$k^{'}=m-d$ or $k^{'}=m-c$, $G_2=R(2,l)$ $($see Fig. \ref{fig:s4}),
and exactly two vertices $u$ and $v$ are in $G_1$ that are connected
to $G_2$. And one of the following cases occurs
\begin{itemize}
\item [(a)] $C(m,n,k,l)$ is odd-sized, $[(m=even)$ or $(m=odd$ and $k=even)]$, and
\begin{itemize}
\item [(a$_1$)] $s,t\in G_2$ (see Fig. \ref{fig:s4}(a)); or
\item [(a$_2$)] $s=u$ and $t\in G_2$ $($or $s\in G_2$ and $t=u)$ (see Fig. \ref{fig:s4}(b)); or
\item [(a$_3$)] $m=$odd, $l=$odd, $s\in G_1$, $t\in G_2$, $s^{'}=s$, $t^{'}=u$ $($or $t\in G_1$, $s\in G_2, t^{'}=t$,
$s^{'}=u)$ and $(G_1,s^{'},t^{'})$ satisfies condition (F1) (that
is, $\{s^{'},t^{'}\}$ is a vertex cut$)$ (see Fig. \ref{fig:s4}(c)),
\end{itemize}
\item [(b)] $C(m,n,k,l)$ is even-sized, $s^{'}=s$, $t^{'}=t$, if $s^{'}$ (or $t^{'})\notin G_1$ then
$s^{'}=u$ (or $t^{'}=u)$, and $(G_1,s^{'},t^{'})$ satisfies
condition (F9) (see Fig. \ref{fig:s4}(d) and \ref{fig:s4}(e)).
\end{itemize}
\item [(F15)] $C(m,n,k,l)$ is odd-sized, $m=even$, $n=odd$, $n-l=4$, and
\begin{itemize}
\item [(i)]  $d=odd>1$, $[(l=1$ and $c=even\geq 4)$ or $(s_y,t_y>l$ and $c=2)]$, and $s_x,t_x>d+k+1$ (Fig. \ref{fig:s11}(a)); or
\item [(ii)] $c=odd>1$, $[(l=1$ and $d=even\geq 4)$ or $(s_y,t_y>l$ and $d=2)]$, and $s_x,t_x<d$.
\end{itemize}
\begin{figure}[tb]
  \centering
  \includegraphics[scale=1]{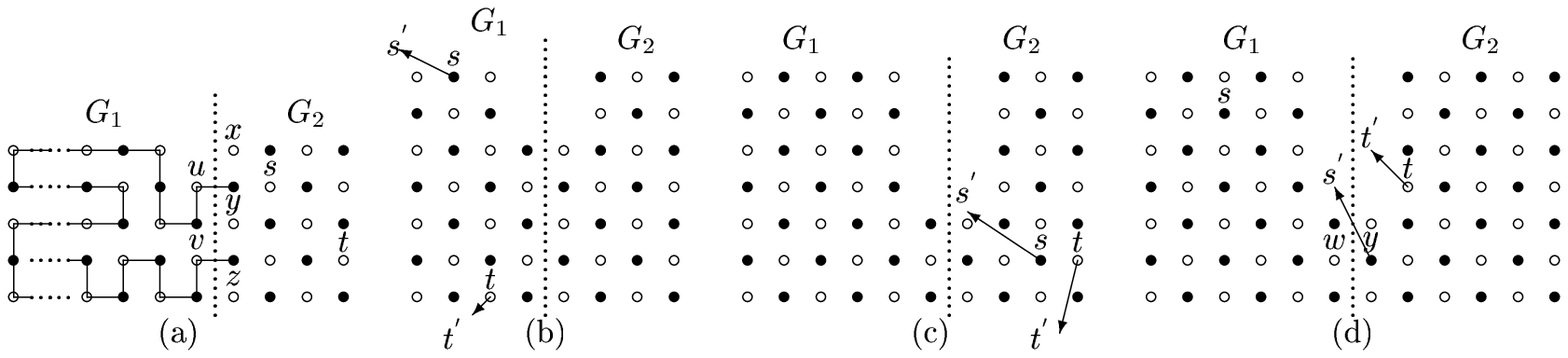}
  \caption[]%
  {\small The $C-$shaped grid graphs in which there is no Hamiltonian $(s,t)-$path.}
\label{fig:s11}
\end{figure}
\item [(F16)] $C(m,n,k,l)$ is even-sized, $m\times n=$even$\times$odd,
$c=odd>1$, $d=odd>1$, and $n-l=odd>1$. Assume that $\{G_1,G_2\}$ is
a vertical separation of $C(m,n,k,l)$ such that
$G_1=L(m^{'},n,k^{'},l)$, where $m^{'}=d+1$ and $k^{'}=m^{'}-d$,
$G_2=L(m-m^{'},n,k^{''},l)$, where $k^{''}=k-k^{'}$ (see Fig.
\ref{fig:s11}(b)-(d)), and at least three vertices $u$, $w$ and $v$
are in $G_1$ that are connected to $G_2$. Let $y\in G_2$ such that
$w$ and $y$ are adjacent. And $(s^{'}=s$ and $t^{'}=w$, if $t\in
G_1$ let $t^{'}=t)$ or $(s^{'}=y$ and $t^{'}=t$, if $s\in G_2$ let
$s^{'}=s)$, and $(G_1,s^{'},t^{'})$ or $(G_2,s^{'},t^{'})$ satisfies
condition (F9).
\begin{figure}[tb]
  \centering
  \includegraphics[scale=1]{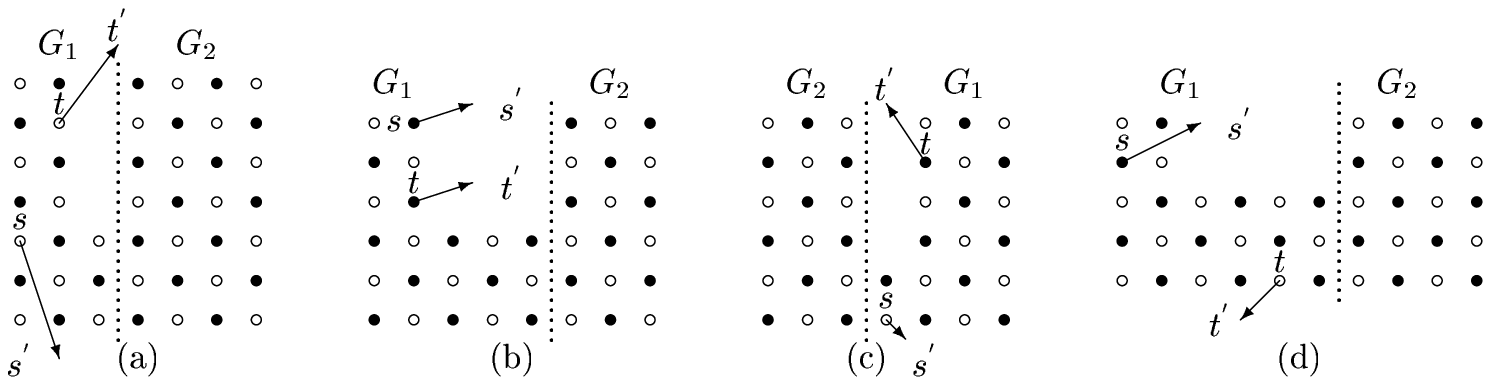}
  \caption[]%
  {\small The $C-$shaped grid graphs in which there is no Hamiltonian $(s,t)-$path.}
\label{fig:s11a1}
\end{figure}
\begin{figure}[tb]
  \centering
  \includegraphics[scale=1]{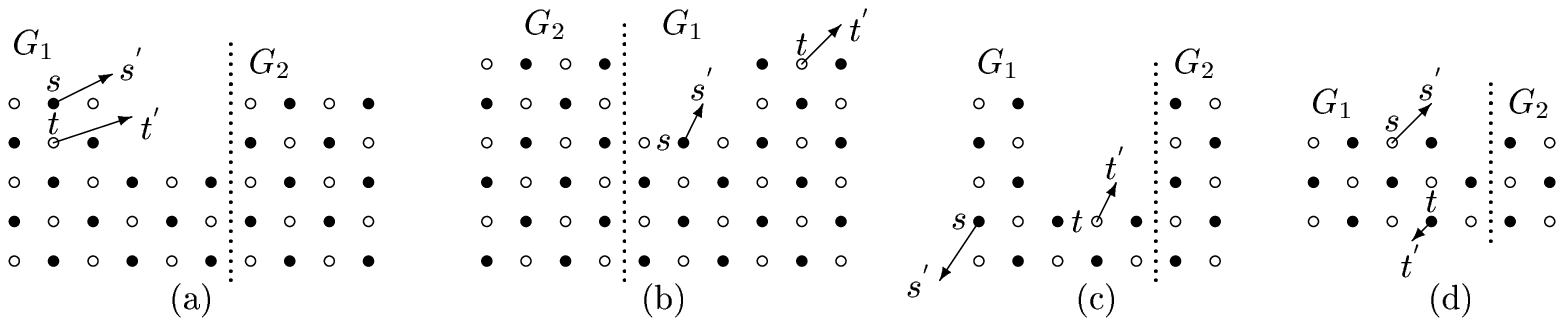}
  \caption[]%
  {\small The $C-$shaped grid graphs in which there is no Hamiltonian $(s,t)-$path.}
\label{fig:s11b1}
\end{figure}
\begin{figure}[tb]
  \centering
  \includegraphics[scale=1]{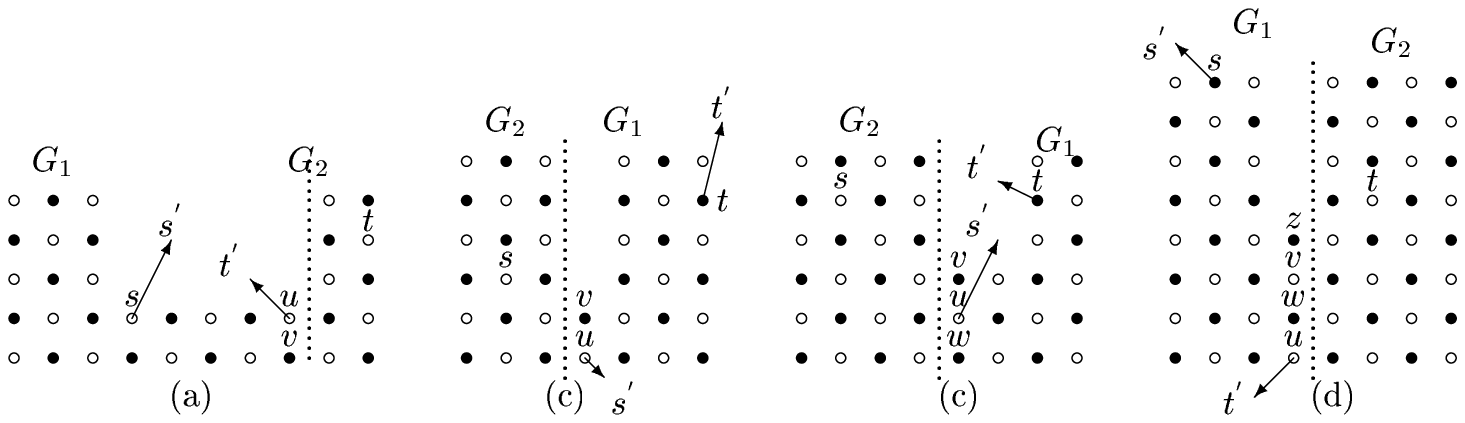}
  \caption[]%
  {\small The $C-$shaped grid graphs in which there is no Hamiltonian $(s,t)-$path.}
\label{fig:s11c1}
\end{figure}
\begin{figure}[tb]
  \centering
  \includegraphics[scale=.95]{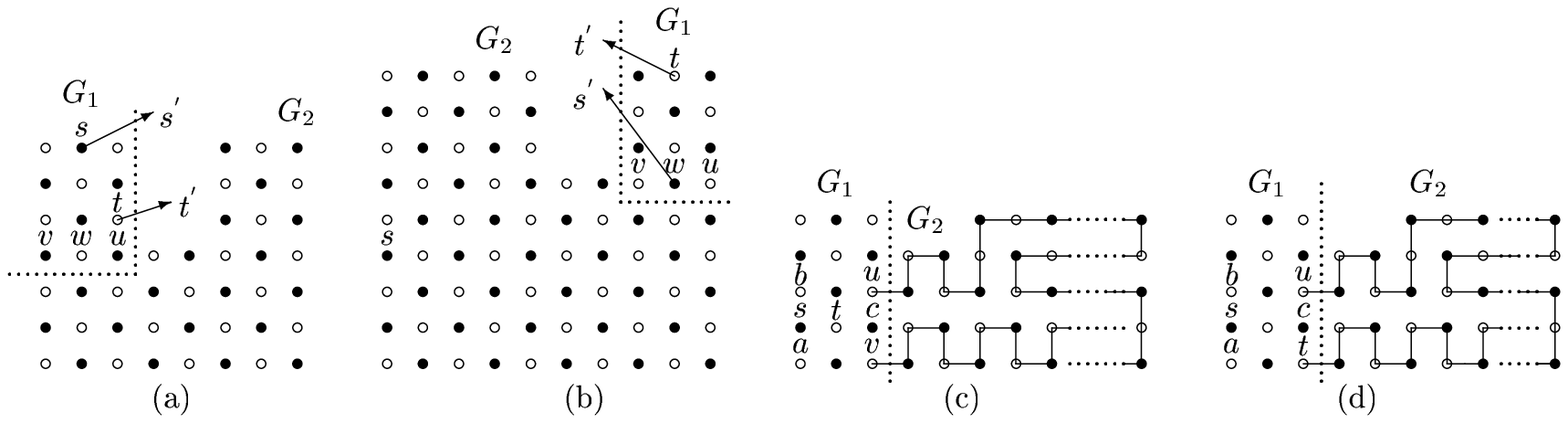}
  \caption[]%
  {\small The $C-$shaped grid graphs in which there is no Hamiltonian $(s,t)-$path.}
\label{fig:s11a2}
\end{figure}
\item [(F17)] $n-l\geq 2$, $d,c>1$ and $[(C(m,n,k,l)$ is odd-sized and $[(n=even)$
or $(n=odd$ and $[(m=even)$ or $(m=odd$ and $[k\times
l=$odd$\times$even or even$\times$odd$])])]$ or $(C(m,n,k,l)$ is
even-sized and $[(m\times n=$odd$\times$odd$)$, $(n=$even$)$, or
$(m\times n=$even$\times$odd and $[(c$ and $d$ are even$)$,
$(c=even\geq 4$ and $d=odd)$, or $(c=odd$ and $d=even\geq 4)])])]$.
Let $\{G_1,G_2\}$ be a vertical separation of $C(m,n,k,l)$ such that
$G_1=L(m^{'},n,k,l)$, $G_2=R(m-m^{'},n)$, $m^{'}=d+k$ (or $G_1=L(m-m^{'},n,k,l)$, $G_2=R(m^{'},n)$, $m^{'}=d)$, $G_2$ is even-sized, and at
least two vertices $v$ and $u$ are in $G_1$ which are connected to
$G_2$. If $C(m,n,k,l)$ is even-sized, $k=1$, and $n-l=even\geq 4$,
then let $m^{'}-k>2$ $($ or $m-m^{'}-k>2)$ in $G_1$. And $s^{'}=s$, $t^{'}=t$, if
$s^{'}$ (or $t^{'})\notin G_1$ then $s^{'}=u$ (or $t^{'}=u)$, and
$(G_1,s^{'},t^{'})$ satisfies one of the conditions (F5), (F6),
(F7), (F8), or (F9) (Fig. \ref{fig:s11a1}, \ref{fig:s11b1}, and
\ref{fig:s11c1}).
\item [(F18)] $C(m,n,k,l)$ is even-sized, $n=odd$, $d=odd>1$, $c=odd>1$, $n-l=even\geq 4$, and one of the following cases occurs
\begin{itemize}
\item [(a)] Let $\{G_1,G_2\}$ be a $L-$shaped separation $($type I$)$ of
$C(m,n,k,l)$ such that $G_1$ is an even-sized rectangular grid
subgraph with $V(G_1)=\{1\leq x\leq d $ $($or $d+k+1\leq x\leq m)$
and $1\leq y\leq l+1\}$, $G_2$ is an even-sized solid grid subgraph
(see Fig. \ref{fig:s11a2}(a) and \ref{fig:s11a2}(b)), and exactly
three vertices $v$, $w$, and $u$ are in $G_1$ that are connected to
$G_2$. And $s^{'}=s$, $t^{'}=t$, if $s^{'}$ (or $t^{'})\notin G_1$
then $s^{'}=w$ (or $t^{'}=w)$, and $(G_1,s^{'},t^{'})$ satisfies
condition (F2); or
\item [(b)] $n-l=4$ and
\begin{itemize}
\item [(b$_1$)] $s_y,t_y>l+1$ and $[(d=3, \ s_x,t_x\leq d,\  s=(1,n-1)$, and $t_x>s_x)$ or $(c=3,\ s_x,t_x>d+k,\ s_x<t_x,$ and
$t=(m,n-1))]$ (Fig. \ref{fig:s11a2}(c), \ref{fig:s11a2}(d), and \ref{fig:s11b2}(a)); or
\item [(b$_2$)] $s$ is black and $[(s_x\leq d$ and $t_x>d)$ or
$(d+1\leq s_x\leq d+k$ and $t_x>d+k)]$ (Fig. \ref{fig:s11b2}(b) and
\ref{fig:s11b2}(c)); or
\item [(b$_3$)] $d+1\leq s_x,t_x\leq d+k$ and $[(t_x>s_x$ and $s$ is black$)$
or $(s_x=t_x$, $s_y$ (or $t_y)=l+2$, and $t_y$ (or $s_y)=l+3)]$
(Fig. \ref{fig:s11b2}(d) and \ref{fig:s11c2}).
\end{itemize}
\end{itemize}
\begin{figure}[tb]
  \centering
  \includegraphics[scale=1]{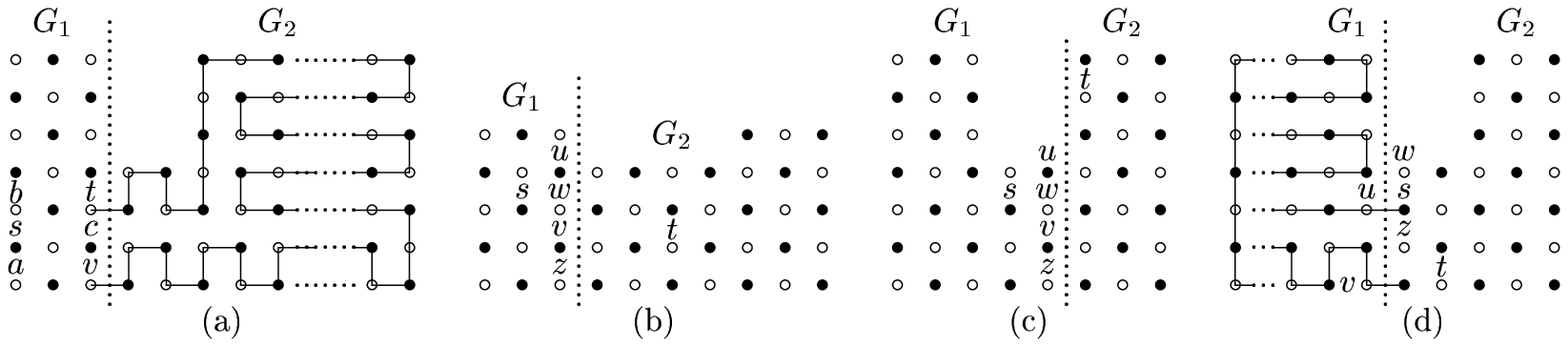}
  \caption[]%
  {\small The $C-$shaped grid graphs in which there is no Hamiltonian $(s,t)-$path.}
\label{fig:s11b2}
\end{figure}
\begin{figure}[tb]
  \centering
  \includegraphics[scale=.93]{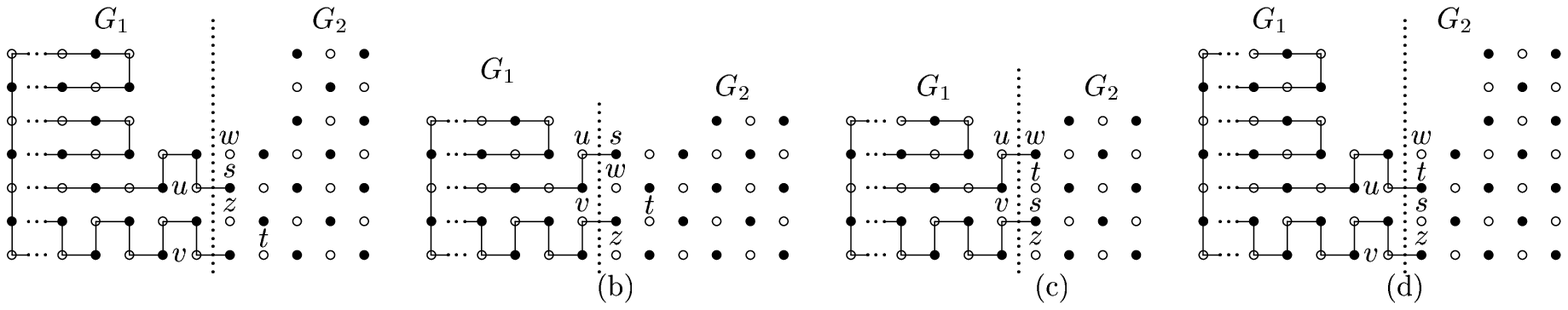}
  \caption[]%
  {\small The $C-$shaped grid graphs in which there is no Hamiltonian $(s,t)-$path.}
\label{fig:s11c2}
\end{figure}
\end{itemize}
\par The following results directly follows from conditions (F1), (F3), and (F10)-(F18).
\begin{co}\label{Corollary:co3}
Suppose that $C(m,n,k,l)$ is a $C-$shaped grid graph with two given vertices
$s$ and $t$. Let $\{G_1,G_2\}$ be a vertical (or $L-$shaped $($type I$)$) separation  of
$C(m,n,k,l)$ such that $G_1$ is a $L-$shaped grid subgraph, $G_2$ is a rectangular grid subgraph, and
$s,t\in G_1$. If $(C(m,n,k,l)$ is not acceptable, then $(G_1,s,t)$ is not acceptable.\end{co}
\begin{defi}
A $C-$shaped Hamiltonian path problem $(C(m,n,k,l), s, t)$ is
\textit{acceptable} if it is color-compatible and $(C(m,n,k,l),s,t)$
does not satisfy any of conditions (F1), (F3), and (F10)-(F18).
\end{defi}
\par
We define the length of a path in a grid graph be the number of vertices of the path. In
any grid graph, the length of any path between two same-colored vertices is odd and the
length of any path between two different-colored vertices is even.
\begin{thm} \label{Theorem:1h}
Let $C(m,n,k,l)$ be a $C-$shaped grid graph and $s$ and $t$ be two
distinct vertices of it. If $(C(m,n,k,l), s, t)$ is Hamiltonian,
then $(C(m,n,k,l),s,t)$ is acceptable.
\end{thm}
\begin{proof}
Arguing by contrapositive, Suppose $(C(m,n,k,l),s,t)$ is not
acceptable, then $(C(m,n,k,l), s, t)$ has no Hamiltonian
$(s,t)-$path. Clearly, if $(C(m,n,k,l),s,t)$ is not color-compatible
then $(C(m,n,k,l), s, t)$ has not Hamiltonian $(s,t)-$path. Thus,
without loss of generality, suppose $(C(m,n,k,l),s,t)$ is
color-compatible. In the following, we will show that if one of the
conditions (F1), (F3), and (F10)-(F18) holds, then $(C(m,n,k,l), s,
t)$ has no Hamiltonian $(s,t)-$path. \par  (F1) and (F3): See Fig.
\ref{fig:s1}(a)-(d).
\par (F10): (i) Consider Fig. \ref{fig:s1}(e) and \ref{fig:s1}(f). Since $n-l=1$,
we can easily see that there is no Hamiltonian path in
$C(m,n,k,l)$.\par  (ii) Consider Fig. \ref{fig:s1c}(a). Since
$n-l=1$, then the Hamiltonian path $P$ of $C(m,n,k,l)$ which starts
from $s$ must pass through all the vertices of $G_1$, and leaves
$G_1$ at $w$, then enter to $G_2$ at $z$, and end at $t$. Clearly,
if $(G_1,s,w)$ or $(G_2,z,t)$ is not acceptable, then by Theorem
\ref{Theorem:6t} or \ref{Theorem:1a}, $(G_1,s,w)$ or $(G_2,z,t)$ has
no Hamiltonian path, respectively, and hence $(C(m,n,k,l),s,t)$ has
no Hamiltonian $(s,t)-$path.
\par(F11). The proof is a straightforward; see Fig. \ref{fig:s1c}(b)-(e).
\par (F12): Consider Fig. \ref{fig:s1c}(f). Notice that, here, $C(m,n,k,l)$ is even-sized
and the number of vertices with white color is two more than the
number of vertices with black color. Since $C(m,n,k,l)$ is
even-sized and colors of vertices of any path must alternate between
black and white, it is clear that two vertices with white color
remain out of the path, and hence $C(m,n,k,l)$ has no Hamiltonian
$(s,t)-$path.
\par (F13): (a$_{11}$) and (b$_1$) Consider Fig. \ref{fig:s2}(a), \ref{fig:s6a}(a), and \ref{fig:s6a}(b). Clearly, since $n-l=2$, the Hamiltonian path $P$ must enter to
$G_1$ (resp. $G_2$, where $C(m,n,k,l)$ is odd-sized) through one of the vertices $u$ (or $v$) (resp. $w$ (or $z))$ then the path $P$
leaves $G_1$ (resp. $G_2$) after visiting the vertices of $G_1$ by $v$ (or $u$) (resp. $G_2$ by $z$ (or $w$)) .
It is clear that $(G_1,v,u)$ (resp. $(G_2,w,z))$ is not acceptable, because $G_1$ (resp. $G_2)$ is
odd-sized and $v$ and $u$ (resp. $w$ and $z$) have different colors. Thus by Theorem
\ref{Theorem:6t}, $(G_1,u,v)$ (resp. $(G_2,w,z))$ does not have any Hamiltonian
$(u,v)-$path (resp. $(w,z)-$path). Hence, $(C(m,n,k,l),s,t)$ has no Hamiltonian
$(s,t)-$path. \par (a$_{12}$), (a$_2)$, and (b$_2$) Consider Fig.
\ref{fig:s2}(b)-(d), \ref{fig:s6a}(c), and \ref{fig:s6a}(d). Since $n-l=2$, the Hamiltonian path $P$ of
$C(m,n,k,l)$ which starts from $s$ must pass through all the
vertices of $G_1$ (or $G_2$), leave $G_1$ at one of the vertices $v$
or $u$ (or leave $G_2$), then enter to $G_2$ ( enter to $G_1$) and
pass through all the vertices of $G_2$ (or $G_1$) and end at $t$.
Clearly, if $(G_1,s^{'},t^{'})$ is not acceptable then by Theorem \ref{Theorem:6t} $(G_1,s^{'},t^{'})$ has no
Hamiltonian path, and hence $(C(m,n,k,l),s,t)$ has no Hamiltonian
$(s,t)-$path.
\par (F14): (a) Like in the proof of condition (F13), we can obtain that $(C(m,n,k,l),s,t)$
has no Hamiltonian path (see Fig. \ref{fig:s4}(a)-(c)).\par (b)
Since $G_2$ is connected to $G_1$ by two vertices $v$ and $u$, using
the same argument as in the proof [\cite{991}, Theorem 3.2,
condition (F8)], it can be proved that  $C(m,n,k,l)$ does not have
any Hamiltonian $(s,t)-$path (see Fig. \ref{fig:s4}(d) and
\ref{fig:s4}(e)). Note that, here, $G_1$ is a $L-$shaped grid
subgraph.\par (F15): Suppose that $c=even$. Consider Fig.
\ref{fig:s11}(a). Let $\{G_1,G_2\}$ be a vertical separation of $C(m,n,k,l)$ such that $G_1=L(m^{'},n,k,l)$, $G_2=R(m-m^{'},n)$, and $m^{'}=d+k$. Since
$c=even$, thus $G_2$ is even-sized. Moreover, since $C(m,n,k,l)$ is
odd-sized, we conclude that $G_1$ is odd-sized. Since $G_1$ is
odd-sized and $n-l=4$, a Hamiltonian path $P$ of $C(m,n,k,l)$ that
starts from $s$ must be enter to $G_1$ for the first time through
one of the vertices $u$ (or $v$) then pass through all the vertices
of $G_1$ and leave $G_1$ at $v$ (or $u$) and end at $t$. Clearly in
this case, if $s_x,t_x>d+k+1$, then one of the three vertices $x$,
$y$, or $z$ remains out of the path. By symmetry, the result
follows, if $d=even$.
\par (F16): By condition (F9), the proof is Straightforward (as shown in Fig. \ref{fig:s11}(b)-(d)).
\par (F17): This follows immediately from conditions (F5)-(F9); see Fig. \ref{fig:s11a1}, \ref{fig:s11b1}, and \ref{fig:s11c1}.\par
(F18): (a) Since $G_1$ is connected to $G_2$ by three vertices $v$,
$w$ and $u$, in a similar manner as in the proof [\cite{991},
Theorem 3.2, condition (F9)], we derive $C(m,n,k,l)$ does not have
any Hamiltonian $(s,t)-$path (see Fig. \ref{fig:s11a2}(a) and
\ref{fig:s11a2}(b)).
\par (b$_1$) We shall only prove the first case
$(s_x,t_x\leq d)$. The other case ($s_x,t_x>d+k$) is similar.
Consider Fig. \ref{fig:s11a2}(c), \ref{fig:s11a2}(d), and \ref{fig:s11b2}(a). Let $\{G_1,G_2\}$ be a
vertical separation (type I) of $C(m,n,k,l)$ such that
$G_1=R(m^{'},n)$, $G_2=L(m-m^{'},n,k,l)$, $m^{'}=d$ and $s,t\in
G_1$. Let $u=(d,l+2)$ and $v=(d,n)$. Notice that, here, $G_1$ and
$G_2$ are odd-sized. Since $n-l=4$ and $G_2$ is odd-sized, thus the
Hamiltonian path $P$ of $C(m,n,k,l)$ which starts from $s$ should
pass through some vertices of $G_1$, leaves $G_1$ at $u$ (or $v)$,
then passes through all the vertices of $G_2$ and reenters to $G_1$
at $v$ (or $u)$, and passes through all the remaining vertices of
$G_1$ and ends at $t$. A simple check shows that one of the three
vertices $a$, $b$, and $c$ remains out of path $P$.
 \par (b$_2$)
Consider Fig. \ref{fig:s11b2}(b) and
\ref{fig:s11b2}(c). Let
$\{G_1,G_2\}$ be a vertical separation of $C(m,n,k,l)$ such that
$G_1=R(m^{'},n)$, $G_2=L(m-m^{'},n,k,l)$, where $m^{'}=d$ and
$s_x\leq d$, or $G_1=L(m^{'},n,k,l)$ and $G_2=R(m-m^{'},n)$, where
$m^{'}=d+k$ and $d+1\leq s_x\leq d+k$,  $s\in G_1$, and $t\in G_2$.
Note that, in this case, $G_1$ is odd-sized with white majority
color and $G_2$ is odd-sized with black majority color. The
following cases may be considered. \par Case 1. The Hamiltonian path
$P$ of $C(m,n,k,l)$ which starts from $s$ should pass through all
the vertices of $G_1$, leaves $G_1$ at $w$ (or $z)$ (or $v$ (or
$u)$, enters $G_2$, and passes through all the vertices of $G_1$,
and ends at $t$. This is impossible, $(G_1,s,t^{'})$ is not
acceptable, where $t^{'}$ is $w$, $z$, $v$, or $u$.\par Case 2. The
Hamiltonian path $P$ of $C(m,n,k,l)$ which starts from $s$ should
pass through some vertices of $G_1$, leaves $G_1$ at $w$ (or $z)$
(or $v$ (or $u)$, enters $G_2$, then passes through some vertices of
$G_2$, \begin{enumerate}
\item reenters to $G_1$ at $z$ (or $w)$ (or $u$ (or $v))$, passes through
all the remaining vertices of it, leaves it at $u$ (or $v)$ (or $w$
(or $z))$, and passes through all the remaining vertices of $G_2$
and finally ends at $t$. In this case, two subpaths of $P$ which are
in $G_1$ are called $P_1$ and $P_2$, $P_1$ from $s$ to $w$ (or $z)$
(or $v$ (or $u))$ and $P_2$ from $z$ (or $w)$ (or $u$ (or $v))$ to
$u$ (or $v)$ (or $w$ (or $z))$. This is impossible, because the size
of $P_1$ is even (or odd) and the size of $P_2$ is even, then
$|P_1+P_2|$ is even (or odd with black majority color) while $G_1$
is odd-sized with white majority color.
\item  reenters to $G_1$ at $u$ (or $v)$ (or $w$ (or $z))$,
passes through all the remaining vertices of it, leaves it at $z$
(or $w)$ (or $u$ (or $v))$, and passes through all the remaining
vertices of $G_2$ and finally ends at $t$. In this case, $P_2$ from
$u$ (or $v)$ (or $w$ (or $z))$ to $z$ (or $w)$ (or $u$ (or $v))$.
This is impossible, because the size of $P_1$ is even (or odd) and
the size of $P_2$ is even, then $|P_1+P_2|$ is even (or odd with
black majority color) while $G_1$ is odd-sized with white majority
color.
\item  reenters to $G_1$ at $u$ (or $v)$ (or $w$ (or $z))$, passes
through all the remaining vertices of it, leaves it at $v$ (or $u)$
(or $z$ (or $w))$, and passes through all the remaining vertices of
$G_2$ and finally ends at $t$. In this case, $P_2$ from $u$ to $v$
(or $v$ to $u)$ (or $w$ to $z$ or ($z$ to $w$)). This is impossible,
because the size of $P_1$ is even (or odd) and the size of $P_2$ is
odd, then $|P_1+P_2|$ is odd with black majority color (or even)
while $G_1$ is odd-sized with white majority color.
\end{enumerate}
\par (b$_3$) Let $\{G_1,G_2\}$ be a vertical separation of $C(m,n,k,l)$ such that
$G_1=R(m^{'},n)$ and $G_2=L(m-m^{'},n,k,l)$ (or
$G_1=L_1(m^{'},n,k^{'},l)$ and $G_2=L_2(m-m^{'},n,k^{''},l))$
$m^{'}=s_x-1$, $k^{'}=m^{'}-d$, $k^{''}=k-k^{'}$, and $s,t\in G_2$.
Let $u,v\in G_1$ such that $u_y,v_y>l$, $v_x=u_x=m^{'}-1$,
$u_y,v_y=odd$ if $s_x=even$; otherwise $u_y,v_y=even$. Consider Fig.
\ref{fig:s11b2}(d) and \ref{fig:s11c2}. Note that
$G_1$ is an odd sized grid subgraph with white majority color. The
Hamiltonian path $P$ of $C(m,n,k,l)$ must enter to $G_1$ through one
of the vertices $v$ (or $u$), then the path $P$ leaves $G_1$ after
visiting all the vertices $G_1$ by $u$ (or $v$), reenters to $G_2$,
and ends at $t$. One easily check that one of the vertices $w$ or
$z$ remains out of path.
 \end{proof}
\section{\bf Sufficient conditions}
Suppose $(C(m,n,k,l),s,t)$ is an acceptable Hamiltonian path
problem. The purpose of this section is to prove that all acceptable
$C-$shaped Hamiltonian path problems have solutions.
\begin{defi}
A separation is \textit{acceptable} if all of its component are acceptable.
\end{defi}
\begin{defi}
 Two nonincident edges $(u_1,v_1)$ and $(u_2,v_2)$ are \textit{parallel}, if $u_1$
 (resp. $v_1)$ is adjacent to $u_2$ and $v_1$ (resp. $u_1)$  is adjacent to $v_2$ .
\end{defi} 
\par The following three lemmas discuss how to construct a Hamiltonian $(s,t)-$path for $C(m,n,k,l)$.
\begin{lem} \label{Lemma:c11}
Suppose that $(C(m,n,k,l),s,t)$ is an acceptable Hamiltonian path
problem. Let $C(m,n,k,l)$ be even-sized. Then there is an acceptable
separation for $(C(m,n,k,l),s,t)$ and it has a Hamiltonian path.
\end{lem}
\begin{figure}[tb]
  \centering
  \includegraphics[scale=1]{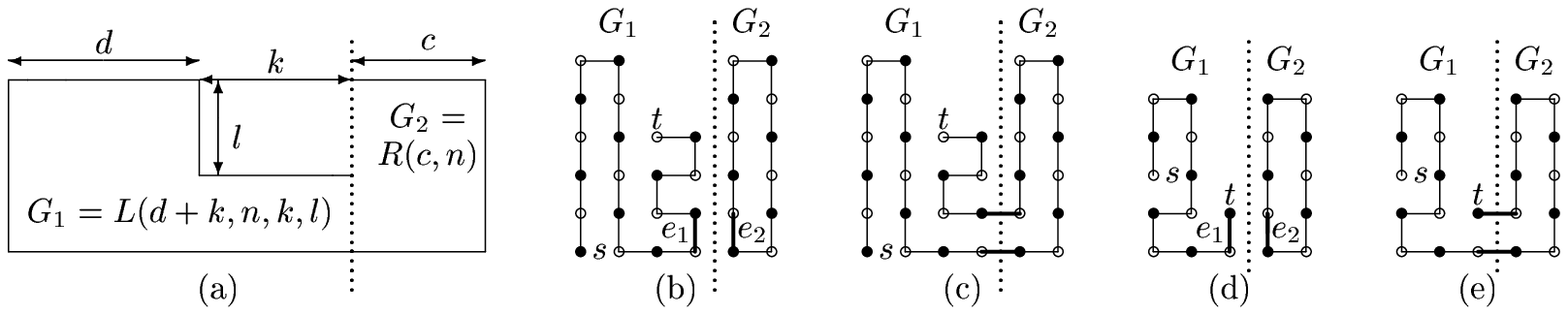}
  \caption[]%
 {\small A vertical separation of $C(m,n,k,l)$ and a Hamiltonian path in $(C(m,n,k,l),s,t)$.}
\label{fig:s13}
\end{figure}
 \begin{proof}
 Here, $s$ and $t$ have different colors. Let $m=even$ (resp. $m=odd)$ and $k=even$
 (resp. $k=odd)$, then $d$ and $c$ are even (or odd). Similarly let $m=even$ (resp. $m=odd)$ and
 $k=odd$ (resp. $k=even)$, then $d=even$ and $c=odd$ or $d=odd$ and $c=even$. Notice that, for $m\times n=$odd$\times$odd,
 since $(C(m,n,k,l),s,t)$ is acceptable, $d$ and $c$ must be even. We have the following five cases.
 \par Case 1. $n-l>1$, $[(c>1$ and $s_x,t_x\leq d+k)$ or ($d>1$ and
 $[(s_x,t_x>d+k)$ or $(d+1\leq s_x\leq d+k$ and $t_x>d+k)])]$, \indent \indent and \\
\indent \indent (a) $n=even$ and $[(k=1$ and $[(n-l=2)$ or
$(n-l=even\geq 4$ and $d^{'}\neq 2)])$ or $(k>1)]$, where $d^{'}=d$
if  \indent \indent $s_x,t_x\leq d+k$; otherwise $d^{'}=c$; or
 \\ \indent \indent (b) $n=odd$ and $[(m=odd$ and $[(k=1$ and $n-l=2)$ or $(k>1)])$ or
$(m=even$, $d=even$, and $c=even)]$.\\ Let $s_x,t_x\leq d+k$. By
symmetry, the result follows, if $(s_x,t_x>d+k)$ or $(d+1\leq
s_x\leq d+k$ and $t_x>d+k)$. Let $\{G_1,G_2\}$ be a vertical
separation of $C(m,n,k,l)$ such that $G_1=L(m^{'},n,k,l)$,
$G_2=R(m-m^{'},n)$, $m^{'}=d+k$, and $s,t\in G_1$ (Fig.
\ref{fig:s13}(a)).
 First, we prove that $(G_1,s,t)$ is acceptable. 
Since $n=even$ or $c=even $, it follows that $G_2$ is even-sized.
Moreover, since $C(m,n,k,l)$ is even-sized, we conclude that $G_1$
is even-sized. By Lemma \ref{Lemma:c0}, $(G_1,s,t)$ is
color-compatible. In the following, we show that $(G_1,s,t)$ is not
in conditions (F1), (F3), (F4), and (F6)-(F9). The condition (F1)
holds, if (i) $d=1$ and $2\leq s_y\ ($ or $t_y)\leq l+1$; (ii) $d=2$
and $2\leq s_y=t_y\leq l+1$; or (iii) $(n-l=2$ or $n=2)$ and $2\leq
s_x=t_x\leq d+k$, clearly if these cases occur, then
$(C(m,n,k,l),s,t)$ is in condition (F1), a contradiction. Therefore,
$(G_1,s,t)$ is not in condition (F1). To satisfy condition (F3), $d$
must be $1$ and $s_y,t_y>l$. If this case holds, then
$(C(m,n,k,l),s,t)$ satisfies condition (F3), a contradiction. Thus,
it follows that $(G_1,s,t)$ does not satisfy condition (F3). The
condition (F4) holds, if $m\times n=$odd$\times$odd, $k\times l=1$,
$n-l>2$, and $m^{'}-k>2$. Since $n-l=2$, where $k=1$, thus
$(G_1,s,t)$ does not satisfy condition (F4). If $(G_1,s,t)$
satisfies condition (F6), (F8), or (F9), then $(C(m,n,k,l),s,t)$
satisfies condition (F17), a contradicting the assumption.
Therefore, it follows that $(G_1,s,t)$ is not in condition (F6),
(F8), and (F9). The condition (F7) holds, if $k=1$, $d=2$,
$n-l=even\geq 4$, $s=(1,l+1)$, and $t=(d+k,l+2)$. This is
impossible, because of $n-l=2$, and hence $(G_1,s,t)$ is not in
condition (F7). Therefore, $(G_1,s,t)$ is acceptable.
 \par Now, we show that $(C(m,n,k,l),s,t)$ has a Hamiltonian path. Since $(G_1,s,t)$ is acceptable,
 by Theorem \ref{Theorem:6t} it has a Hamiltonian $(s,t)-$path. Thus, we construct a Hamiltonian path
 in $(G_1,s,t)$ by the algorithm in \cite{991}. Furthermore, since $G_2$ is even-sized, it has a Hamiltonian
 cycle by Lemma \ref{Lemma:1m}. Then by combining Hamiltonian cycle and path using two parallel edges $e_1$
 and $e_2$ (Fig. \ref{fig:s13}(b) and \ref{fig:s13}(d)), a Hamiltonian $(s,t)-$path for $(C(m,n,k,l),s,t)$ is obtained,
 as shown in Fig. \ref{fig:s13}(c) and \ref{fig:s13}(e). Now, we describe combining a Hamiltonian path in $(G_1,s,t)$ with the constructed
cycle in $G_2$. Any Hamiltonian path $P$ in $G_1$ contains all the
vertices of $G_1$. Therefore, $P$ should contain a boundary edge of
$G_1$ that has a parallel edge in $G_2$. Moreover, since $n-l>1$, it is easy to check that there is at least one edge for combining Hamiltonian cycle and
path. \par
Case 2. $n-l=even\geq 4$, $k=1$ and \\
\indent \indent
 (a) $c>1$, $s_x,t_x\leq d+k$, and $[(n=even$ and $d=2)$ or $(n=odd$ and $[(l=1$ and $c=2)$ or  $(l=odd>1$ \indent \indent and $c\geq 2)])]$; or \\
 \indent \indent (b) $d>1$, $[(s_x,t_x>d+k)$ or $(d+1\leq s_x\leq d+k$ and $t_x>d+k)]$,
 and $[(n=even$ and $c=2)$ or $(n=odd$ and  \indent \indent $[(l=1$ and $d=2)$ or $(l=odd>1$ and $d\geq 2)])]$.\\
Assume that $s_x,t_x\leq d+k$. By symmetry, the result follows, if $(s_x,t_x>d+k)$ or $(d+1\leq s_x\leq d+k$ and $t_x>d+k)$.
Let $\{G_1,G_2\}$ be a $L-$shaped separation (type I) of $C(m,n,k,l)$ such that $G_1=L(m,n,k^{'},l)$ and $G_2=R(m^{'},n^{'})$,
$k^{'}=m-d$, $m^{'}=k^{'}-k$, $n^{'}=l$, and $s,t\in G_1$ (see Fig. \ref{fig:s13a}(a)). 
In the following, we show that $(G_1,s,t)$ is acceptable. Since
$l=even$, where $n=even$, or $c=even$, where $n=odd$, thus $G_2$ is
even-sized, and since $C(m,n,k,l)$ is even-sized, we conclude that
$G_1$ is even-sized. By Lemma \ref{Lemma:c0}, $(G_1,s,t)$ is
color-compatible. Since $m-k^{'}=even$, $k=1$, and $c>1$, we have
$m\geq 5$ and $k^{'}>1$. Moreover, since $m\geq 5$, $m-k^{'}=even$,
$n-l=even\geq 4$, and $k^{'}>1$, it is obvious that $(G_1,s,t)$ is
not in conditions (F3), (F4), and (F6)-(F9). $(G_1,s,t)$ is not in condition (F1), the proof is the same as Case 1. Therefore,
$(G_1,s,t)$ is acceptable. Now, we show that $(C(m,n,k,l),s,t)$ has
a Hamiltonian path. Let $l>1$, then the Hamiltonian path in $(C(m,
n, k, l), s, t)$ is obtained similar to Case 1. Notice that since
$m^{'}\geq 2$, there exists at least one edge for combining
Hamiltonian cycle and path. Now, let $l=1$. In this case, $G_2$ is a
one-rectangle, where $|G_2|=2$. Let two vertices $v_1,v_2\in G_2$
(Fig. \ref{fig:s13a}(b)) and $P$ be a Hamiltonian path in $G_1$.
Using algorithm in \cite{991}, there exists an edge $e_1$ such that
$e_1\in P$ is on the boundary of $G_1$ facing $G_2$, as shown in Fig. \ref{fig:s13a}(b). Thus, by
merging $(v_1,v_2)$ to this edge, we obtain a Hamiltonian path for
$C(m,n,k,l)$ as shown in Fig. \ref{fig:s13a}(c).
 \par Case 3.
 $n=odd$, $m=odd$, $k\times l=1$, $n-l=even\geq 4$, and $[(s_x,t_x\leq d+k$ and $c=even\geq 4)$ or $(d=even\geq 4$ and $[(s_x,t_x>d+k)$ or
 $(d+1\leq s_x\leq d+k$ and $t_x>d+k)])]$.
Suppose $s_x,t_x\leq d+k$. By symmetry, the result follows, if
$(s_x,t_x>d+k)$ or $(d+1\leq s_x\leq d+k$ and $t_x>d+k)$. Let
$\{G_1,G_2\}$ be a vertical separation of $C(m,n,k,l)$ such that
$G_1=C(m^{'},n,k,l)$, $G_2=R(m-m^{'},n)$, $m^{'}=d+k+2$, and $s,t\in
G_1$.
 We know that $d=even$ and $c=even$. Thus, $m-m^{'}=even$, $G_2$ is even$\times$odd, and $G_2$ is even-sized. By Lemma
 \ref{Lemma:c0}, $(G_1,s,t)$ is color-compatible. Since $n-l\geq 4$, $l= 1$, $d=even$, $k=1$, and $c^{'}=2$ implies $m^{'},n\geq 5$.
 Moreover, since $d=even$, $c^{'}=2$, $m^{'},n\geq 5$, and $n-l\geq 4$, it suffices to prove $(G_1,s,t)$ is not in condition (F1).
 A simple check shows that $(G_1,s,t)$ is not in condition (F1). Now, we show that $(C(m,n,k,l),s,t)$ has a Hamiltonian path. In this case,
 $(G_1,s,t)$ in Case 2. The Hamiltonian path in $(C(m, n, k, l), s, t)$ is obtained similar
to Case 1. Notice that, since $n\geq 5$, there is at least one edge for combining Hamiltonian cycle and path. \par
Case 4. $n=odd$ and $[d=odd$ or $c=odd]$.
 \par Subcase 4.1. $d=odd$, $c=odd$, and $n-l=odd>1$. \par
 Subcase 4.1.1. $(s_x,t_x\leq d+1$ and $c>1)$ or $(s_x,t_x>d+1$ and $d>1)$. This case is similar to Case 3,
 where $G_1=L(m^{'},n,k^{'},l)$, $G_2=L(m-m^{'},n,k^{''},l)$, $m^{'}=d+1$, $k^{'}=m^{'}-d$, $k^{''}=k-k^{'}$, and $[s,t\in G_1$ or
 $s,t\in G_2]$. Consider Fig. \ref{fig:s13a}(d). Clearly, $G_1$ and $G_2$ are even-sized. By Lemma \ref{Lemma:c0}, $(G_1,s,t)$ (or $(G_2,s,t))$
 is color-compatible. Note that, because of $n-l=odd>1$ and $l=even$, we have $n>3$. Moreover, since $d$, $c$, and $n-l$ are odd and $n>3$, it
 suffices to prove $(G_1,s,t)$ (or $(G_2,s,t))$ is not in conditions (F1), (F3), and (F9). $(G_1,s,t)$ (or $(G_2,s,t))$ does not satisfy conditions
 (F1) and (F3), the proof is the same as Case 1. If $(G_1,s,t)$ (or $(G_2,s,t))$ satisfies condition (F9), then $(C(m,n,k,l),s,t)$ satisfies condition
 (F16), we have a contradiction. Thus $(G_1,s,t)$ (or $(G_2,s,t))$ does not satisfy condition (F9). Hence $(G_1,s,t)$ (or $(G_2,s,t))$ is acceptable. The Hamiltonian path
 in $(C(m, n, k, l), s, t)$ is
obtained similar to Case 1. In this case, $G_2$ (or $G_1)$ is an
even-sized $L-$shaped grid graph, thus it has a Hamiltonian cycle by
Lemma \ref{Lemma:6t1}.
 \par Subcase 4.1.2. $s_x\leq d+1$ and $t_x>d+1$. This case is the same as Subcase 4.1.1,
 where $s,p\in G_1$, $q,t\in G_2$, $p$ and $q$ are adjacent, and  $p=(m^{'},n)$ if $s$ is
 white; otherwise $p=(m^{'},n-1)$. Consider Fig. \ref{fig:s13a}(d). A simple check shows that $(G_1,s,p)$
 and $(G_2,q,t)$ are color-compatible. By the same argument as in the proof of Subcase 4.1.1, it is sufficient to show that
 $(G_1,s,p)$ and $(G_2,q,t)$ are not in conditions (F1), (F3), and (F9). The condition (F1) occurs, when
 $s_y,t_y\leq l+1$ and $[(d=1$ and $s\neq (1,1))$ or $(c=1$ and $t\neq (m,1))]$. If this case holds, then $(C(m,n,k,l),s,t)$ satisfies
 condition (F1), a contradiction. Thus, $(G_1,s,p)$ is not in condition (F1). The condition (F3) holds, if $s_y,t_y>l$ and $[(d=1)$ or
 $(c=1)].$ It is obvious that if this case holds, then $(C(m,n,k,l),s,t)$ satisfies condition (F3), a contradiction. Therefore, $(G_1,s,p)$ does not
 satisfy condition (F3). $(G_1,s,p)$ and $(G_2,q,t)$ do not satisfy condition (F9), the proof is similar to Subcase 4.1.1. Hence, $(G_1,s,p)$ and $(G_2,q,t)$ are acceptable. Now, we show that $(C(m,n,k,l),s,t)$ has a Hamiltonian. Since $(G_1,s,p)$ and $(G_2,q,t)$ are acceptable, by Theorem \ref{Theorem:6t} have Hamiltonian paths. Thus, we construct Hamiltonian paths in $(G_1,s,p)$ and $(G_2,q,t)$ by the algorithm in \cite{991}. Then the Hamiltonian path for $(C(m,n,k,l),s,t)$ can be obtained by connecting two vertices $p$ and $q$ as shown in Fig. \ref{fig:s13a}(e).
\begin{figure}[tb]
  \centering
  \includegraphics[scale=.95]{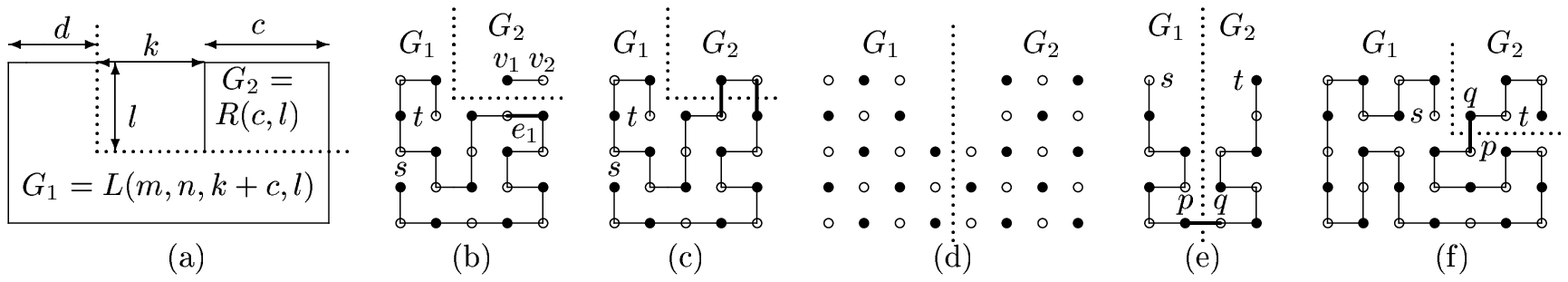}
  \caption[]%
 {\small (a) A $L-$shaped separation of $C(m,n,k,l)$, (b) a Hamiltonian path in $G_1$, (c) combine a Hamiltonian path in $G_1$ with edge $(v_1,v_2)$ in $G_2$, (d) a vertical separation of $C(m,n,k,l)$, (e) and (f) a Hamiltonian $(s,t)-$path in $C(m,n,k,l)$.}
\label{fig:s13a}
\end{figure}
\par
Subcase 4.2. $n-l=odd>1$, $[(d=odd$ and $c=even)$ or $(d=even$ and $c=odd)]$, and $[(s_x,t_x\leq d+k$ and $c>1)$ or $(d>1$ and $[(s_x,t_x>d+k)$ or $(d+1\leq s_x\leq d+k$ and $t_x>d+k)])]$. Since $l=even$ and $n-l=odd>1$, we have $n\geq 5$.
Let $d=odd$ and $c=even$. By symmetry, the result follows, if $d=even$ and $c=odd. $ Consider the following subcases. \par
 Subcase 4.2.1. $s_x,t_x\leq d+k$. This case is similar to Case 1.
 \par Subcase 4.2.2. $(s_x,t_x>d+k)$ or $(d+1\leq s_x\leq d+k$ and $t_x>d+k)$. Since $d=odd>1$, $c=even$, and $k\geq 1$, it follows that $m\geq 6$. This case is similar to Case 2. Since $l=even$, thus $G_2$ is even-sized. Moreover, since $C(m,n,k,l)$ is even-sized, we conclude that $G_1$ is even-sized. By Lemma \ref{Lemma:c0}, $(G_1,s,t)$ is color-compatible. Since $m\geq 6$, $n\geq 5$, $c=even$, and $n-l=odd>1$, it is enough to show that $(G_1,s,t)$ is not in conditions (F1), (F8), and (F9). $(G_1,s,t)$ is not in condition (F1), the proof is the same as Case 1. If $(G_1,s,t)$ satisfies condition (F8), then $(C(m,n,k,l),s,t)$ satisfies condition (F14), a contradiction. Therefore, $(G_1,s,t)$ is not in condition (F8). The condition (F9) holds, if $n-l=3$, $c\geq 4$, $s$ is black, and $s_x\leq d+k$. If this case occurs, then $(C(m,n,k,l),s,t)$ satisfies condition (F17), a contradiction. Thus, $ (G_1,s,t)$ does not satisfy condition (F9). Hence, $(G_1,s,t)$ is acceptable. The Hamiltonian path in $(C(m, n, k, l), s, t)$ is obtained similar
to Case 2.
  \par Case 5. $s_x\leq d$, $t_x>d+k$, and \\
  \indent \indent (a) $n=even$; or \\
  \indent \indent (b) $n=odd$ and $[(m=$even and $[(d=odd$, $c=odd$, and $n-l\leq 2)$, $(d=odd$ and $c=even$), $(d=even$ and \indent \indent $c=odd$), or $(c=even$ and $d=even)])$ or $(m=$odd and $[(k>1)$ or $(k=1$ and $[(c>2)$ or $(c=2$, $s\neq (d,l+1)$, \indent \indent or $t\neq (m,l+1))])])]$.\\This case is similar to Case 1 such that $s,p\in G_1$, $q,t\in G_2$, $p$ and $q$ are adjacent, and
 $$p=
  \begin{cases}
   (d+k,l+1);   &  if\ w\ and\ t\ have\ different\ colors,\ n-l=even,\ and\  [(c=2\ and \ t\neq(m,l+1))\ or\ (c> 2)]  \\
    (d+k,l+3);   &  if\ w\ and\ t\ have\ different\ colors,\ n-l=even>2,\ c=2,and \ t=(m,l+1) \\
 (d+k,l+2);& if \ w\ and\ t\ have\ the\ same\ color,\ n-l=odd>1,\ and\  [(c=2\ and \ t\neq(m,l+2))\ or\ (c>2)]\\
   (d+k,l+4);   &  if\ w\ and\ t\ have\ the\ same\ color,\ n-l=odd>3,\ and \ t=(m,l+2)\\
    (d+k,n);& otherwise
   \end{cases}$$
   where $w=(d+k+1,l+1)$. In the following, we prove that $(G_1,s,p)$ and $(G_2,q,t)$ are acceptable.
There are the following two subcases for $G_1$ and $G_2$.
 \par Subcase 5.1. $G_1$ and $G_2$ are even-sized. A simple check shows that $(G_1,s,p)$ and $(G_2,q,t)$ are color-compatible. Consider $(G_2,q,t)$. In this case, $G_2$ is even$\times$even, even$\times$odd, or odd$\times$even. The condition (F1) holds, if (i) $c=1$ and $[(t\neq (m,1)$ or $q\neq (d+k+1,n)]$. Since $(C(m,n,k,l),s,t)$ is acceptable, thus $t=(m,1)$. Moreover, Since $q=(d+k+1,n)$, clearly $q$ and $t$ are corner vertices in $G_2$; (ii) $c=2$ and $2\leq q_y=t_y\leq n-1$. We can easily see that $q_y\neq t_y$ or $q_y=t_y=n$; or (iii) $n=2$ and $d+k+2\leq q_x=t_x\leq m-1$. This case can not occur, because of $q_x=d+k+1$. Hence, $(G_2,q,t)$ is not in condition (F1). The condition (F2) occurs, when (i) $c=3$, $[(t$ and $w$ have different colors and $n-l=odd)$ or $(t$ and $w$ have the same color and $n-l=even)]$, and $q_y<t_y-1$. Since $q=(d+k+1,n)$, thus $(G_2,q,t)$ is not in condition (F2); or (ii) $n=3$, $t_x>d+k+1$, and $t$ is black (when $m$ is odd) or $s$ is white (when $m$ is even). Since $(C(m,n,k,l),s,t)$ is acceptable, the only case that occurs is $t=(d+k+1,1)$ or $t=(d+k+1,n)$. In this case, $q=(d+k+1,l+1)$ and hence $(G_2,q,t)$ is not in the condition (F2). So, $(G_2,q,t)$ is acceptable.
\par Now, consider $(G_1,s,p)$. Since $p_x=d+k$ and $s_x\leq d$, a simple check shows that $(G_1,s,p)$ is not in condition (F1) and (F3). The condition (F4) holds, if $k\times l=1$, $m^{'}-k>2$, $n-l>2$, $c=2$, $s=(d,l+1)$, and $t=(m,l+1)$. By the assumption, this is impossible, and hence $(G_1,s,p)$ is not in condition (F4).
 The condition (F7) occurs, when $d=2$, $k=1$, $n-l\geq 4$, $s=(1,l+1)$, and $p=(d+k,l+2)$. This is impossible, because of $p=(d+k,n)$. Hence, $(G_1,s,p)$ does not satisfy conditions (F7).
 The condition (F8) holds, if
(i) $n-l=2$, $d=3$ (or $n=3)$, and $s$ is black; (ii) $n-l=3$,
$d=2$, and $s$ is white. If these cases hold, then
$(C(m,n,k,l),s,t)$ satisfies condition (F17), a contradiction; (iii)
$m^{'}=3$, $d=2$, $s$ is black, and $p_y<s_y-1$, this is impossible,
because of $p_y=n$; or (iv) $m^{'}=3$, $d=2$, $c=2$, $s=(d,l+1)$, and
$t=(m,l+1)$, by the assumption, this case can not occur. So,
$(G_1,s,p)$ is not in conditions (F8). If $(G_1,s,p)$ satisfies
conditions (F6) and (F9), then $(C(m,n,k,l),s,t)$ satisfies
condition (F17), a contradiction. Therefore, it follows that
$(G_1,s,p)$ does not satisfy conditions (F6) and (F9). Hence,
$(G_1,s,p)$ is acceptable. It remains to show that
$(C(m,n,k,l),s,t)$ has a Hamiltonian path. In this case, the
Hamiltonian path in $(C(m, n, k, l), s, t)$ is obtained similar to
Subcase 4.1.2. Notice that, here, $(G_2,q,t)$ is a rectangular grid
graph and by Theorem \ref{Theorem:1a} it has a Hamiltonian path.
Thus, we construct  a Hamiltonian path in $(G_2,q,t)$ by the
algorithm in \cite{CST:AFAFCHPIM}. \par Now, let $m=odd$, $n=odd$,
$k=1$, $c=2$, $n-l>2$, $s=(d,l+1)$, and $t=(m,l+1)$. Let
$\{G_1,G_2\}$  be a $L-$shaped separation (type I) of $C(m,n,k,l)$
such that $G_1=L(m,n,k^{'},l^{'})$, $G_2=L(m^{'},n^{'},k,l)$,
$k^{'}=m-d$, $l^{'}=l+1$, $m^{'}=k^{'}$, $n^{'}=l^{'}$. Let $s,p\in
G_1$, $q,t\in G_2$, $q$ and $p$ are adjacent, and $p=(d+1,l+2)$.
Consider Fig. \ref{fig:s13a}(f). One easily check that $(G_1,s,p)$
and $(G_2,q,t)$ are acceptable. The Hamiltonian path in $(C(m, n, k,
l), s, t)$ is obtained similar to Subcase 4.1.2.
 \par Subcase 5.2. $G_1$ is odd-sized and $G_2$ is odd$\times$odd. We can easily see that $(G_1,s,p)$ and $(G_2,q,t)$ are color-compatible.
 Consider $(G_2,q,t)$. $(G_2,q,t)$ is not in conditions (F1) and (F2), the proof is similar to Subcase 5.1. Thus, $(G_2,q,t)$ is acceptable. Now, consider $(G_1,s,p)$. Since $G_1$ is odd-sized, it suffices to prove $(G_1,s,p)$ is not in conditions (F1), (F3), and (F5). Since $s_x\leq d$ and $p_x=d+k$, a simple check  that $(G_1,s,p)$ is not in conditions (F1) and (F3). The condition (F5) holds, if $n-l=2$ and $s=(d,n)$. If this case occurs, then $(G_1,s,p)$ satisfies condition (F13) (case (b)), a contradiction. Therefore, $(G_1,s,p)$ is not in condition (F5). Hence $(G_1,s,p)$ is acceptable. The Hamiltonian path in $(C(m, n, k, l), s, t)$ is
obtained similar to Subcase 5.1. Now, Lemma \ref{Lemma:c11} completes the proof.
\end{proof}
\begin{lem} \label{Lemma:c15}
Assume $(C(m,n,k,l),s,t)$ is an acceptable Hamiltonian path problem with $m\times n=$even$\times$odd, $c=odd$, $d=odd$, and $n-l=even\geq 4$.
 Then there is an acceptable separation for $(C(m,n,k,l),s,t)$ and it has a Hamiltonian path.
\end{lem}
\begin{proof} Note that, here, $C(m,n,k,l)$ is even-sized and $s$ and $t$ have different colors. For all the following cases, we prove that $(C(m,n,k,l),s,t)$ has an acceptable separation and show that it has a Hamiltonian path.
\begin{figure}[tb]
  \centering
  \includegraphics[scale=.95]{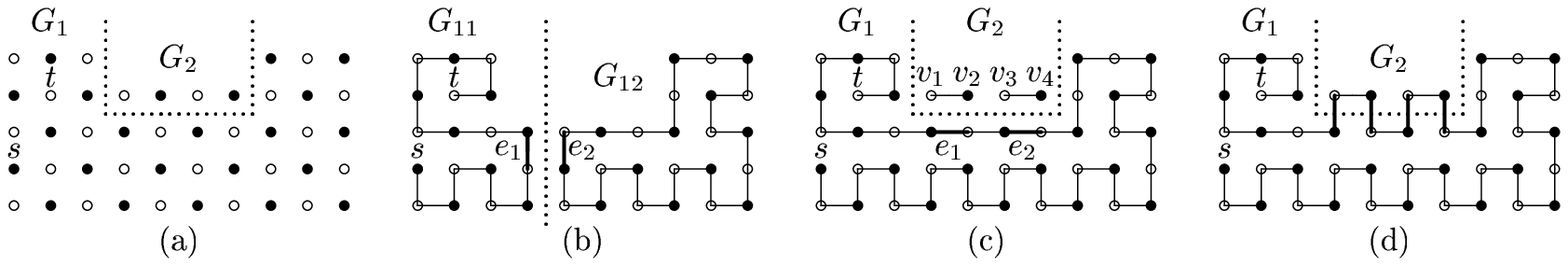}
  \caption[]%
 {\small (a) A $C-$shaped separation of $C(m,n,k,l)$, (b) and (c) a Hamiltonian $(s,t)-$path in $G_1$, and (d) a Hamiltonian $(s,t)-$path in $C(m,n,k,l)$.}
\label{fig:s13e}
\end{figure}
\par Case 1. $n-l\leq 6$ and $[(s_x,t_x\leq d$ and $c>1)$ or $(s_x,t_x>d+k$ and $d>1)]$.
 Assume $\{G_1,G_2\}$ is a $C-$shaped separation (type I) of $C(m,n,k,l)$ such that $G_1=C(m,n,k,l^{'})$, $G_2=R(m^{'},n^{'})$, $l^{'}=l+1$, $m^{'}=k$, $n^{'}=1$, and $s,t\in G_1$, as depicted in Fig. \ref{fig:s13e}(a). Because of $k=even$, $G_2$ is even-sized. Also, since $C(m,n,k,l)$ is even-sized, we conclude that $G_1$ is even-sized. By Lemma \ref{Lemma:c0}, $(G_1,s,t)$ is color-compatible. Since $n-l\leq 6$ and $l^{'}=l+1$, it follows that $n-l^{'}=odd\leq 5$. Moreover, since $n-l^{'}=odd$, $c=odd$, and $d=odd$, it suffices to prove $(G_1,s,t)$ is not in conditions (F1), (F3), (F11), and (F16). Let $s_x,t_x\leq d$, for case $s_x,t_x>d+k$, the proof is similar. $(G_1,s,t)$ is not in conditions (F1) and (F3), the proof is similar to Case 1 of Lemma \ref{Lemma:c11}. A simple check shows that $(G_1,s,t)$ is not in condition (F11). The condition (F16) holds, if $d=3$ and (i) $s_y\leq l$, $t_y> l$, and $s$ is black (here the role of $s$ and $t$ can be swapped), (ii) $n-l=4$, $s=(1,n-1)$, and $t_x>s_x$, or (iii) $s$ is black and $[(s_x=odd$, $t_y>s_y+1)$ or $(s_x=even$ and $t_y>s_y)]$ (here the role of $s$ and $t$ can be swapped). It is obvious that if these cases hold, then $(G_1,s,t)$ satisfies condition (F18), a contradiction. So, $(G_1,s,t)$ is not in condition (F16), and hence it is acceptable. In this case, $(G_1,s,t)$ is in Subcase 4.1.1 of Lemma \ref{Lemma:c11}. \par The Hamiltonian $(s,t)-$path is constructed as follows. First by Subcase 4.1.1 of Lemma \ref{Lemma:c11}, $G_1$ partitions into two subgraphs $G_{11}$ and $G_{12}$, and the Hamiltonian $(s,t)-$path in $G_{11}$ and Hamiltonian cycle in $G_{12}$ is constructed by the algorithm in \cite{991} and Lemma \ref{Lemma:6t1}, respectively. Notice that the pattern for constructing a Hamiltonian cycle in $G_{12}$ is shown in Fig. \ref{fig:s13e}(b). Then we combine the Hamiltonian path and cycle in $G_1$ using two parallel edges $e_1$ and $e_2$ as shown in Fig. \ref{fig:s13e}(b). Let four vertices $v_1,v_2,v_3$ and $v_4$ be in $G_2$ and let $P$ be a Hamiltonian path in $G_1$. Consider Fig. \ref{fig:s13e}(c). Clearly, there exist two edges $e_1$ and $e_2$ such that $e_1,e_2\in P$ are on boundary of $G_1$ facing $G_2$. By merging $(v_1,v_2)$ and $(v_3,v_4)$ to these edges, we obtain a Hamiltonian path for $(C(m,nk,l),s,t)$, as illustrated in Fig  \ref{fig:s13e}(d). When $k=2$ or $k>4$, a similar to the case $k=4$, the result follows.
\par Case 2. $n-l=4$ and $[(s_x\leq d$ and $t_x>d)$, $(d+1\leq s_x\leq d+k$ and $t_x>d+k)$, or $(d+1\leq s_x,t_x\leq d+k)]$.
\par Subcase 2.1. $(s_x\leq $ and $t_x>d)$ or $(d+1\leq s_x\leq d+k$ and $t_x>d+k)$. Let $s_x\leq d$ and $t_x>d$. By symmetry, the result follows, if $d+1\leq s_x\leq d+k$ and $t_x>d$. Consider the following subcases.
\par Subcase 2.1.1. $d=1$, $c>1$, and $s=(1,1)$. This case is similar to Case 2 of Lemma \ref{Lemma:c11}, where $s,p\in G_2$, $q,t\in G_1$, $p$ and $q$ are adjacent, and $p=(1,l)$. If $(G_2,s,p)$ is not acceptable, then $(C(m,n,k,l),s,t)$ satisfies condition (F11), a contradiction. Therefore, $(G_2,s,p)$ is acceptable. Moreover, since $s=(1,1)$ and $p=(1,l)$, a simple check shows that $(G_1,q,t)$ is acceptable. The Hamiltonian path in $(C(m, n, k, l), s, t)$ is
obtained similar to Case 5 of Lemma \ref{Lemma:c11} (Fig. \ref{fig:s13e1}(a)).
\par
\begin{figure}[tb]
  \centering
  \includegraphics[scale=.95]{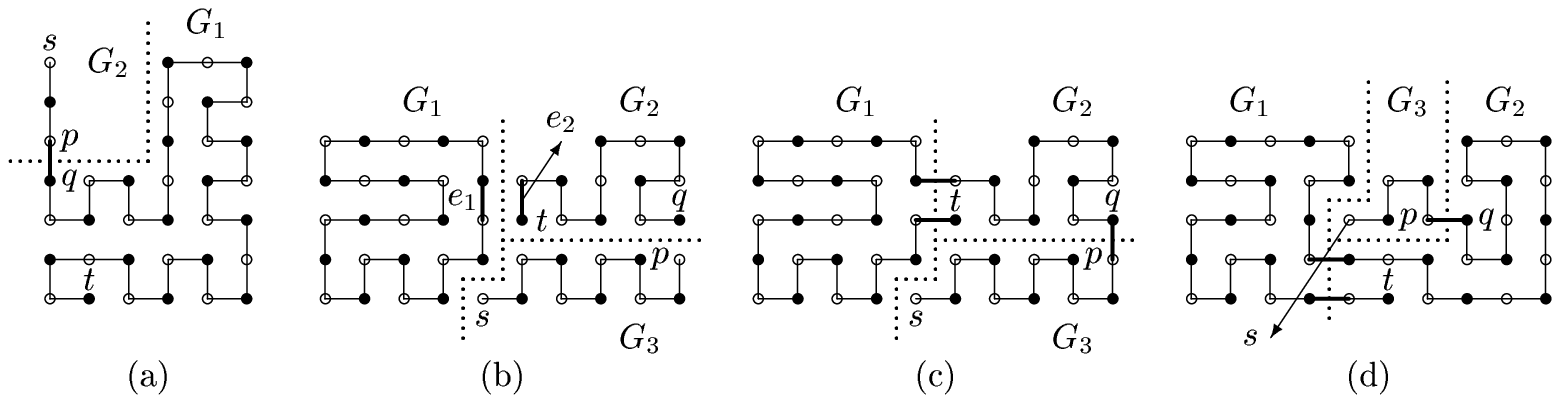}
  \caption[]%
 {\small (a) A Hamiltonian $(s,t)-$path in $C(m,n,k,l)$, (b) a $L-$shaped separation of $C(m,n,k,l)$, and (c) and (d) a Hamiltonian $(s,t)-$path in $C(m,n,k,l)$.}
\label{fig:s13e1}
\end{figure}
\par Subcase 2.1.2. $d,c>1$,
\par
Subcase 2.1.2.1. $s=(d,n)$, and $[t=(d+1,l+2)$ or $t=(d+2,l+1)]$.
Let $\{G_1,G_2,G_3\}$ be a $L-$shaped separation (type III) of
$C(m,n,k,l)$ such that $V(G_1)=\{1\leq x\leq d, 1\leq y\leq n-1$ and
$1\leq x\leq d-1, y=n\}$, $V(G_2)=\{d+1\leq x\leq m,1\leq y\leq
n-2\}$, $G_3=C(m,n,k,l)\backslash (G_1+G_2) $, $q,t\in G_2$, $s,p\in
G_3$, $p$ and $q$ are adjacent, and $p=(m,n-1)$. Consider Fig.
\ref{fig:s13e1}(b). A simple check shows that $(G_2,q,t)$ and
$(G_1,s,p)$ are acceptable. In order to build a
Hamiltonian$(s,t)-$path in $C(m,n,k,l)$, first we construct
Hamiltonian paths in $(G_2,q,t)$ and $(G_3,s,p)$ by the algorithm in
\cite{991}. Then we connect two vertices $p$ and $p$. Moreover,
since $G_1$ is even-sized, then it has a Hamiltonian cycle by Lemma
\ref{Lemma:6t1}. Finally, we combine Hamiltonian cycle in $G_1$ and
Hamiltonian $(s,t)-$path by two parallel edges. The full
construction of a Hamiltonian path in $(C(m,n,k,l)s,t)$ is
illustrated in Fig. \ref{fig:s13e1}(c). The pattern for
constructing a Hamiltonian cycle in $G_1$ is shown in Fig.
\ref{fig:s13e1}(b). It is easy to see that there exists at least one
edge for combining Hamiltonian cycle and path.
 \par
 Subcase 2.1.2.2. $s=(d,l+2)$, and $t=(d+1,n)$. This case is similar to Subcase 2.1.2.1, where  $V(G_1)=\{1\leq x\leq d, 1\leq y\leq l+1$ and $1\leq x\leq d-1, l+2\leq y\leq n\}$, $V(G_2)=\{d\leq x\leq m, n-1\leq y\leq n$ and $d+k+1\leq x\leq m, 1\leq y\leq n-2\}$, $p=(d+k,n-2)$, and $q=(d+k+1,n-2)$ (as shown Fig. \ref{fig:s13e1}(d)).
 \par
 Subcase 2.1.3. $[(d,c>1)$ or $(d=c=1)]$ and $[(s\neq (d,n)$, $t\neq (d+1,l+2)$ or $t\neq(d+2,l+1))$, $(s\neq (d,l+2)$ or $t\neq (d+1,n))]$. This case is similar to Case 5 of Lemma \ref{Lemma:c11}, where $G_1=R(m^{'},n)$, $G_2=L(m-m^{'},n,k,l)$, $m^{'}=d$, and $p=(m^{'},n)$ if $s\neq (m^{'},n)$ or $t\neq (m^{'}+1,n)$; otherwise $p=(m^{'},l+2)$. Since $d=odd$ and $n=odd$, it follows that $G_1$ is odd$\times$odd with white majority color, $G_2$ is odd-sized with black majority color, and $p$ is white. Clearly, $(G_1,s,p)$ and $(G_2,q,t)$ are color-compatible. Consider $(G_1,s,p)$. It is easy to check that $(G_1,s,p)$ is not in conditions (F1)
 and (F2). Now, consider $(G_2,q,t)$. Since $n-l=4$ and $c=odd$, it is enough to show that $(G_2,q,t)$ is not in conditions (F1) and (F3). Since $q_x=d+1$ and $t_x\geq d+1$, a simple check shows that $(G_2,q,t)$ is not in conditions (F1) and (F3). Hence $(G_2,q,t)$ is acceptable. The Hamiltonian path in $(C(m, n, k, l), s, t)$ is
obtained similar to Case 5 of Lemma \ref{Lemma:c11}.\par
\begin{figure}[tb]
  \centering
  \includegraphics[scale=.95]{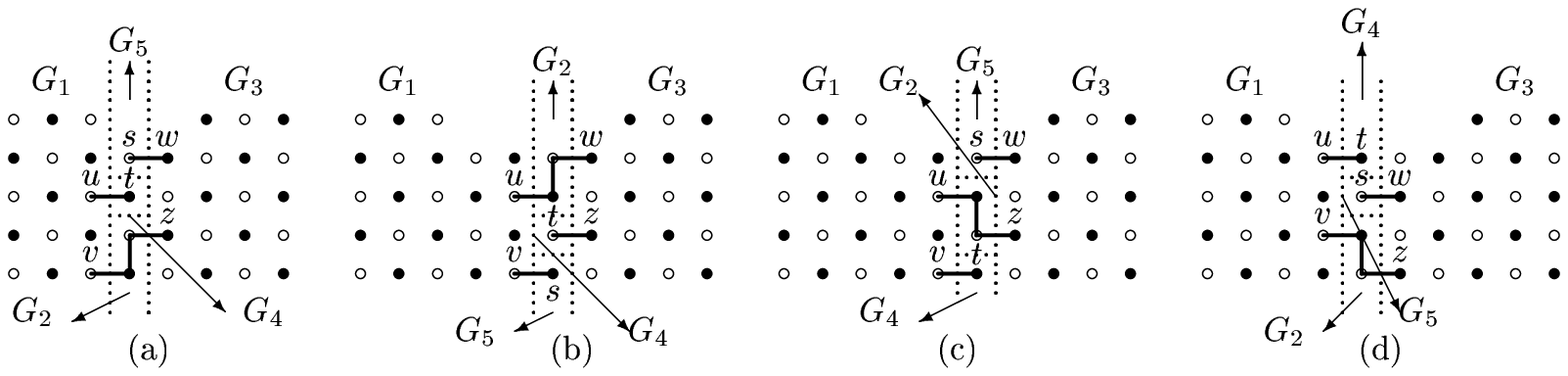}
  \caption[]%
 {\small A Hamiltonian $(s,t)-$path in $C(m,n,k,l)$.}
\label{fig:s13e2}
\end{figure}
 \begin{figure}[tb]
  \centering
  \includegraphics[scale=.95]{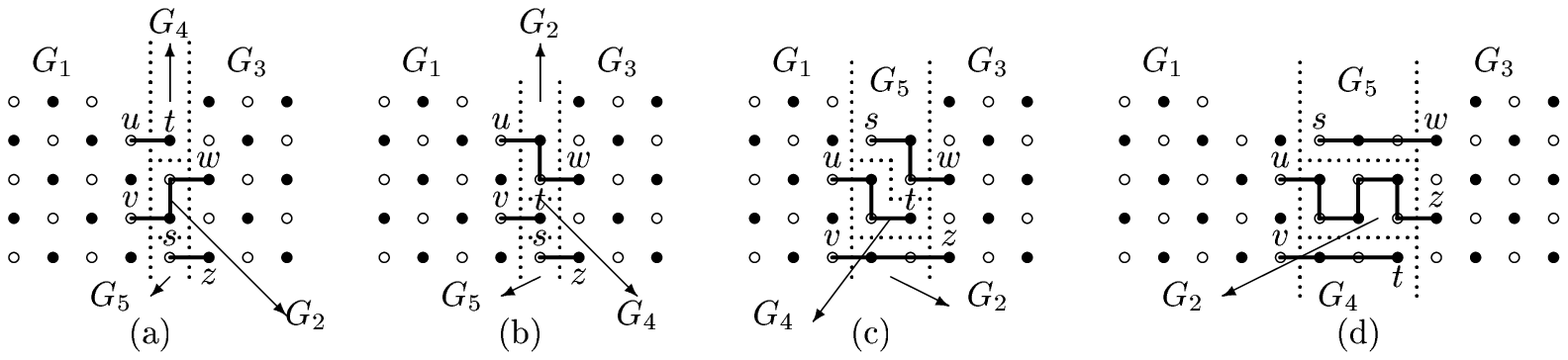}
  \caption[]%
 {\small A Hamiltonian $(s,t)-$path in $C(m,n,k,l)$.}
\label{fig:s13e3}
\end{figure}
\begin{figure}[tb]
  \centering
  \includegraphics[scale=.95]{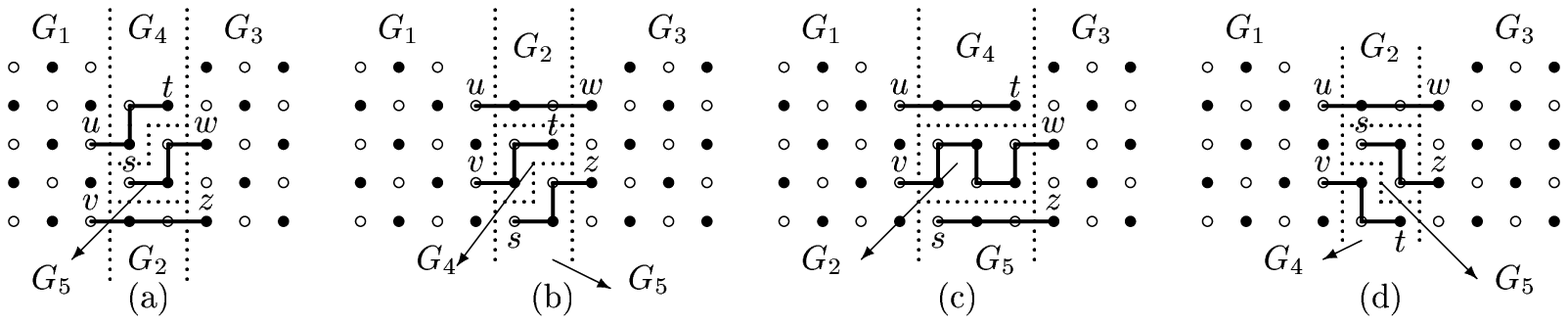}
  \caption[]%
 {\small A Hamiltonian $(s,t)-$path in $C(m,n,k,l)$.}
\label{fig:s13e4}
\end{figure}
\par Subcase 2.2. $d+1\leq s_x,t_x\leq d+k$ and $d,c>1$.
\par Subcase 2.2.1. $s_x=t_x$. Let $\{G_1,G_2,G_3,G_4,G_5\}$ be a $C-$shaped separation (type III) of $C(m,n,k,l)$, as shown in Fig. \ref{fig:s13e2}, \ref{fig:s13e3}(a), and \ref{fig:s13e3}(b). The patterns in Fig. \ref{fig:s13e2}, \ref{fig:s13e3}(a), and \ref{fig:s13e3}(b) can be used for finding a Hamiltonian $(s,t)-$path for any values of $d$, $c$, $l$, and $k$. Notice that in Fig. \ref{fig:s13e2}(a)-(c) $s_x=even$, and in Fig. \ref{fig:s13e2}(d), \ref{fig:s13e3}(a), and \ref{fig:s13e3}(b) $s_x=odd$.
In this case, $G_1$ is a rectangular (or $L-$shaped) grid subgraph, where $s_x=d+1$ (or $s_x>d+1)$ and also $G_3$ is a rectangular (or $L-$shaped) grid subgraph, where $s_x=d+k$ (or $s_x<d+k)$. The Hamiltonian path in $(G_1,u,v)$ and $(G_2.w,z)$ constructed by algorithm in \cite{CST:AFAFCHPIM} or \cite{991}.
\par Subcase 2.2.2. $s_x\neq t_x$.
\par Subcase 2.2.2.1. $(s_x=odd$ and $[(s=(s_x,n)$ and $[t=(s_x+1,l+2)$ or $t=(s_x+2,l+1)])$ or $(s=(s_x,l+2)$ and $t=(s_x+1,n))])$ or $(s_x=even$ and $[(s=(s_x,l+3)$ and $t=(s_x+1, l+1))$ or $(s=(s_x,l+1)$ and $[(t=(s_x+1,l+3)$ or $t=(s_x+2,n)])$. This case is similar to Subcase 2.2.1. The patterns in Fig. \ref{fig:s13e3}(c), \ref{fig:s13e3}(d), and \ref{fig:s13e4} can be used for finding a Hamiltonian $(s,t)-$path between for any values of $d$, $c$, $l$, and $k$.
\par Subcase 2.2.2.2. Other possible cases. This case is similar to Subcase 2.1.3, where $G_1=L(m^{'},n,k^{'},l)$, $G_2=L(m-m^{'},n,k^{''},l)$, $m^{'}=s_x$, $k^{'}=m^{'}-d$, and $k^{''}=k-k^{'}$. Let $s_x=even$, then $p=(m^{'},l+1)$ if $s\neq (m^{'},l+1)$ or $t\neq (m^{'}+1,l+1)$; otherwise $p=(m^{'},l+3)$. Now, let $s_x=odd$, then $p$ is defined similar to Subcase 2.1.3; where $m^{'}=s_x$.
\begin{figure}[tb]
  \centering
  \includegraphics[scale=.95]{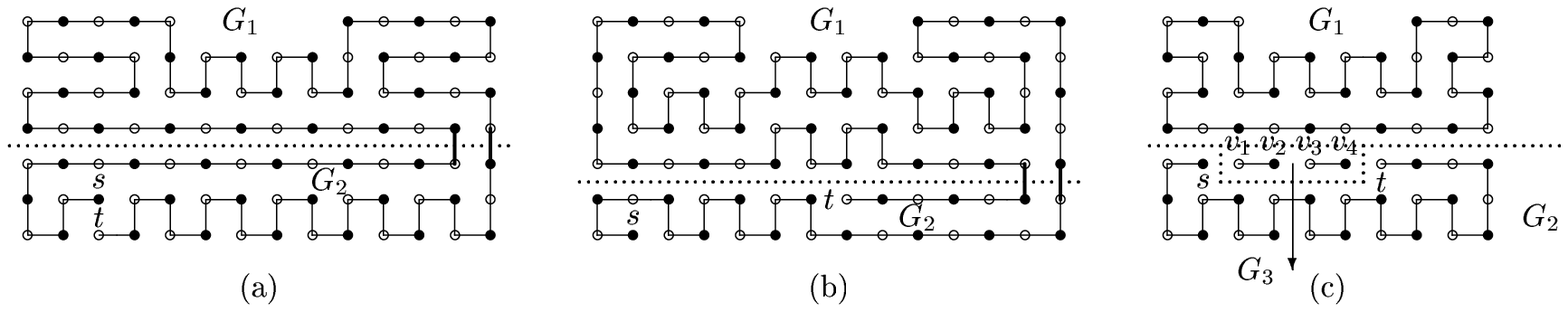}
  \caption[]%
 {\small (a) and (b) A Hamiltonian $(s,t)-$path in $C(m,n,k,l)$, (c) a $C-$shaped separation type (IV) of $C(m,n,k,l)$.}
\label{fig:s13e5}
\end{figure}
\par Case 3. $n-l= 6$ and $[(s_x\leq d$ and $t_x>d)$, $(d+1\leq s_x\leq d+k$, $t_x>d+k$, and $d>1)$, or $(d+1\leq s_x,t_x\leq d+k$ and $c,d>1)]$. Since $n-l= 6$ and $l\geq 1$, thus $n\geq 7$.
\par
Subcase 3.1. $d,c>1$ and $s_y,t_y>l+3$. Since $d=odd>1$, $c=odd>1$, and $k=even$, we have $m\geq 8$.
\par Subcase 3.1.1. $(s_x=t_x)$, $(s $ is black, $s_y=odd$, and $t_x=s_x+1)$, or $(s$ is white and $t_x>s_x)$.
Let $\{G_1,G_2\}$ be a horizontal separation of $C(m,n,k,l)$ such that $G_1=C(m,n^{'},k,l)$, $G_2=R(m,n-n^{'})$, $n^{'}=l+3$ and $s,t\in G_2$. Since $n$ is odd, $n-l= 6$, and $n^{'}=l+3$, it follows that $n-n^{'}= 3$ and $n^{'}=even$. Moreover, since $m=even$, we conclude that $G_2$ is even$\times$odd. By Lemma \ref{Lemma:c0}, $(G_2,s,t)$ is color-compatible. Since $m\geq 8$ and $n-n^{'}=3$, it suffices to prove that $(G_2,s,t)$ is not in condition (F2). The condition (F2) holds, if $s$ is black and $t_x>s_x+1$. This is impossible, because we assume that $t_x=s_x+1$. Thus $(G_2,s,t)$ is not in condition (F2), and hence $(G_2,s,t)$ is acceptable. The Hamiltonian path in $(C(m, n, k, l), s, t)$ is
obtained similar to Case 1 of Lemma \ref{Lemma:c11}. Since $m\geq 8$, thus there is at least one edge for combining Hamiltonian cycle and path. In this case, the pattern for constructing a Hamiltonian cycle in $G_1$ is shown in Fig. \ref{fig:s13e5}(a).
\par
Subcase 3.1.2. $s_y,t_y>l+4$, $s$ is black, and $[(s_y=even$ and $t_x>s_x)$ or $(s_y=odd$ and $t_x>s_x+1)]$. This case is similar to Subcase 3.1.1, where $n^{'}=l+4$. Since $n=odd$, $n-l=6$, and $n^{'}=l+4$, it follows that $n-n^{'}=2$ and $n^{'}=odd$. Clearly, $G_1$ is even-sized and $G_2$ is even$\times$ even. By Lemma \ref{Lemma:c0}, $(G_2,s,t)$ is color-compatible. Since $G_2=$even$\times$even, it is enough to prove that $(G_2,s,t)$ is not in condition (F1). The condition (F1) occurs, when $2\leq s_x=t_x\leq m-1$. Since $t_x>s_x$, thus $s_x\neq t_x$, and hence $(G_2,s,t)$ is not in condition (F1). Therefore, $(G_2,s,t)$ is acceptable. The Hamiltonian path in $(C(m, n, k, l), s, t)$ is
obtained similar to Subcase 3.1.1. In this case, the pattern for constructing a Hamiltonian cycle in $G_1$ is shown in Fig. \ref{fig:s13e5}(b).
\par Subcase 3.1.3. $s_y=l+4$, $t_y>l+4$, $t_x>s_x+1$, and $s$ is black. This case is the same as Subcase 3.1.2, where $s,p\in G_1$, $q,t\in G_2$, $p$ and $q$ are adjacent, and $p=(1,n^{'})$. From Subcase 3.1.2, we know that $G_1$ and $G_2$ is even-sized. Since $l=odd$, we have $n^{'}=odd$. Moreover, since $p=(1,n^{'})$, it is clear that $p$ is white. Hence, $(G_1,s,p)$ and $(G_2,q,t)$ are color-compatible. $(G_2,q,t)$ is not in conditions (F1) and (F2), the proof is similar to Subcase 3.1.2. Consider $(G_1,s,p)$. Since $n^{'}-l=4$ and $d,c\geq 3$, it suffices to prove that $(G_2,s,p)$ is not in condition (F18). The condition (F18) holds, if $p_x<s_x$ and $p$ is black. Since $p$ is white, it is clear that $(G_1,s,p)$ is not in condition (F18). Hence $(G_1,s,p)$ is acceptable. In this case, $(G_1,s,p)$ is in Case 1 or 2. The Hamiltonian path in $(C(m, n, k, l), s, t)$ is
obtained similar to Case 5 of Lemma \ref{Lemma:c11}. Here, if $s_y>l+4$ and $t_y=l+4$, then the role of $p$ and $q$ can be swapped (that is, $s,p\in G_2$ and $q,t\in G_1)$.
\begin{figure}[tb]
  \centering
  \includegraphics[scale=.95]{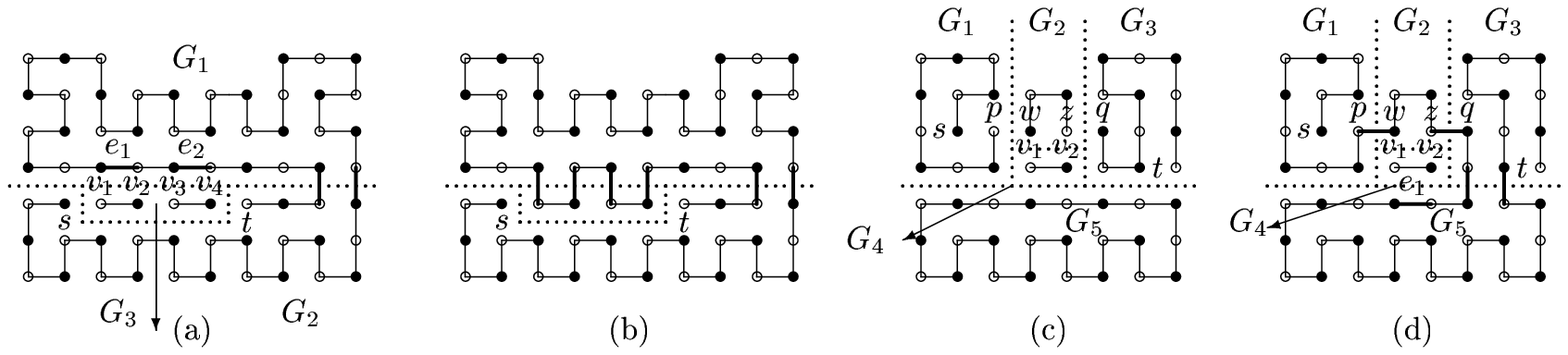}
  \caption[]%
 {\small (a) and (b) Combining a Hamiltonian $(s,t)-$path in $G_2$ and a Hamiltonian cycle in $G_1$, (b) a Hamiltonian $(s,t)-$path in $C(m,n,k,l)$, (c) $C-$shaped separation type (II) of $C(m,n,k,l)$, (d) combining Hamiltonian paths in $G_1$, $G_2$, and $G_3$ and a Hamiltonian cycle in $G_5$.}
\label{fig:s13e5a}
\end{figure}
\par Subcase 3.1.4. $s_y=t_y=l+4$, $s$ is black, and $t_x>s_x+1$. Let $\{G_1,G_2,G_3\}$ be a $C-$shaped separation (type IV) of $C(m,n,k,l)$ such that $G_1=C(m,n^{'},k,l)$, $G_2=C(m,n-n^{'},k^{'},l^{'})$, $G_3=R(m^{'},n^{''})$, $n^{'}=l+3$, $k^{'}=t_x-s_x-1$, $l^{'}=1$, $m^{'}=k^{'}$, $n^{''}=l^{'}$, and $s,t\in G_2$. Consider Fig. \ref{fig:s13e5}(c). Since $m$ is even, $s$ is black, and $t$ is white, thus $t_x=odd$ and $s_x=even$, and hence $d^{'}$, $c^{'}$, and $k^{'}$ are even. Clearly, $G_1$, $G_2$, and $G_3$ are even-sized. By Lemma \ref{Lemma:c0}, $(G_1,s,t)$ is color-compatible. Since $n-n^{'}-l^{'}=2$, $c^{'},d^{'}=even$, $s_x\leq d^{'}$, and $t_x>d^{'}+k^{'}$, it suffices to prove $(G_2,s,t)$ is not in condition (F17). The condition (F17) holds, if $s_y=l+5$ or $t_y=l+5$. Since $t_y=s_y=l+4$, this is impossible, and hence $(G_2,s,t)$ is acceptable. In this case, $(G_2,s,t)$ is in Subcase 5.1 of Lemma \ref{Lemma:c11}. For constructing a Hamiltonian $(s,t)-$path, first combine a Hamiltonian path in $G_2$ and a Hamiltonian cycle in $G_1$, this path is called $P_1$, as shown in Fig \ref{fig:s13e5a}(a). the pattern for constructing a Hamiltonian cycle in $G_1$ is shown in Fig. \ref{fig:s13e5}(c). Notice that since $m\geq 8$, thus there exists at least one edge for combining Hamiltonian cycle and path. Let four vertices $v_1,v_2,v_3$ and $v_4$ be in $G_3$. Consider Fig. \ref{fig:s13e5a}(a). Clearly, there exist two edges $e_1$ and $e_2$ such that $e_1,e_2\in P_1$ are on boundary of $G_1$ facing $G_3$. By merging $(v_1,v_2)$ and $(v_3,v_4)$ to these edges, we obtain a Hamiltonian path for $(C(m,nk,l),s,t)$, as illustrated in Fig  \ref{fig:s13e5a}(b). When $k=2$ or $k>4$, a similar to the case $k=4$, the result follows.
\par Subcase 3.2. $s_y,t_y\leq l+3$.
\par
Subcase 3.2.1. $s$ is white and $[(s_x\neq t_x)$ or $(d+1\leq s_x=t_x\leq d+k$ and $[(s_y$ (or $t_y)>l+3)$ or $(s_y$ (or $t_y)<l+2)])]$. Since $k=even$ and $d,c\geq 1$, we have $m\geq 4$. This case is similar to Subcase 3.1.2, where $s,t\in G_1$. Since $n-l=6$ and $n^{'}=l+4$, it follows that $n^{'}-l=4$. Moreover, since $d$ and $c$ are odd and $n^{'}-l=4$, it suffices to prove $(G_1,s,t)$ is not in condition (F1), (F3), (F11), and (F18). $(G_1,s,t)$ is not in conditions (F1) and (F3), the proof is similar to Case 1. A simple check shows that $(G_1,s,t)$ is not in conditions (F11) and (F18).
 Therefore, $(G_1,s,t)$ is acceptable. In this case, $(G_1,s,t)$ is in Case 2. The Hamiltonian path in $(C(m, n, k, l), s, t)$ is
obtained similar to Case 1 of Lemma \ref{Lemma:c11}. Notice that, because of $m\geq 4$, there is at least one edge for combining Hamiltonian cycle and path.
\par Subcase 3.2.2. $(d+1\leq s_x=t_x\leq d+k$, and $[s_y$ (or $t_y)=l+2$ and $t_y$ (or $s_y)=l+3])$, or $(s$ is black, $s_y=even$, and $t_x=s_x+1)$. Note that, in this case, $d,c>1$. This case is the same as Subcase 3.1.1, where $s,t\in G_1$. One can check that $(G_1,s,t)$ is acceptable. In this case, $(G_1,s,t)$ is in Case 1 of Lemma \ref{Lemma:c11}. The Hamiltonian path in $(C(m, n, k, l), s, t)$ is
obtained similar to Subcase 3.1.1.
\par Subcase 3.2.3. $d,c>1$, $s$ is black, and $[(s_y=even$ and $t_x>s_x+1)$ or $(s_y=odd$ and $t_x>s_x)]$.
\par Subcase 3.2.3.1. $s_x\leq d$ and $t_x>d+k$. Let $\{G_1,G_2,G_3,G_4, G_5\}$ be a $C-$shaped separation (type II) of $C(m,n,k,l)$ such that $G_1=R(m^{'},n^{'})$, $G_2=R(m^{''},n^{''})$, $G_3=R(m-m^{'}-m^{''},n^{'})$, $G_4=R(m^{''}, 1)$, and $G_5=C(m,n,k,l)\backslash (G_1+G_2+G_3+G_4)$, where $V(G_1)=\{1\leq x\leq x^{'}, 1\leq y\leq l+3\}$, $V(G_2)=\{x^{'}+1\leq x\leq x^{''}, 1\leq y\leq l+2\}$, $V(G_3)=\{x^{''}+1\leq x\leq m,1\leq y\leq l+3\}$, $V(G_4)=\{x^{'}+1\leq x\leq x^{''}, y=l+3\}$, $x^{'}=d$, and $x^{''}=d+k$. Assume that $s,p\in G_1$, $w,z\in G_2$, $q,t\in G_3$ such that $w$ and $p$, and $q$ and $z$ are adjacent, $p=(x^{'},l+2)$, and $q=(x^{''}+1,l+2)$. Consider Fig. \ref{fig:s13e5a}(c). It is clear that $(G_1,s,p)$, $(G_2,w,z)$, and $(G_3,q,t)$ are color-compatible. Consider $(G_1,s,p)$ and $(G_3,q,t)$. Since $n^{'}=l+3$ and $l=odd$, it follows that $n^{'}=even\geq 4$. Moreover, since $d,c=odd>1$ and $n^{'}\geq 4$, $(G_1,s,p)$ and $(G_3,q,t)$ are not in condition (F1). The condition (F2) holds, if $(d=3$ and $s_y\leq l)$ or $(c=3$ and $t_y\leq l)$. If this case holds, then $(C(m,n,k,l),s,t)$ satisfies condition (F18), a contradiction. So, $(G_1,s,p)$ and $(G_3,q,t)$ are acceptable. Consider $(G_2,w,z)$. Since $k$ and $n^{''}$ are even, it is clear that $G_2$ is even$\times$even and $G_4$ is even-sized. $(G_2,w,z)$ is not in condition (F2). Moreover, since $w_x=x^{'}+1$ and $z_x=x^{''}$, clearly $(G_2,w,z)$ is not in condition (F1). Therefore, $(G_2,w,z)$ is acceptable. \par Because $(G_1,s,p)$, $(G_2,w,z)$, and $(G_3,q,t)$ are acceptable, by Theorem \ref{Theorem:1a} they have Hamiltonian paths. So, we construct a Hamiltonian path in $(G_1,s,p)$, $(G_2,w,z)$, and $(G_3,q,t)$ by the algorithm in \cite{CST:AFAFCHPIM}. Then we connect vertices $p$, $w$, $z$, and $q$. Furthermore, since $G_5$ is even-sized rectangular grid subgraph, it has a Hamiltonian cycle by Lemma \ref{Lemma:1m}. Then combine Hamiltonian cycle
and path using two parallel edges; see Fig. \ref{fig:s13e5a}(d). Notice that, since $d,c>1$, there exists at least one edge for combining Hamiltonian
cycle and path. Let two vertices $v_1$ and $v_2$ be in $G_4$ and $P$ be a Hamiltonian $(s,t)-$path. Obviously, there exists an edge $e_1$ such that $e_1\in P$ are on boundary of $G_5$ facing $G_4$. By merging $(v_1,v_2)$ to this edge, we obtain a Hamiltonian path for $(C(m,n,k,l),s,t)$, as illustrated in Fig  \ref{fig:s13e7}(a). When $|G_4|>2$, a similar to the case $|G_4|=2$, the result follows.
\begin{figure}[tb]
  \centering
  \includegraphics[scale=1]{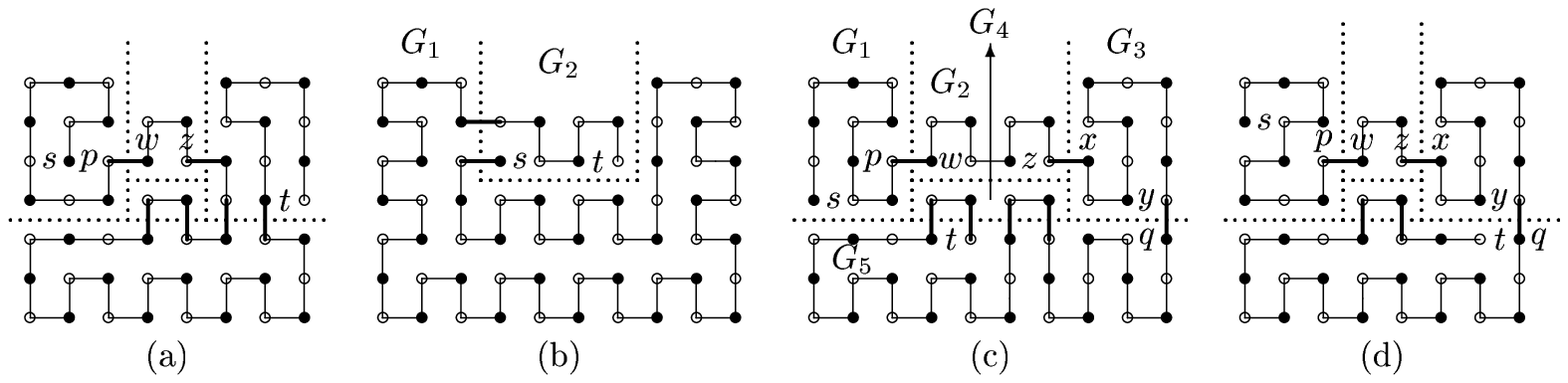}
  \caption[]%
 {\small A Hamiltonian cycle in $(C(m,n,k,l),s,t)$.}
\label{fig:s13e7}
\end{figure}
\par Subcase 3.2.3.2. $d+1\leq s_x,t_x\leq d+k$.
\par Subcase 3.2.3.2.1. $s_y,t_y\leq l+2$. This case is the same as Case 1, where $l^{'}=l+2$, $m^{'}=k$, $n^{'}=l^{'}-l$, and $s,t\in G_2$. Consider Fig. \ref{fig:s13e7}(b). A simple check shows that $(G_2,s,t)$ is acceptable. The Hamiltonian path in $(C(m, n, k, l), s, t)$ is
obtained similar to Case 1 of Lemma \ref{Lemma:c11}. In this case, the pattern for constructing a Hamiltonian cycle in $G_1$ is shown in Fig. \ref{fig:s13e7}(b). It is obvious that there is at least one edge for combining Hamiltonian
cycle and path.
\par Subcase 3.2.3.2.2. $s$ is black and $[(s_y\leq l+2$ and $t_y>l+2)$ or $(t_y\leq l+2$ and $s_y>l+2)]$.
\par Subcase 3.2.3.2.2.1. $s_y\leq l+2$ and $t_y>l+2$. This case is similar to Subcase 3.2.3.2.1, where $s,p\in G_2$, $q,t\in G_1$, $p=(d+k,l+2)$, $p$ and $q$ are adjacent, and $q=(d+k,l+3)$. Since $d+k=odd$ and $l+2=odd$, it follows that $p=(d+k,l+2)$ is white and $q$ is black. Clearly, $(G_2,s,p)$ and $(G_1,q,t)$ are color-compatible. We can easily see that $(G_2,s,p)$ is not in conditions (F1) and (F2). Consider $(G_1,q,t)$. Since $n-l=6$ and $l^{'}=l+2$, it follows that $n-l^{'}=4$. So, it suffices to prove $(G_1,q,t)$ is not in condition (F18). Since $t_x<q_x$ and $t$ is white, $(G_1,q,t)$ does not satisfy condition (F18), and hence $(G_1,q,t)$ is acceptable. In this case, $(G_1,q,t)$ is in Case 2. The Hamiltonian path in $(C(m, n, k, l), s, t)$ is
obtained similar to Case 5 of Lemma \ref{Lemma:c11}.
\par Subcase 3.2.3.2.2.2. $s_y>l+2$ and $t_y\leq l+2$. This case is the same as Subcase 3.2.3.2.2.1, where $s,p\in G_1$, $q,t\in G_2$, $p=(d+1,l+3)$, and $q=(d+1,l+2)$. By the same argument as in proof Subcase 3.2.3.2.2.1, we obtain $(G_1,s,p)$ and $(G_2,q,t)$ are acceptable. The Hamiltonian path in $(C(m, n, k, l), s, t)$ is
obtained similar to Case 1 of Lemma \ref{Lemma:c11}.
\par Subcase 3.2.3.2.3. $s_y,t_y=l+3$. This case is similar to Subcase 3.2.3.1, where $x^{'}=s_x$ and $x^{''}=t_x-1$. We can easily see that $(G_1,s,p)$ and $(G_2,q,t)$ are acceptable. Notice that, in this case, $G_1$ and $G_3$ are $L-$shaped grid graphs. So, we construct a Hamiltonian path in $(G_1,s,p)$ and $(G_2,q,t)$ by the algorithm in \cite{991}.
\par Subcase 3.2.4. $(s_x\leq d$ and $d+1\leq t_x\leq d+k)$ or $(d+1\leq s_x\leq d+k$ and $t_x>d+k)$. Let $t_x>d+k$ and $d+1\leq s_x\leq d+k$. By symmetry, the result follows, if $s_x\leq d$ and $d+1\leq t_x\leq d+k$.
\par Subcase 3.2.4.1. $s_y\leq l+2$. This case is similar to Subcase 3.2.3.2.2.1. Notice that, in this case,
the condition (F18) holds, if $c=3$ and $t_y\leq l$. If this case occurs, then $(C(m,n,k,l),s,t)$ satisfies condition (F18), a contradiction. So, $(G_1,q,t)$ is acceptable.
\par Subcase 3.2.4.2. $s_y=l+3$. This case is similar to Subcase 3.2.3.1, where $x^{'}=s_x$ and $x^{''}=d+k$. Here, $G_1$ is a $L-$shaped grid subgraph, hence we construct a Hamiltonian path in $(G_1,s,p)$ by the algorithm in \cite{991}. Note that, in this case, if $s_x=d+k$, then $G_2=\emptyset$ and $G_4=\emptyset$.
\par Subcase 3.3. $s_y\leq l+3$ and $t_y>l+3$ (or $t_y\leq l+3$ and $s_y>l+3)$.
\par Subcase 3.3.1. $t_y>l+4$. This case is similar to Subcase 3.1.3, where $p=(1,n^{'})$ if $s$ is black; otherwise $p=(m,n^{'})$. From Subcase 3.1.3, we know that $G_1$ is even-sized, $G_2$ is even$\times$even, and $n^{'}=odd$. Since $n^{'}=odd$ and $m=even$, we conclude that $p=(1,n^{'})$ is white and $p=(m,n^{'})$ is black. Thus, It is clear that $(G_1,s,p)$ and $(G_2,q,t)$ are color-compatible. $(G_2,q,t)$ is not in conditions (F2), the proof is the same as Subcase 3.1.3. The condition (F1) occurs, when $2\leq q_x=t_x\leq m-1$. Since $q_x=1$ or $m$, thus $(G_2,q,t)$ is not in condition (F1). So, $(G_2,q,t)$ is acceptable. Now, consider $(G_1,s,p)$. The condition (F1) holds, if $d=1$ (resp. $c=1)$, $s_y\leq l+1$, and $s\neq (1,1)$ (resp. $t_y\leq l+1$ and $t\neq (m,1))$. Since $C(m,n,k,l),s,t)$ is acceptable, it follows that $s=(1,1)$ (resp. $t=(m,1))$. Therefore, $(G_1,s,p)$ is not in condition (F1). The condition (F3) occurs, when $s_y>l$ and $(d=1$ or $c=1)$. Clearly, If this case holds, then $(C(m,n,k,l),s,t)$ satisfies condition (F3), a contradiction. Therefore, $(G_1,s,p)$ s not in condition (F3). A simple check shows that $(G_1,s,p)$ is not in condition (F11). The condition (F18) holds, if (i) $d=3$, $s_y\leq l$, and $s$ is black, (ii) $c=3$, $s_y\leq l$, and $s$ is white; if these cases occur, then $(C(m,n,k,l),s,t)$ satisfies condition (F18), a contradiction; or (iii) $s_y>l$, $s$ is black, and $p_x>s_x$. We can easily check that this case can not ocuur. Thus $(G_1,s,p)$ is not in condition (F18). Hence, $(G_1,s,p)$ is acceptable. In this case, $(G_1,s,p)$ is in Case 1 or 2. The Hamiltonian path in $(C(m, n, k, l), s, t)$ is
obtained similar to Case 5 of Lemma \ref{Lemma:c11}. Notice that if $s_y>l+4$ and $t_y<l+4$, the the role of $p$ and $q$ can be swapped.\par
Subcase 3.3.2. $(t_y=l+4$, $s_y\leq l+3$, $s$ is black, and $t$ is white) or $(s_y=l+4$, $t_y\leq l+3$, $s$ is black, and $t$ is white). Notice that, in this case, $d,c>1$. Let $t_y=l+4$. By symmetry, the result follows, if $s_y=l+4.$\par
\par Subcase 3.3.2.1. $t_x=s_x+1$ or $t_x=s_x+2$. This case is similar to Subcase 3.1.3, where $n^{'}=l+3$ and $p=(t_x-1,n^{'})$. We can easily check that $(G_1,s,p)$ and $(G_2,q,t)$ are acceptable. In this case, $(G_1,s,p)$ is in Case 1 of Lemma \ref{Lemma:c11}. The Hamiltonian path in $(C(m, n, k, l), s, t)$ is
obtained similar to Case 5 of Lemma \ref{Lemma:c11}. 
\par
Subcase 3.3.2.2. $t_x>s_x+2$. This case is similar to Subcase 3.2.3.1, where $s,p\in G_1$, $w,z\in G_2$, $x,y\in G_3$, and $q,t\in G_4$. In this case, $x^{'}=d$ if $s_x\leq d$; otherwise $x^{'}=s_x$ and $x^{''}=d+k$. Let $p$ and $w$, $z$ and $x$, and $q$ and $y$ are adjacent such that $p=(x^{'},l+2)$, $x=(x^{''}+1,l+2)$, and $q=(m,l+4)$. $(G_1,s,p)$, $(G_2,w,z)$, and $(G_3,x,y)$ are acceptable, the proof is similar to Subcases 3.2.3.1 and 3.2.3.2.3. A simple check shows that $(G_5, q,t)$ is acceptable. In this case, $G_1$ is a rectangular grid graph if $s_x\leq d$; otherwise $G_1$ is a $L-$shaped grid graph. The Hamiltonian path in $(C(mn,k,l),s,t)$ is obtained similar to Subcase 3.2.3.1 (as shown in Fig. \ref{fig:s13e7}(c) and \ref{fig:s13e7}(d)). Notice that, here, first we connect vertices $p$ and $w$, $z$ and $x$, and $q$ and $y$. In this case, the patterns for constructing a Hamiltonian path in $G_5$ is shown in Fig. \ref{fig:s13e7}(c) and \ref{fig:s13e7}(d).
\par
Subcase 3.3.3. $(t_y=l+4$, $t$ is black and $s$ is white) or $(s_y=l+4$, $s$ is white and $t$ is black). This case is similar to Subcase 3.2.1.\par
\par Case 4. $n-l>6$.\par
Subcase 4.1.
$s_y,t_y>l+5$. This case is similar to Subcase 3.1.2. By the same argument as in proof Subcase 3.1.2, we drive $(G_2,s,t)$ is acceptable. The Hamiltonian path in $(C(mn,k,l),s,t)$ is obtained similar to Subcase 3.1.2.
\par Subcase 4.2. $s_y,t_y\leq l+5$. This case is similar to Subcase 3.1.1, where $n^{'}=l+5$ and $s,t\in G_1$. Since $m$ is even, it follows that $G_2$ is even-sized. Moreover, since $C(m,n,k,l)$ is even-sized, we conclude that $G_1$ is even-sized. By Lemma \ref{Lemma:c0}, $(G_1,s,t)$ is color-compatible. Because of $l=odd$ and $n^{'}=l+5$, we have $n^{'}=even\geq 6$ and $n^{'}-l=odd\geq 5$. Furthermore, since $n^{'}-l=odd\geq 5$ and $d,c\geq 1$, it suffices to prove $(G_1,s,t)$ is not in conditions (F1), (F3), (F11), and (F17). $(G_1,s,t)$ is not in conditions (F1) and (F3), the proof is the same as Case 1. Clearly, since $n^{'}-l\geq 5$, a simple check shows that $(G_1,s,t)$ is not in condition (F11). It is obvious that if $(G_1,s,t)$ satisfies condition (F17), then $(C(m,n,k,l),s,t)$ satisfies condition (F18), a contradiction. Therefore, $(G_1,s,t)$ is not in condition (F17). Hence, $(G_1,s,t)$ is acceptable.
In this case, $(G_1,s,t)$ is in Case 1 or 2, or Subcase 5.1 of Lemma \ref{Lemma:c11}. The Hamiltonian path in $(C(m, n, k, l), s, t)$ is
obtained similar to Subcase 3.2.1.
\par Subcase 4.3. $s_y\leq l+5$, and $t_y>l+5$. This case is similar to Subcase 4.2, where $s,p\in G_1$, $q,t\in G_2$, $p$ and $q$ are adjacent, and $p=(1,n^{'})$ if $s$ is white; otherwise $p=(m,n^{'})$. Since $n^{'}=even$ and $m=even$, thus $p=(1,n^{'})$ is black and $p=(m,n^{'})$ is white. Hence, $(G_1,s,p)$ and $(G_2,q,t)$ are color-compatible. In this case, $G_1$ is even-sized and $G_2$ is even$\times$ odd. $(G_1,s,p)$ is not in conditions (F1), (F3), and (F11), the proof is the same as Subcase 3.3.1. $(G_1,s,p)$ is not in condition (F17) , the proof is like to Subcase 4.2. Hence $(G_1,s,p)$ is acceptable. Now, consider $(G_2,q,t)$. Since $n-l>6$ and $n^{'}=l+5$, it follows $n-n^{'}\geq 3$. Moreover, since $m\geq 4$, $n-n^{'}\geq 3$, it is sufficient to show that $(G_2,q,t)$ is not in condition (F1). The condition (F2) holds, if $n-n^{'}=3$ and $[(t$ is white and $q_x<t_x-1)$ or $(t$ is black and $q_x>t_x+1)]$. Since  $q_x=1$, where $t$ is black, or $q_x=m$, where $t$ is white, it is clear that $(G_2,q,t)$ is not in condition (F2). Therefore $(G_2,q,t)$ is acceptable. In this case, $(G_1,s,p)$ is in Case 1 or 2, or Subcase 5.1 of Lemma \ref{Lemma:c11}. The Hamiltonian path in $(C(m, n, k, l), s, t)$ is
obtained similar to Case 5 of Lemma \ref{Lemma:c11}. Notice that if $s_y>l+5$ and $t_y<l+5$, then the role of $p$ and $q$ can be swapped (i.e., $s,p\in G_2$ and $q,t\in G_1)$.
This finishes the proof.
\end{proof}
\begin{lem} \label{Lemma:c13}
Suppose that $(C(m,n,k,l),s,t)$ is an acceptable Hamiltonian path problem. Assume $C(m,n,k,l)$ is odd-sized. Then there is an acceptable separation for $(C(m,n,k,l),s,t)$ and it has a Hamiltonian path.
\end{lem}
 \begin{proof}
 Let $m\times n= $odd$\times$odd, then $k\times l$ is even$\times$even, even$\times$odd, or odd$\times$even and two vertices $s$ and $t$ are white. Let $m=even$, then $k\times l=$odd$\times$odd, $d=even$ and $c=odd$ (or $d=odd$ and $c=even$), and two vertices $s$ and $t$ are black if $d=even$; otherwise $s$ and $t$ are white. Now, let $m\times n= $odd$\times$even, then $k\times l$ is odd$\times$odd, $d$ and $c$ are $even$ (or $odd$), and  two vertices $s$ and $t$ are black if $d$ and $c$ are even; otherwise $s$ and $t$ are white. Notice that, here for $n=odd$, $l=odd$, and $[(d=even$ and $c=odd)$ or $(d=odd$ and $c=even)]$, we only consider the case $d=odd$, and $c=even$. By symmetry, the result follows, if $d=even$ and $c=odd$. Consider the following cases. We will show that there is an acceptable separation for $(C(m,n,k,l),s,t)$ and it has a Hamiltonian path.
 \par Case 1. $(n=odd$ and $m=odd)$ or $(n=even)$.
\par Subcase 1.1. $n=odd$, $l=even$, $n-l>1$, and $[(s_x,t_x\leq d+k$ and $c>1)$ or $(d>1$ and $[(s_x,t_x>d+k)$ or $(d+1\leq s_x\leq d+k$ and $t_x>d+k)])]$. Let $s_x,t_x\leq d+k$.
By symmetry, the result follows, if $(s_x,t_x>d+k)$ or $(d+1\leq s_x\leq d+k$ and $t_x>d+k)$.
 This case is similar to Case 2 of Lemma \ref{Lemma:c11}. Since $l$ is even, thus $G_2$ is even-sized and $n-l=odd$. Moreover, since $C(m,n,k,l)$ is odd-sized, we conclude that $G_1$ is odd-sized. Hence by Lemma \ref{Lemma:c0}, $(G_1,s,t)$ is color-compatible. Since $G_1$ is odd-sized and $n-l\geq 3$, it suffices to prove that $(G_1,s,t)$ is not in conditions (F1), (F3), and (F5). The condition (F1) holds, if (i) $d=1$, $s_y$ (or $t_y)\leq l+1$ and $s$ (or $t)\neq (1,1)$; (ii) $d=2$, $s_y,t_y\leq l+1$, and $|t_y-s_y|=1$; (iii) $n-l=2$, $s_x,t_x\geq d$, $t_x-s_x=1$. Clearly, if these cases hold, then $(C(m,n,k,l),s,t)$ satisfies condition (F1), a contradiction. Therefore, $(G_1,s,t)$ is not in condition (F1). $(G_1,s,t)$ is not in condition (F3), the proof is similar to Case 1 of Lemma \ref{Lemma:c11}. The condition (F5) holds, if $m-k^{'}=2$ and $[(s_y,t_y\leq l)$ or $(s_y\ ($or$\ t_y)\leq l$ and $t\ ($or$\ s)=(1,l+1))]$. If this case holds, then $(C(m,n,k,l),s,t)$ satisfies condition (F14) or (F17), a contradiction. Therefore, it follows that $(G_1,s,t)$ does not satisfy condition (F5), and hence it is acceptable. The Hamiltonian path in $(C(m, n, k, l), s, t)$ is
obtained similar to Case 1 of Lemma \ref{Lemma:c11}.
 \par   Subcase 1.2. $n-l>1$ and $[(n=odd$, $l=odd$, $d=odd$, $c=even$, $s_x,t_x\leq d+k$, and $[(n-l=2)$ or $(n-l\geq 4$ and $s\neq (d+k-1,l+1)$ or $t\neq (d+k,l+2))])$ or ($n=even$ and $[(c>1$ and $s_x,t_x\leq d+k)$ or $(d>1$ and $[(s_x,t_x>d+k)$ or $(d+1\leq s_x\leq d+k$ and $t_x>d+k)])])]$. Notice that, in this case, $n-l=odd$ if $n=even$; otherwise $n-l=even$. Let $s_x,t_x\leq d+k$.
By symmetry, the result follows, if $(s_x,t_x>d+k)$ or $(d+1\leq s_x\leq d+k$ and $t_x>d+k)$. This case is similar to Case 1 of Lemma \ref{Lemma:c11}. Since $n=even$ or $c=even$ implies $G_2$ is even-sized. Moreover, since $C(m,n,k,l)$ is odd-sized, we conclude that $G_1$ is odd-sized. By Lemma \ref{Lemma:c0}, $(G_1,s,t)$ is color-compatible. In the following, we show that $(G_1,s,t)$ is not in conditions (F1), (F3), and (F5). $ (G_1,s,t)$
 is not in condition (F3), the proof is the same as Case 1 of Lemma \ref{Lemma:c11}. The condition (F1) holds, if $n-l\geq 4$, $s=(d+k-1,l+1)$, and $t=(d+k,l+2)$. This is impossible, because we assume that $s\neq (d+k-1,l+1)$ or $t\neq (d+k,l+2)$. $(G_1,s,t)$ is not in condition (F5), the proof is the same as Subcase 1.1. So, $(G_1,s,t)$ is acceptable. The Hamiltonian path in $(C(m, n, k, l), s, t)$ is
obtained similar to Case 1 of Lemma \ref{Lemma:c11}.  Now, let $n=odd$, $n-l\geq 4$, $s=(d+k-1,l+1)$, and $t=(d+k,l+2)$. This case is similar to Subcase 4.1.2 of Lemma \ref{Lemma:c11}, where $m^{'}=s_x$ and $p=(m^{'},n-1)$; see Fig. \ref{fig:s13e8}(a). Clearly, $(G_1,s,p)$ and $(G_2,q,t)$ are acceptable. The Hamiltonian path in $(C(m, n, k, l), s, t)$ is
obtained similar to Subcase 4.1.2 of Lemma \ref{Lemma:c11}.
\begin{figure}[tb]
  \centering
  \includegraphics[scale=.96]{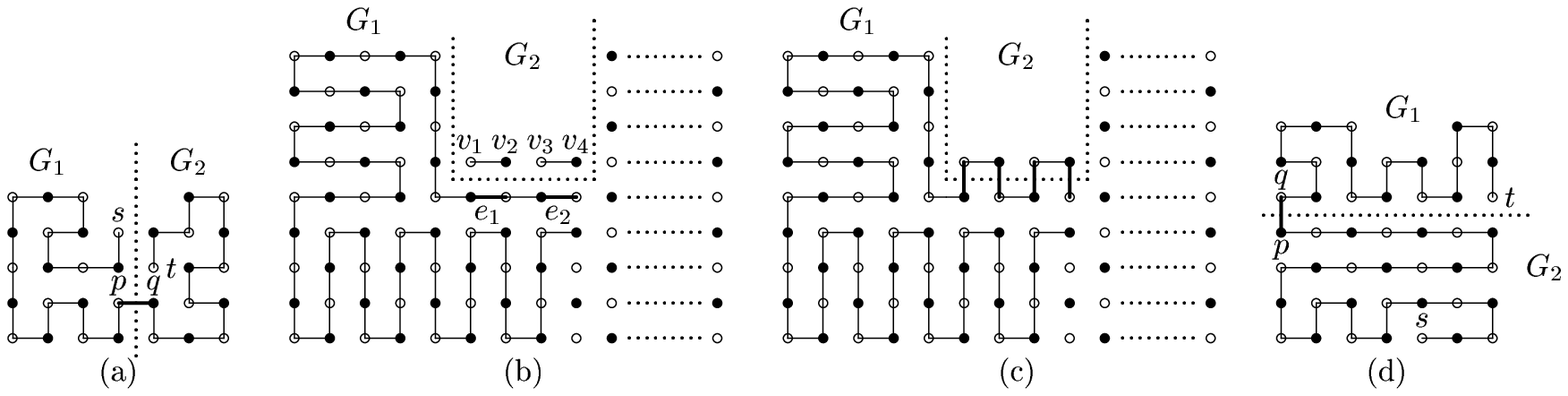}
  \caption[]%
 {\small A Hamiltonian $(s,t)-$path in $C(m,n,k,l)$.}
\label{fig:s13e8}
\end{figure}
\par Subcase 1.3. $n=odd$, $l=odd$, $n-l>1$, $d=odd>1$, $c=even$, and $s_x,t_x>d+k$. Notice that, in this case, $n-l=even\geq 4$. This case is the same as Case 1 of Lemma \ref{Lemma:c15}. By the same argument as in proof Case 1 of Lemma \ref{Lemma:c15}, $G_2$ is even-sized and $(G_1,s,t)$ is color-compatible. Since $l^{'}=l+1$, $l=odd$, and $n-l=even\geq 4$, we have $n-l^{'}=odd\geq 3$. Moreover, since $d>1$, $n-l=odd\geq 3$, and $c=even$, it is enough to show that $(G_1,s,t)$ is not in conditions (F1) and (F14). The condition (F1) or (F14) holds, if $c=2$ and $s_y,t_y\leq l+2$. Clearly, if this case occurs then $(C(m,n,k,l),s,t)$ satisfies condition (F1) or (F14), a contradiction. Thus, $(G_1,s,t)$ is not in condition (F1) or (F14). Hence, $(G_1,s,t)$ is acceptable. In this case, $(G_1,s,t)$ is in Subcase 1.1. The Hamiltonian path in $(C(m, n, k, l), s, t)$ is
obtained similar to Case 1 of Lemma \ref{Lemma:c15}. Notice that, in this case, we can always construct a Hamiltonian path $P$ in $G_1$ that contains a subpath $P_1$, as shown in Fig. \ref{fig:s13e8}(b). Let four vertices $v_1,v_2,v_3$ and $v_4$ be in $G_2$. Consider Fig. \ref{fig:s13e8}(b). Clearly, there exist two edges $e_1$ and $e_2$ such that $e_1,e_2\in P_1$ are on boundary of $G_1$ facing $G_2$. By merging $(v_1,v_2)$ and $(v_3,v_4)$ to these edges, we obtain a Hamiltonian path for $(C(m,n,k,l),s,t)$, as illustrated in Fig \ref{fig:s13e5a}(c). When $k=2$ or $k>4$, a similar to the case $k=4$, the result follows. 
\par Subcase 1.4.
$(n=even,$ $s_x\leq d$, and $t_x>d+k)$ or $(n=odd$ and $[(l=even$, $s_x\leq d$, and $t_x>d+k)$ or $(l=odd$, $s_x\leq d+k$, and $t_x>d+k$)]).
This case is similar to Case 5 of Lemma \ref{Lemma:c11}, where
$$p=
  \begin{cases}
  (d+k,n);   &  if\ (n=even)\ or\ (n=odd\ and\ [(l=even)\ or\\
    &(l=odd\ and\ [(n-l=2)\ or\ (n-l>2\ and\ s\neq (d+k,n))])]) \\
   (d+k,l+2);   &  if\ n=odd,\ l=odd,\ n-l>2\,\ s=(d+k,n),\ and\ [(c>2)\ or \ (c=2\ and\  t\neq (m,l+2)]\\
   \end{cases}$$ Consider the following Subcases.
\par Subcase 1.4.1. $G_1$ is odd-sized and $G _2$ is even-sized.  We can easily check that $(G_1,s,p)$ and $(G_2,q,t)$ are color-compatible. Consider $(G_2,q,t)$. $(G_2,q,t)$ is not in condition (F1), the proof is the same as Subcase 5.1 of Lemma \ref{Lemma:c11}. The condition (F2) holds, if (i) $n=3$ and $t$ is black; this is impossible because $t$ and $s$ are white, or (ii) $c=3$, $q_y<t_y-1$, this case does not occur because of $q=(d+k+1,n)$. Thus $(G_2,q,t)$ is not in condition (F2). Hence, $(G_2,q,t)$ is acceptable. Now, consider $(G_1,s,p)$. The condition (F1) holds, if (i) $d=1$, $s_y\leq l+1$, and $s\neq (1,1)$, clearly this is impossible; (ii) $n-l=1$ and $s_x\geq d$; (iii) $n-l=2$ and $s=(d+k-1,n-1)$; if these cases occur, then $(C(m,n,k,l),s,t)$ satisfies condition (F1) or (F13), a contradiction; or (iv) $n-l>2$, $s=(d+k-1,l+1)$, and $p=(d+k,l+2)$, this is impossible because of in this case $p=(d+k,n)$. Hence, $(G_1,s,p)$ is not in condition (F1). The condition (F3) holds, if $n-l=1$ and $s_x,p_x\leq d$. By the assumption, this is impossible. The condition (F5) occurs, when $(d=2$ and $s_y,p_y\leq l+1)$ or $(n-l=2, s_y,p_y\geq n-1$, and $s_x,p_x>d)$. It is obvious that if this case holds, then $(C(m,n,k,l),s,t)$ satisfies condition (F17), a contradiction. Therefore, $(G_1,s,p)$ is not in condition (F5). Hence, $(G_1,s,p)$ is acceptable. The Hamiltonian path in $(C(m, n, k, l), s, t)$ is
obtained similar to Case 5 of Lemma \ref{Lemma:c11}.\par Now, let $n=odd$, $n-l>2$, $l=odd$, $c=2$, $s=(d+k,n)$, and $t=(m,l+2)$. This case is similar to Subcase 3.1.3 of Lemma \ref{Lemma:c15}, where $n^{'}=l+2$, $s,p\in G_2$, $q,t\in G_1$, and $p=(1,n^{'}+1)$. Consider Fig. \ref{fig:s13e8}(d). Clearly $(G_1,q,t)$ and $(G_2,s,p)$ are acceptable. In this case, $(G_1,q,t)$ is in Subcase 1.4.1. The Hamiltonian path in $(C(m, n, k, l), s, t)$ is
obtained similar to 3.1.3 of Lemma \ref{Lemma:c15}.
 \par Subcase 1.4.2. $G_1$ is even-sized and $G_2$ is odd$\times$odd. A simple check shows that $(G_1,s,p)$ and $(G_2,q,t)$ are color-compatible. $(G_2,q,t)$ is not in conditions (F1) and (F2), the proof is the same as Subcase 5.2 of Lemma \ref{Lemma:c11}. Now, consider $(G_1,s,p)$. $(G_1,s,p)$ is not in conditions (F1) and (F3), the proof is similar to Subcase 1.4.1. The condition (F4) holds, if $k\times l=1$. Since $k\times l>1$, thus $(G_1,s,p)$ is not in condition (F4). The conditions (F6), (F8), and (F9) hold, if $t$ is black. This is impossible, because $s$ and $t$ are white. The condition (F7) holds, if $p_x=d$. Since $p_x=d+k$, thus $(G_1,s,p)$ does not satisfy condition (F7). Therefore, $(G_1,s,p)$ and $(G_2,q,t)$ are acceptable. The Hamiltonian path in $(C(m, n, k, l), s, t)$ is
obtained similar to Case 5 of Lemma \ref{Lemma:c11}.
\par Case 2. $n=odd$ and $m=even$. In this case, $n-l=even$.
\par Subcase 2.1. $[(s_x,t_x\leq d+k)$, $(s_x,t_x>d+k$, $s_y,t_y>l$, and $d>1)$, or $(s_x\leq d+k$, $t_x>d+k$, and $t_y>l)]$ and $[(c>2$ and $l>1)$ or $(c=2$ and $t\neq (m,l+1))].$ This case is the same as Case 2 of Lemma \ref{Lemma:c11}. Since $c=even$, thus $G_2$ is even-sized. Moreover, since $C(m,n,k,l)$ is odd-sized, then $G_1$ is odd-sized. By Lemma \ref{Lemma:c0}, $(G_1,s,t)$ is color-compatible. Now, we show that $(G_1,s,t)$ is not in conditions (F1), (F3), and (F5). $(G_1,s,t)$ is not in condition (F1), the proof is the same as Subcase 1.1. $(G_1,s,t)$ is not in condition (F3), the proof is the same as Case 1 of Lemma \ref{Lemma:c11}. The condition (F5) holds, if $n-l=2$, $s_y,t_y\geq n-1$, and $s_x,t_x>d$. If this case holds, then $(C(m,n,k,l),s,t)$ satisfies condition (F13) or (F17), a contradiction. Therefore, $(G_1,s,t)$ is not in condition (F5). Hence, $(G_1,s,t)$ is acceptable. The Hamiltonian path in $(C(m, n, k, l), s, t)$ is
obtained similar to Case 2 of Lemma \ref{Lemma:c11}.
\par Subcase 2.2. $c=2$ and $t$ (or $s)=(m,l+1)$. Let $t=(m,l+1)$. Notice that, here, $n-l>2$. Consider the following subcases.
\begin{figure}[tb]
  \centering
  \includegraphics[scale=1]{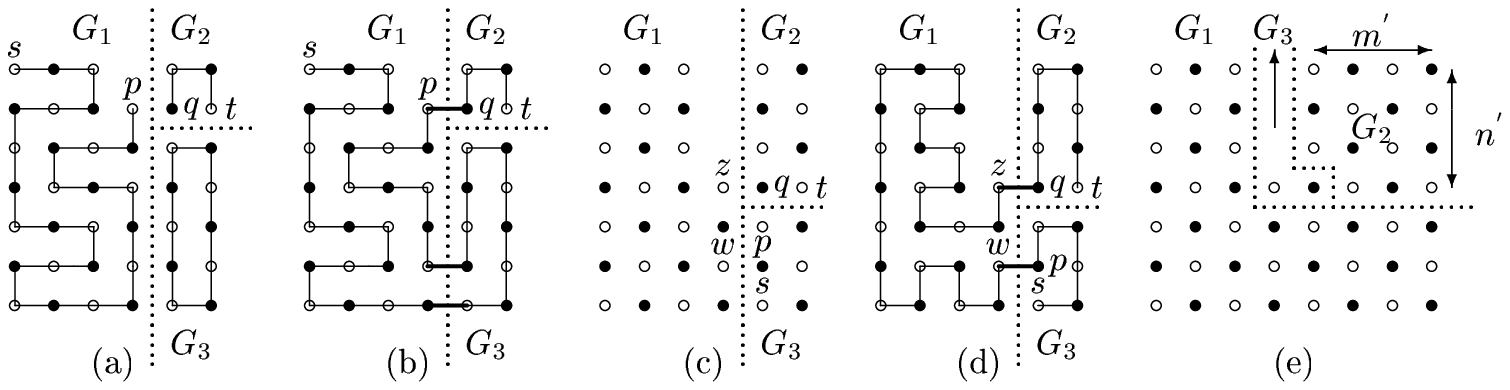}
  \caption[]%
 {\small (a) A $L-$shaped separation of $C(m,n,k,l)$, (b) a Hamiltonian $(s,t)-$path in $C(m,n,k,l)$, (c) a $L-$shaped separation of $C(m,n,k,l)$, (d) (b) a Hamiltonian $(s,t)-$path in $C(m,n,k,l)$, and (e) a $C-$shaped separation (type V) of $C(m,n,k,l)$ .}
\label{fig:s13e9}
\end{figure}
\par Subcase 2.2.1. $s_x\leq d$ and $s\neq (d+k,l+1)$. Let $\{G_1,G_2,G_3\}$ be a $L-$shaped separation (type II) of $C(m,n,k,l)$ such that $G_1=L(m^{'},n,k,l)$, $G_2=R(m-m^{'},n^{'})$, $G_3=R(m-m^{'},n-n^{'})$, $m^{'}=d+k$, and $n^{'}=l+1$. Let $s,p\in G_1$, $q,t\in G_2$, $q$ and $p$ are adjacent, and $p=(d+k,l+1)$. Consider Fig. \ref{fig:s13e9}(a). A simple check shows that $(G_1,s,p)$ and $(G_2,q,t)$ are acceptable. The Hamiltonian path in $(C(m, n, k, l), s, t)$ is
obtained similar to Subcase 3.2.3.1 of Lemma \ref{Lemma:c15}, as shown in Fig. \ref{fig:s13e9}(b). Notice that, here, $G_1$ is a $L-$shaped grid graph, thus we construct a Hamiltonian path in $(G_1,s,p)$ by the algorithm in \cite{991}. 
Obviously, since $n-n^{'}\geq 3$ there is at least one edge for combining Hamiltonian
cycle and path.
\par Subcase 2.2.2. $s=(d+k,l+1)$. This case is similar to Case 5 of Lemma \ref{Lemma:c11}, where $p=(d+k,n-1)$. A simple check shows that $(G_1,s,p)$ and $(G_2,q,t)$ are acceptable. The Hamiltonian path in $(C(m, n, k, l), s, t)$ is
obtained similar to Case 5 of Lemma \ref{Lemma:c11}.
\par Subcase 2.2.3. $s_x,t_x>d+k$. This case is similar to Subcase 2.2.1, where $s,p\in G_3$, $q,t\in G_2$, $w,z\in G_1$, $q=(m-1,l+1)$, and $p=(m-1,n-1)$ if $s\neq (m,n-1)$; otherwise $(m-1,n-3)$. Assume that $z$ and $q$, and $w$ and $p$ are adjacent. Consider Fig. \ref{fig:s13e9}(c). It is clear that $(G_1,z,w)$ and $(G_2,q,t)$ are acceptable. Consider $(G_3,s,p)$. Obviously, $(G_3,s,p)$ is color-compatible. The condition (F1) holds, if $p_y=s_y=n-1$. Since $s\neq (m,n-1)$, thus $(G_3,s,p)$ is not in condition (F1). The condition (F2) occurs, when $n-n^{'}=3$ and $s_y=p_y=n-1$. If this case occurs, then, $(C(m,n,k,l),s,t)$ satisfies condition (F15), a contradiction. Therefore, $(G_3,s,p)$ is not in condition (F2). Hence, $(G_1,s,p)$ is acceptable. The Hamiltonian path in $(C(m, n, k, l), s, t)$ is
obtained similar to Subcase 3.2.3.1 of Lemma \ref{Lemma:c15}, as shown in Fig. \ref{fig:s13e9}(d). Notice that, here, $G_1$ is a $L-$shaped grid graph, thus we construct a Hamiltonian path in $(G_1,s,p)$ by the algorithm in \cite{991}.
\par Subcase 2.3. $s_x,t_x>d+k$, $s_y,t_y\leq l+1$, and $l>1$. Notice that, in this case, $c\geq 4$. Let $\{G_1,G_2, G_3\}$ be a $L-$shaped separation (type V) of $C(m,n,k,l)$ such that $V(G_2)=\{d+k+1\leq x\leq m,\ 1\leq y\leq l$ and $d+k+2\leq x\leq m, y=l+1\}$, $V(G_3)=\{d+1\leq x\leq d+k+1,y=l+1\}$, $G_1=L(m,n,k^{'},l^{'})$, $k^{'}=m-d$, $l^{'}=l+1$, and $s,t\in G_2$. Consider Fig. \ref{fig:s13e9}(e). Clearly, $G_2$ is odd-sized and $G_1$ and $G_3$ are even-sized. Here, $G_2$ is a $L-$shaped grid subgraph $L(m^{'},n^{'},k^{'},l^{'})$, where $m^{'}=m-d-k$, $n^{'}=l+1$, $k^{'}=1$, and $l^{'}=1$. By Lemma \ref{Lemma:c0}, $(G_2,s,t)$
 is color-compatible. Since $l>1$ and $n^{'}=l+1$, we have that $n^{'}\geq 4$. Also, since $m^{'}-k^{'}\geq 3$ and $n^{'}-l^{'}\geq 3$, it is clear that $(G_1,s,t)$ is not in conditions (F1), (F3), and (F5), and hence $(G_2,s,t)$ is acceptable. The Hamiltonian path in $(C(m, n, k, l), s, t)$ is
obtained similar to Subcase 3.1.4 of Lemma \ref{Lemma:c15}; see Fig. \ref{fig:s13e10}(a). Note that, here, since $G_1$ is an even-sized $L-$shaped grid subgraph, and by Lemma \ref{Lemma:6t1} has a Hamiltonian cycle. The pattern for constructing a Hamiltonian cycle in $G_1$ is shown in Fig. \ref{fig:s13e10}(a). Moreover since $m^{'}-k^{'}\geq 3$, there is at least one edge for combining Hamiltonian cycle and path.
\par Subcase 2.4. $t_y\leq l$, $t_x>d+k$, and $[(s_x\leq d+k)$ or $(s_x>d+k$ and $s_y>l+1)]$. This case is similar to Subcase 2.1, where $s,p\in G_1$, $q,t\in G_2$, $q$ and $p$ are adjacent, and $p=(m,l+1)$.  From Subcase 2.1, we know that $G_1$ is even-sized and $G_2$ is odd-sized. Since $m=even$ and $l+1=even$, it follows that $p$ is white. Therefore, $(G_1,s,p)$ and $(G_2,q,t)$ are color-compatible. Consider $(G_2,q,t)$. The condition (F1) holds, if $c=2$ and $t_y=q_y<l$. This is impossible, because of $q_y=l$. Thus, $(G_2,q,t)$ is not in condition (F1). The condition (F2) occurs, when $l=3$ and $q_x<t_x-1$. Since $q_x=m$, thus $(G_2,q,t)$ is not in condition (F2). So, $(G_2,q,t)$ is acceptable. Consider $(G_1,s,p)$. The condition (F1) holds, if (i) $d=1$ and $2\leq s_y\leq l+1$; (ii) $n-l=2$ and $s=(m-1, n)$, clearly if these condition hold, then $(C(m,n,k,l),s,t)$ satisfies condition (F1), a contradiction. Therefore, $(G_1,s,p)$ is not in condition (F1). A simple check shows that $(G_1,s,p)$ is not in condition (F3). $(G_1,s,p)$ is not in condition (F5), the proof is the same as Subcase 2.1. Hence, $(G_1,s,p)$ is acceptable.
 The Hamiltonian path in $(C(m, n, k, l), s, t)$ is
obtained similar to Case 5 of Lemma \ref{Lemma:c11}. Here if $s_y\leq l$, $s_x,t_x>d+k$, and $t_y>l+1$, then the role of $p$ and $q$ can be swapped.
\begin{figure}[tb]
  \centering
  \includegraphics[scale=1]{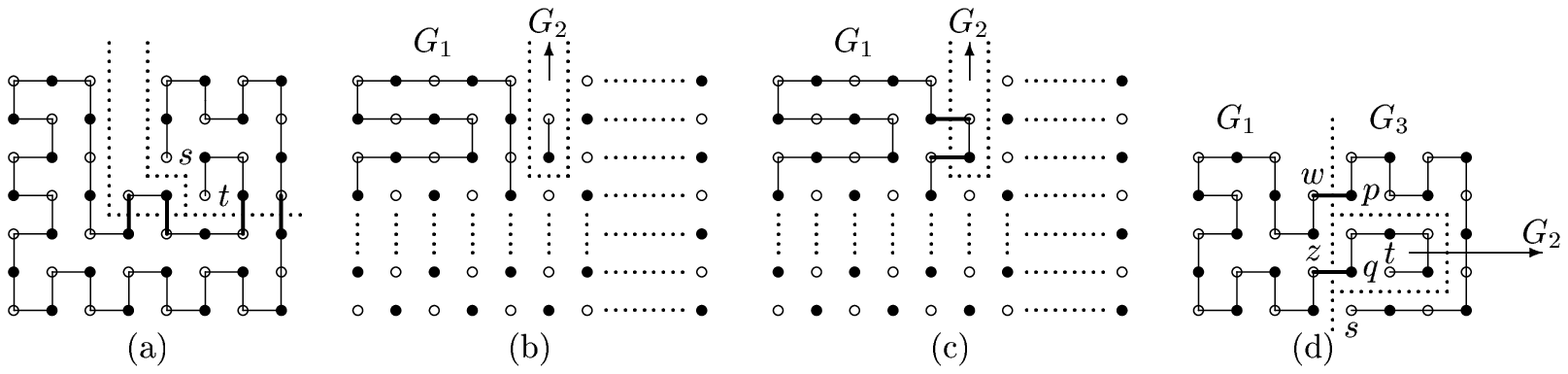}
  \caption[]%
 {\small A Hamiltonian $(s,t)-$path in $C(m,n,k,l)$.}
\label{fig:s13e10}
\end{figure}
\par Subcase 2.5. $l=1$ and $c=even\geq 4$.
\par Subcase 2.5.1. $s_x,t_x\leq d+k$ and $[(n-l=2)$ or $(n-l>2$ and $s\neq (d+k-1,n)$ or $t\neq (d+k,n-1))]$. This case is similar to Subcase 1.2. by the same argument as in proof Subcase 1.2, we obtain $(G_1,s,t)$ is color-compatible and $G_2$ is even-sized. $(G_1,s,t)$ is not in condition (F3)
 and (F5), the proof is similar to Subcase 1.2. The condition (F1) holds, if $n-l\geq 4$, $s=(d+k-1,n)$, and $t=(d+k,n-1)$. This is impossible, because we assume that $s\neq (d+k-1,n)$ or $t\neq (d+k,n-1)$. Therefore, $(G_1,s,t)$ is acceptable. The Hamiltonian path in $(C(m, n, k, l), s, t)$ is
obtained similar to Case 1 of Lemma \ref{Lemma:c11}.
Now, let $n-l\geq 4$, $s=(d+k-1,n)$, and $t=(d+k,n-1)$. This case is similar to Subcase 4.1.2 of Lemma \ref{Lemma:c11}, where $m^{'}=s_x$ and $p=(m^{'},n-2)$. It is easy to see that $(G_1,s,p)$ and $(G_2,q,t)$ are acceptable. The Hamiltonian path in $(C(m, n, k, l), s, t)$ is
obtained similar to Subcase 4.1.2 of Lemma \ref{Lemma:c11}.
\par Subcase 2.5.2. $s_x,t_x>d+k$ and $d>1$. Note that, in this case, $n-l>2$.
\par Subcase 2.5.2.1. $n-l>4$. This case is the same as Case 1 of Lemma \ref{Lemma:c15}, where $l^{'}=l+2$. Since $n^{'}=l^{'}-l=even$, it follows that $G_2$ is even-sized. Moreover, since $C(m,n,k,l)$ is odd-sized, we conclude that $G_1$ is odd-sized. By Lemma \ref{Lemma:c0}, $(G_1,s,t)$ is color-compatible. Since $s_x,t_x>d+k$, $c\geq 4$, and $n-l\geq 4$, it is clear that $(G_1,s,t)$ is not in conditions (F1), (F3), (F10)-(F15), and (F17). Therefore, $(G_1,s,t)$ is acceptable. In this case, $(G_1,s,t)$ is in Subcase 2.1 or 2.3. Now, we show that $(C(m,n,k,l),s,t)$ has a Hamiltonian path. Let $k>1$, then the Hamiltonian path in $(C(m, n, k, l), s, t)$ is
obtained similar to Case 1 of Lemma \ref{Lemma:c11}. Note that since $k=odd\geq 3$, there is at least one edge for combining Hamiltonian cycle and path. Now, Let $k=1$, then the Hamiltonian path in $(C(m, n, k, l), s, t)$ is
obtained similar to Case 2 of Lemma \ref{Lemma:c11}. Notice that, in this case we can always construct a Hamiltonian path $P$ in $G_1$ that contains a subpath $P_1$, as shown Fig. \ref{fig:s13e10}(b). The pattern for constructing a Hamiltonian path in $G_1$ is shown in Fig. \ref{fig:s13e10}(c).
\par Subcase 2.5.2.2. $n-l=4$. In this case, $s_x=d+k+1$ or $t_x=d+k+1$.
\par Subcase 2.5.2.2.1. $s_x,t_x\leq d+k+2$ and $[(s\neq (d+k+1,1)$ or $t\neq (d+k+2,l+1))$ or $(s\neq (d+k+1,n)$ or $t\neq (d+k+2,n-1))]$. This case is the same as Case 3 of Lemma \ref{Lemma:c11}. A simple check shows that $(G_1,s,t)$ is acceptable. In this case $(G_1,s,t)$ is in Subcase 1.1 or 2.2.3. The Hamiltonian path in $(C(m, n, k, l), s, t)$ is
obtained similar to Case 3 of Lemma \ref{Lemma:c11}. Notice that, since $n-l=4$ and $l=1$, we have $n=5$.
\par Subcase 2.5.2.2.1.1. $s= (d+k+1,1)$ and $t= (d+k+2,l+1)$. This case is the same as Subcase 2.4.
\par Subcase 2.5.2.2.1.2. $s=(d+k+1,n)$ and $t= (d+k+2,n-1)$. Let $\{G_1,G_2,G_3\}$  be a $C-$shaped separation (type IV) of $C(m,n,k,l)$ such that $G_1=L(m^{'},n,k,l)$, where $m^{'}=d+k$, $G_2=R_1(m^{''},n^{'})$, where $V(G_2)=\{d+k+1\leq x\leq m-1$ and $l+2\leq y\leq l+3\}$, and $G_3=C(m,n,k,l)\backslash (G_1+G_2)$. Let $s,p\in G_3$, $q,t\in G_2$, $w,z\in G_1$, $w$ and $p$ and $q$ and $z$ are adjacent, $p=(d+k+1,l+1)$, and $q=(d+k+1,n-1)$ (see Fig. \ref{fig:s13e10}(e)). It is known that $(G_3,s,p)$, $(G_2,q,t)$, and $(G_1,w,z)$ are acceptable. In this case, $(G_3,s,p)$ is in Subcase 5.1 of Lemma \ref{Lemma:c11}. The Hamiltonian path in $(C(m, n, k, l), s, t)$ is
obtained similar to Subcase 2.2.3. 
\par Subcase 2.5.2.2.2. $s_x=d+k+1$ and $t_x>d+k+2$. This case is the same as Subcase 2.5.2.2.1., where $s,p\in G_1$, $q,t\in G_2$, $q$ and $p$ are adjacent, and
$$p=
  \begin{cases}
   (d+k+2,n-1);   &  if\ [(m-m^{'}>2)\ or\ (m-m^{'}=2\ and\ t\neq (m,n-1))]\\
    (d+k+2,l+1);   &  if\ m-m^{'}=2,\ t=(m,n-1),\ and \ s\neq (d+k+1,1)\\
   \end{cases}$$
Since $d+k+2=even$ and $l+1=even$ or $n-1=even$, clearly $p$ is white. Thus, $(G_1,s,p)$ and $(G_2,q,t)$ are color-compatible. In this case, $G_2$ is even$\times$odd. Since $n=5$, $(G_2,q,t)$ is not in condition (F2). A simple check shows that $(G_2,q,t)$ is not in condition (F1). Therefore, $(G_2,q,t)$ is acceptable. Now, consider $(G_1,s,p)$. Since $d>1$, $n-l=4$, $c^{'}=2$, and $s_x,p_x>d+k$, it suffices to prove that $(G_1,s,p)$ is not in condition (F1). The condition (F1) holds, if $s=(d+k+1,1)$ and $p=(d+k+2,l+1)$. By the assumption, this is impossible, and hence $(G_1,s,p)$ is not condition (F1). So, $(G_1,s,p)$ is acceptable. In this case, $(G_1,s,p)$ is in Subcase 2.1 or 2.2.3. The Hamiltonian path in $(C(m, n, k, l), s, t)$ is
obtained similar to Case 5 of Lemma \ref{Lemma:c11}. Now let $m-m^{'}=2$, $s=(d+k+1,1)$, and $t=(m,n-1)$. This case is similar to Subcase 2.4.
\par Subcase 2.5.3. $s_x\leq d+k$ and $t_x>d+k$. This case is similar to Subcase 1.4.1, where $p=(d+k,l+1)$ if $s\neq (d+k,l+1)$; otherwise $p=(d+k,l+3)$.
 \end{proof}
\begin{thm} \label{Theorem:6x}
The cases that are mentioned in Lemmas \ref{Lemma:c11}$ -$\ref{Lemma:c13} include all possible cases that may occur in $(C(m,n,k,l),s,t)$.
\end{thm}
\begin{proof}
Consider the following cases. 
\par
Case 1. $C(m,n,k,l)$ is even-sized.
\\
\indent \indent
Subcase 1.1. $(n=even)$ or $(n=odd$ and $[(m=odd)$ or $(m=even$ and $c$ and $d$ are $even)])$. \\
\indent \indent \indent Subcase 1.1.1. $(s_x,t_x\leq d+k)$, $(s_x,t_x>d+k)$, or $(d+1\leq s_x\leq d+k$ and $t_x>d+k)$. \\
\indent \indent \indent \indent Subcase 1.1.1.1. $n=even$. $(C(m,n,k,l),s,t)$ is in Case 1 or 2 of Lemma \ref{Lemma:c11}.\\
\indent \indent \indent \indent Subcase 1.1.1.2. $n=odd$. \\
\indent \indent \indent \indent \indent Subcase 1.1.1.2.1. $m=even$. $(C(m,n,k,l),s,t)$ is in Case 1 of Lemma \ref{Lemma:c11}.
\\
\indent \indent \indent \indent \indent Subcase 1.1.1.2.2. $m=odd$. $(C(m,n,k,l),s,t)$ is in Case 1, 2, or 3 of Lemma \ref{Lemma:c11}.\\
\indent \indent \indent Subcase 1.1.2. $s_x\leq d$ and $t_x>d+k$. $(C(m,n,k,l),s,t)$ is in Subcase 5.1 of Lemma \ref{Lemma:c11}.\\
\indent \indent Subcase 1.2. $n=odd$, $m=even$, and $[c=odd$ or $d=odd]$.
\\ \indent \indent \indent Subcase 1.2.1. $c,d$, and $n-l$ are $odd$.\\
\indent \indent \indent \indent Subcase 1.2.1.1. $s_x,t_x\leq d+1$ or $s_x,t_x>d+1$. $(C(m,n,k,l),s,t)$ is in Subcase 4.1.1 of Lemma \ref{Lemma:c11}.\\
\indent \indent \indent \indent Subcase 1.2.1.2. $s_x\leq d+1$ or $t_x>d+1$. $(C(m,n,k,l),s,t)$ is in Subcase 4.1.2 or 5.2 of Lemma \indent \indent \indent \indent \ref{Lemma:c11}.
\\ \indent \indent \indent Subcase 1.2.2. $n-l=odd$ and $[(d=odd$ and $c=even)$ or $(d=even$ and $c=odd)]$.
\\
\indent \indent \indent \indent Subcase 1.2.2.1. $(s_x,t_x\leq d+k)$, $(s_x,t_x>d+k)$, or $(d+1\leq s_x\leq d+k$ and $t_x>d+k)$. $(C(m,n,k,l),s,t)$ \indent \indent \indent \indent is in Subcase 4.2 of Lemma \ref{Lemma:c11}.\\
\indent \indent \indent \indent Subcase 1.2.2.2. $s_x\leq d$ and $t_x>d+k$. $(C(m,n,k,l),s,t)$ is in Case 5 of Lemma \ref{Lemma:c11}.
\\ \indent \indent \indent Subcase 1.2.3. $c=odd$, $d=odd$, and $n-l=even$.
\\ \indent \indent \indent \indent Subcase 1.2.3.1. $n-l=2$.\\
\indent \indent \indent \indent \indent Subcase 1.2.3.1.1. $(s_x,t_x\leq d+k)$, $(s_x,t_x>d+k)$, or $(d+1\leq s_x\leq d+k$ and $t_x>d+k)$. By \indent \indent \indent \indent \indent Theorem \ref{Theorem:1h}, this case does not occur.\\
\indent \indent \indent \indent \indent Subcase 1.2.3.1.2. $s_x\leq d$ or $t_x>d+k$. $(C(m,n,k,l),s,t)$ is in Subcase 5.2 of Lemma \ref{Lemma:c11}.
\\ \indent \indent \indent \indent Subcase 1.2.3.2. $n-l=4$ or $n-l=6$.\\
\indent \indent \indent \indent \indent Subcase 1.2.3.2.1. $s_x,t_x\leq d$ or $s_x,t_x>d+k$. $(C(m,n,k,l),s,t)$ is in Case 1 of  Lemma \ref{Lemma:c15}.\\
\indent \indent \indent \indent \indent Subcase 1.2.3.2.2. $(s_x\leq d$ and $t_x>d)$, $(d+1\leq s_x\leq d+k$ and $t_x>d+k)$, or $(d+1\leq s_x,t_x\leq d+k)$.\\
\indent \indent \indent \indent \indent \indent Subcase 1.2.3.2.2.1. $n-l=4$. $(C(m,n,k,l),s,t)$ is in Case 2 of Lemma \ref{Lemma:c15}.
\\
\indent \indent \indent \indent \indent \indent Subcase 1.2.3.2.2.2. $n-l=6$. $(C(m,n,k,l),s,t)$ is in Case  3 of Lemma \ref{Lemma:c15}.
\\ \indent \indent \indent \indent Subcase 1.2.3.3. $n-l>6$.
\\ \indent \indent \indent \indent \indent Subcase 1.2.3.3.1. $s_y,t_y>l+5$. $(C(m,n,k,l),s,t)$ is in Subcase 4.1 of Lemma \ref{Lemma:c15}.
\\ \indent \indent \indent \indent \indent Subcase 1.2.3.3.2. $s_y,t_y\leq l+5$. $(C(m,n,k,l),s,t)$ is in Subcase 4.2 of Lemma \ref{Lemma:c15}.
\\ \indent \indent \indent \indent \indent Subcase 1.2.3.3.3. $s_y\leq l+5$ and $t_y>l+5$ (or $t_y\leq l+5$ and $s_y>l+5)$. $(C(m,n,k,l),s,t)$ is in  \indent \indent \indent \indent \indent  Subcase 4.3 of Lemma \ref{Lemma:c15}.\par Case 2. $C(m,n,k,l)$ is odd-sized. \\
\indent \indent Subcase 2.1. $n=even$. \\
\indent \indent \indent Subcase 2.1.1. $(s_x,t_x\leq d+k)$, $(s_x,t_x>d+k)$, or $(d+1\leq s_x\leq d+k$ and $t_x>d+k)$. $(C(m,n,k,l),s,t)$ is \indent \indent \indent in Subcase 1.2 of Lemma \ref{Lemma:c13}.\\
\indent \indent \indent Subcase 2.1.2. $s_x\leq d$ and $t_x>d+k$. $(C(m,n,k,l),s,t)$ is in Subcase 1.4 of Lemma \ref{Lemma:c13}.\\
\indent \indent Subcase 2.2. $n=odd$. \\
\indent \indent \indent Subcase 2.2.1. $m=odd$. \\
\indent \indent \indent \indent Subcase 2.2.1.1. $l=even$.
 \\
 \indent \indent \indent \indent \indent Subcase 2.2.1.1.1. $(s_x,t_x\leq d+k)$, $(s_x,t_x>d+k)$, or $(d+1\leq s_x\leq d+k$ and $t_x>d+k)$. \indent \indent \indent \indent \indent $(C(m,n,k,l),s,t)$ is in Subcase 1.1 of Lemma \ref{Lemma:c13}.\\
 \indent \indent \indent \indent \indent Subcase 2.2.1.1.2. $s_x\leq d$ and $t_x>d+k$. $(C(m,n,k,l),s,t)$ is in Subcase 1.4 of Lemma \ref{Lemma:c13}.\\
 \indent \indent \indent \indent Subcase 2.2.1.2. $l=odd$. Notice that, in this case, $d=odd$ and $c=even$. By symmetry, the case \indent \indent \indent \indent $d=even$ and $c=odd$ has been removed.
 \\
  \indent \indent \indent \indent \indent Subcase 2.2.1.2.1. $s_x,t_x\leq d+k$. $(C(m,n,k,l),s,t)$ is in Subcase 1.2 of Lemma \ref{Lemma:c13}.\\
   \indent \indent \indent \indent \indent Subcase 2.2.1.2.2. $s_x,t_x>d+k$. $(C(m,n,k,l),s,t)$ is in Subcase 1.3 of Lemma \ref{Lemma:c13}.\\
    \indent \indent \indent \indent \indent Subcase 2.2.1.2.3. $s_x\leq d+k$ and $t_x>d+k$. $(C(m,n,k,l),s,t)$ is in Subcase 1.4 of Lemma \ref{Lemma:c13}.\\
 \indent \indent \indent Subcase 2.2.2. $m=even$. Notice that, in this case, $d=odd$ and $c=even$. By symmetry, the case $d=even$ \indent \indent \indent and $c=odd$ has been removed.\\
  \indent \indent \indent \indent Subcase 2.2.2.1. $(l>1$ and $c>2)$ or $(c\geq 2)$. \\
 \indent \indent \indent \indent \indent Subcase 2.2.2.1.1. $(s_x,t_x\leq d+k)$, $(s_x,t_x>d+k$ and $s_y,t_y>l)$ or $(s_x\leq d+k$, $t_x>d+k$, and $t_y>l)$. \indent \indent \indent \indent \indent $(C(m,n,k,l),s,t)$ is in Subcase 2.1 or 2.2 of Lemma \ref{Lemma:c13}.\\
 \indent \indent \indent \indent \indent Subcase 2.2.2.1.2. $s_x,t_x>d+k$ and $s_y,t_y\leq l+1$. $(C(m,n,k,l),s,t)$ is in Subcase 2.3 of Lemma \indent \indent \indent \indent \indent  \ref{Lemma:c13}.  \\
 \indent \indent \indent \indent \indent Subcase 2.2.2.1.3. $(t_x>d+k$, $t_y\leq l$, and $[(s_x\leq d+k)$ or $(s_x>d+k$ and $s_y>l+1)])$ or \indent \indent \indent \indent \indent  $(s_x,t_x>d+k$, $s_y\leq l$, and $t_y>l)$. $(C(m,n,k,l),s,t)$ is in Subcase 2.4 of Lemma \ref{Lemma:c13}. \\
 \indent \indent \indent \indent Subcase 2.2.2.2. $l=1$ and $c>2$. \\
  \indent \indent \indent \indent \indent Subcase 2.2.2.2.1. $s_x,t_x\leq d+k$. $(C(m,n,k,l),s,t)$ is in Subcase 2.5.1 of Lemma \ref{Lemma:c13}. \\
    \indent \indent \indent \indent \indent Subcase 2.2.2.2.2. $s_x,t_x>d+k$. $(C(m,n,k,l),s,t)$ is in Subcase 2.5.2 of Lemma \ref{Lemma:c13}.\\
      \indent \indent \indent \indent \indent Subcase 2.2.2.2.3. $s_x\leq d+k$ and $t_x>d+$.$(C(m,n,k,l),s,t)$ is in Subcase 2.5.3 of Lemma \ref{Lemma:c13}.
\par
 All possible cases are exhausted, and the proof of Theorem \ref{Theorem:6x} is completed.
\end{proof}
\par
By Theorem \ref{Theorem:1h} and Lemmas \ref{Lemma:c11}$-$ \ref{Lemma:c13}, we have the following result:
\begin{thm}
$C(m,n,k,l)$ has a Hamiltonian $(s,t)-$path if and only if $(C(m,n,k,l),s,t)$ is acceptable.
\end{thm}
\par In the following theorem, we state the main result of this paper:
\begin{thm}\label{Lemma:6tr}
In an acceptable $(L(m,n,k,l),s,t)$, a Hamiltonian $(s,t)-$path can be found in linear time.
\end{thm}
\begin{proof}
 The algorithm construct a Hamiltonian $(s,t)-$path in $C(m,n,k,l)$ via the following three steps.
 \par $Step \ 1:$ It divides $C(m,n,k,l)$ into some grid subgraphs, by Lemmas \ref{Lemma:c11}$-$\ref{Lemma:c13}, in $O(1)$ time.\par $Step\ 2:$ It finds a Hamiltonian path or cycle in these grid subgraphs by algorithm \cite{CST:AFAFCHPIM} or \cite{991}. This  step takes \\ \indent \indent \indent linear time.\par $Step\ 3:$ It combines Hamiltonian paths and cycles for constructing a Hamiltonian $(s,t)-$path, by Lemmas \ref{Lemma:c11}$-$\ref{Lemma:c13}, \\ \indent \indent \indent in $O(1)$ time. \par Thus, the algorithm has a linear-time complexity.
\end{proof}
\section{Conclusion} \label{ConclusionSect}
We gave necessary and sufficient conditions for the existence of a Hamiltonian
path in $C-$shaped grid graphs between two given vertices, which are a special type of solid grid graphs.
The Hamiltonian path problem is NP-complete in general grid graphs \cite{IPS:HPIGG}, it
remains open if the problem is polynomially solvable in solid grid graphs. Further study can be done on the Hamiltonian path problem in other
special classes of graphs, in order to find way to solve the problem for solid grid graphs.









\end{document}